\documentclass[12pt,letterpaper,oneside]{book}

\usepackage{setspace,dsfont,verbatim,paralist,indentfirst}
\usepackage{mathtools}
\usepackage{graphicx}
\graphicspath{{figures/}}
\usepackage{amssymb}
\usepackage{comment}
\usepackage{mathtools}
\usepackage{bm}
\usepackage{bbold}
\usepackage{amsmath}
\usepackage{tikz}
\usepackage{titlesec}
\titleformat{\section}
  {\normalfont\fontsize{12}{15}\bfseries}{\thesection}{1em}{}
\titleformat{\subsection}
  {\normalfont\fontsize{12}{15}\bfseries}{\thesubsection}{1em}{}

\newlength{\bibitemsep}
\newlength{\bibparskip}
\let\oldthebibliography\thebibliography
\renewcommand\thebibliography[1]{%
  \oldthebibliography{#1}%
  \setlength{\parskip}{\bibitemsep}%
  \setlength{\itemsep}{\bibparskip}%
}
\usepackage{epstopdf}
\usepackage{physics}
\usepackage{subcaption}
\usepackage{cite}

\def\B{\mathcal{B}}

\def\D{\mathcal{D}}

\def\H{\mathcal{H}}

\def\I{\mathcal{I}}

\def\N{\mathcal{N}}

\def\T{\mathcal{T}}

\def\W{\mathcal{W}}

\usepackage{accents}

\newtheorem{theorem}{Theorem}

\newtheorem{corollary}[theorem]{Corollary}

\newtheorem{definition}[theorem]{Definition}

\newtheorem{lemma}{Lemma}

\newtheorem{proposition}[theorem]{Proposition}
\newtheorem{remark}[theorem]{Remark}

\newenvironment{proof}[1][Proof]{\noindent\textbf{#1.} }{\ \rule{0.5em}{0.5em}}

\makeatletter
\usepackage[pdflang={en-US},pdftex]{hyperref}
\hypersetup{%
  pdftitle={My Thesis},
  pdfauthor={Robert Hynes},
  pdfkeywords={Physics,Astronomy},
  pdfdisplaydoctitle=true,   
  bookmarksnumbered,
  pdfstartview={FitH},
  colorlinks,
  citecolor=black,
  filecolor=black,
  linkcolor=black,
  urlcolor=black,
  breaklinks=true,
  hypertexnames=false
}
\usepackage{cleveref}
\newcommand{\bigsize}{\fontsize{16pt}{20pt}\selectfont}

\numberwithin{equation}{section}

\begin{document}

\renewcommand\@pnumwidth{1.55em}
\renewcommand\@tocrmarg{9.55em}
\renewcommand*\l@chapter{\@dottedtocline{0}{1.5em}{2.3em}}
\renewcommand*\l@figure{\@dottedtocline{1}{0em}{3.1em}}
\let\l@table\l@figure

\pagenumbering{roman}
\thispagestyle{empty}
\begin{center}
{\bigsize LIMITS ON PARAMETER ESTIMATION OF QUANTUM CHANNELS}


\vfill
\doublespacing
A Dissertation \\
\singlespacing
Submitted to the Graduate Faculty of the \\
Louisiana State University and \\
Agricultural and Mechanical College \\
in partial fulfillment of the \\
requirements for the degree of \\
Doctor of Philosophy \\
\doublespacing
in \\
                                       
The Department of Physics and Astronomy\\
\singlespacing
\vfill

by \\
Vishal Katariya\\
B. Tech., Indian Institute of Technology Madras, 2017  \\

May 2022
\end{center}
\pagebreak

\chapter*{\Large Acknowledgments}
\vspace{0.55ex}

\addcontentsline{toc}{chapter}{\hspace{-1.5em} {ACKNOWLEDGMENTS} \vspace{12pt}}

As they say, it takes a village. This dissertation, and my PhD, has been made possible by the support and kindness of a number of people who have been there for me.

I'd like to start by thanking my parents, who have supported me unconditionally regardless of what I've wanted to do. A few years of living away from home has taught me how much I rely on them and their steadfast presence in my life.

I'm very thankful for my advisor, Mark M. Wilde, for being a source of support through the entirety of my grad school experience. He has been a mentor, collaborator, and friend. He has been a major influence on the way I have conducted research, and there are many things I've learnt from him and subsequently incorporated into my own work. I'm grateful to him for introducing me to almost all the aspects of being an academic, whether it be learning to write research papers effectively, attending conferences and workshops, knowing how to handle the tricky process of peer review, or becoming an effective reviewer myself. It has been great to live the life of an academic, so to speak, and it is Mark who has been most instrumental in making that happen for me.

I would like to thank Ravi P. Rau, Ivan Agullo, and Oliver Dasbach for agreeing to be on my PhD committee. I have had a number of valuable discussions with Dr. Rau over the course of my time at LSU, and in particular I would like to thank him for organizing his weekly lunch seminars for grad students. At the beginning of my time at LSU, those seminars helped me feel part of the department at large and more connected with other grad students. I'd also like to thank Ivan Agullo and Hwang Lee for a number of interesting discussions over the years. I'd like to thank Claire R. Bullock, Carol Duran, Mimi LaValle, Arnell Nelson, Paige Whittington, Yao Zeng, and other members of the Department of Physics and Astronomy at LSU for their help during my PhD.

I am grateful for the late Jonathan P. Dowling. Jon was part of my committee for my general exam and tragically passed away in 2020. Jon was special to every student in the group, even if he wasn't their formal advisor. He is one of the few people I know who made the question, "how are you?" feel pleasant and even comforting. Whether it was talking about research, mental health struggles, or regaling each other with wild stories, he was a bulwark during my time at LSU. I miss you, Jon.

I have been fortunate to have had excellent collaborators to work with. Working with them on my research projects taught me that the value of a team is more than the sum of its parts. They have enriched my life and been a huge part of my research. I'd like to thank Patrick J. Coles, Iman Marvian, Seth Lloyd, Nilanjana Datta, Eric P. Hanson, Narayan Bhusal, Chenglong You, Anthony Brady and Stav Haldar. 

Conversations in our QST office have always been pleasant. This was made more apparent during the pandemic, as I started to miss the conviviality of sharing an office with other grad students as each of us were tussling with and enjoying our various projects and problems. I'd like to thank my officemates for their congeniality and comradeship. Grad school can be a lonely prospect, but it is made much better with a support system. I'd like to thank Pratik Barge, Aby Philip, Arshag Danageozian, Stav Haldar, Sumeet Khatri, Vishal Singh, Soorya Rethinasamy, Siddhartha Das, Anthony Brady, Aliza Siddiqui, Kevin V. Jacob, Margarite LaBorde, Kunal Sharma, and Eneet Kaur for our shared time together in the Quantum Science and Technologies group. Thanks also to Yihui Quek, whose visit to LSU made her feel like an honorary member of our group.

A number of people have made Baton Rouge feel like home. More than anything else, it is them whom I will miss most. Valerie Flynn took me under her wing in my first year here, and our shared library, thrift store and grocery trips played a huge part in making me feel welcome and comfortable. She has been a mentor, friend, source of support, and I will dearly miss our stimulating conversations. Thank you, Valerie, for being kind and a pillar of support. I am grateful, too, for Karunya's friendship. Whether it be shared cups of tea or bars of chocolate, joint cooking sessions, or conversations about life: I will miss it all. I met Elizabeth Courville through potlucks that she organized at her home. That quickly blossomed into a lovely and supportive friendship. Thank you, Liz, for the potlucks, for your time, and our shared meals. You've made Baton Rouge feel more like home. My friendship with Emily Goldsmith started off with a sourdough starter, and blossomed into us becoming close friends supporting each other through grad school. Stav Haldar is one of the kindest and most generous people I know. He has been around for me when I've needed it, and I've learnt a lot about large-heartedness from you. Thank you, Stav. I'd like to thank Tyler Ellis for inviting me to a number of social gatherings, for periodically checking in on me, sending me cheery memes and making me feel cared for. Prerna Agarwal and I have recently become fast friends, and I appreciate your company and support dearly. I'd also like to thank Siddharth Soni, Akhil Bhardwaj, Anshumitra Baul, and Sonia Blauvelt for their friendship.

Sometimes it feels like I live half of my life online. It's funny to think how supported, loved, and cared for I feel via relationships mediated by digital screens. Learning how to navigate, conduct, and make the most of internet friendships has been one of the best things I've done in the past few years. Thank you to Nithin Ramesan for being around, every single day. I am grateful to my close-knit circle of friends for their support and our conversations that span the entirety of the human experience. I am thankful for Kriti Dhanania, Upasana Bhattacharjee, Annapoorani Hariharan, Anjani Balu, Aksheeya Suresh, Anna Thomas, Kaushik Satapathy, Pranavi AR, Sanjana Srikant, Bhavika Bhatia, Renate Boronowsky, and Mohna Priyanka for their support and friendship. I am thankful to Subramanian Balakrishna, Rishi Rajasekaran, Shashwat Salgaocar, Abhinand Sukumar and Ananth Sundararaman for our online board games sessions that played a large part in maintaining a sense of calm and camaraderie during the early stages of pandemic.

\singlespacing
\tableofcontents
\pagebreak

\pagebreak
\addcontentsline{toc}{chapter}{\hspace{-1.5em} LIST OF FIGURES \vspace{12pt}}
\listoffigures
\pagebreak

\chapter*{\Large Abstract}

The aim of this thesis is to develop a theoretical framework to study parameter estimation of quantum channels. We begin by describing the classical task of parameter estimation that we build upon. In its most basic form, parameter estimation is the task of obtaining an estimate of an unknown parameter from some experimental data. This experimental data can be seen as a number of samples of a parameterized probability distribution. In general, the goal of such a task is to obtain an estimate of the unknown parameter while minimizing its error.

We study the task of estimating unknown parameters which are encoded in a quantum channel. A quantum channel is a map that describes the evolution of the state of a quantum system. We study this task in the sequential setting. This means that the channel in question is used multiple times, and each channel use happens subsequent to the previous one. A sequential strategy is the most general way to use, or process, a channel multiple times. Our goal is to establish lower bounds on the estimation error in such a task. These bounds are called Cramer--Rao bounds. Quantum channels encompass all possible dynamics allowed by quantum mechanics, and sequential estimation strategies capture the most general way to process multiple uses of a channel. Therefore, the bounds we develop are universally applicable.

We consider the use of catalysts to enhance the power of a channel estimation strategy. This is termed amortization. The reason we do so is to investigate if an $n$-round sequential estimation strategy does better than a simpler parallel strategy. Quantitatively, the power of a channel for a particular estimation task is determined by the channel's Fisher information. Thus, we study how much a catalyst quantum state can enhance the Fisher information of a quantum channel by defining the \textit{amortized Fisher information}. In the quantum setting, there are many Fisher information quantities that can be defined. We limit our study to two particular ones: the symmetric logarithmic derivative (SLD) Fisher information and the right logarithmic derivative (RLD) Fisher information.

We establish our Cramer--Rao bounds by proving that for certain Fisher information quantities, catalyst states do not improve the performance of a sequential estimation protocol. The technical term for this is an \textit{amortization collapse}. We show how such a collapse leads directly to a corresponding Cramer--Rao bound. We establish bounds both when estimating a single parameter and when estimating multiple parameters simultaneously. For the single parameter case, we establish Cramer--Rao bounds for general quantum channels using both the SLD and RLD Fisher information. The task of estimating multiple parameters simultaneously is more involved than the single parameter case. In the multiparameter case, Cramer--Rao bounds take the form of matrix inequalities. We provide a method to obtain scalar Cramer--Rao bounds from the corresponding matrix inequalities. We then establish a scalar Cramer--Rao bound using the RLD Fisher information. Our bounds apply universally and we also show how they are efficiently computable by casting them as optimization problems. 

In the single parameter case, we recover the so-called ``Heisenberg scaling'' using our SLD-based bound. On the other hand, we provide a no-go condition for Heisenberg scaling using our RLD-based bound for both the single and multiparameter settings. Finally, we apply our bounds to the example of estimating the parameters of a generalized amplitude damping channel.

\pagebreak


\pagenumbering{arabic}
\addtocontents{toc}{ \hspace{-1.8em} CHAPTERS}
\singlespacing
\setlength{\textfloatsep}{12pt plus 2pt minus 2pt}
\setlength{\intextsep}{6pt plus 2pt minus 2pt}


\singlespacing
\chapter{Introduction}\label{ch:intro}
\vspace{0.5em}

The invention of quantum mechanics in the early 20th century revolutionized physics and enabled accurate explanations of physical phenomena at atomic and sub-atomic scales. It allowed for further theoretical and experimental exploration of the physical world, which has now resulted in an improved understanding of physics as well as a number of technological innovations that improve our lives today. Some of the initial concepts, neologisms, and thought experiments from the early days of quantum mechanics remain in our vocabulary even now. For example, ``spooky action at a distance'', coined by Einstein, Podolsky and Rosen, is still in use today when referring to the counterintuitive phenomenon of quantum entanglement. The ``Schrodinger's Cat" thought experiment too, lives on both in science and popular culture as a way to understand the concept of quantum superposition and the probabilistic nature of quantum measurements.

Along with explaining hitherto unexplained phenomena, quantum mechanics also opened up new methods of encoding, transmitting and protecting information. The subject of quantum information was developed in the decades following the invention of quantum mechanics itself. It is an interdisciplinary field combining the physics of quantum mechanics with the mathematical and statistical machinery of classical information theory, which itself was single-handedly invented by Shannon in 1948 \cite{Shannon1948}. 

Quantum information deals with scenarios where one handles information that is encoded in quantum systems. Some examples of quantum systems that can hold or transmit quantum information are electronic or ionic spins, the polarization of a photon, two-level atoms, superconducting transmon qubits \footnote{These superconducting transmon qubits are used by IBM in their quantum computers.}, quantum dots, and nitrogen-vacancy centers. Some of the canonical tasks studied in quantum information are compression of data \cite{Schumacher1995}, transmitting classical or quantum information over quantum communication channels \cite{Bennett1992, Bennett1993, Holevo1998,Schumacher1997}, establishing secret key (shared, private, true random bits) between two parties using quantum key distribution \cite{Bennett1984}, among others. In many of these tasks, using quantum resources like entanglement or states in superposition yields an advantage over what is possible using purely classical resources. Another task which allows for such a quantum advantage, so to speak, is parameter estimation. 

Parameter estimation, also known as metrology, is a fundamental primitive in all of science and technology. It refers to the statistical task of accurately estimating an unknown parameter of interest encoded in data collected from some experimental procedure. This is a basic task in statistics, and thus has a rich underlying theory that starts from the work of Fisher \cite{Fisher1925}.

The setup of parameter estimation in the classical setting begins with a probability distribution $p_\theta (x)$, where $\theta$ is the unknown parameter of interest and is encoded in samples of random variable $X$. Here, we follow the usual notation that uppercase letters are used for random variables and the corresponding lowercase letters are used for particular realizations of the random variable in question. Given a value of $\theta$, the random variable $X$ is distributed according to the distribution $p_\theta$. The goal, then, is to obtain $\hat{\theta}(X)$, an estimate of the parameter $\theta$, from a certain number $n$ of samples of $X$. Intuitively, with an infinite number of samples of $X$, we expect to be able to infer the value of $\theta$ perfectly. Therefore, we would like to quantify how well we can do with finite $n$.

We go into more details of the technicalities of parameter estimation later on in the thesis, but here we state simply that the specific quantitative goal is to minimize the variance of the estimator $\hat{\theta}(X)$. We also assume throughout this thesis that the estimator is ``unbiased''; i.e., it is accurate and converges to the correct value of $\theta$ on average. Then the fundamental tool used to find limits on the attainable variance (or precision) of an estimator is the Cramer--Rao bound \cite{Rao1945, Cramer1946, Kay1993}. In the classical setting, it takes on the following form:
\begin{equation}
    \text{Var}(\hat{\theta}(X)) \geq \frac{1}{n I_{F}(\theta; \{p_{\theta}\}_\theta)},
\end{equation}
where $I_{F}(\theta; \{p_{\theta}\}_\theta)$ is the Fisher information of the family of distributions $\{ p_\theta \}_\theta$. Again, we do not go into details here, and note that this bound and the Fisher information are explained more fully in Chapter~\ref{ch:prelims}. We state this fundamental inequality here to bring the Fisher information into focus. The Fisher information, due to its appearance in the Cramer--Rao bound, takes on a fundamental operational meaning in classical estimation theory.

The classical Cramer--Rao bound above is applicable for the case when there is a single unknown parameter to be estimated. However, there are experiments and scenarios when one may wish to estimate multiple parameters simultaneously. As one may expect, estimating multiple parameters simultaneously is a more mathematically and technologically involved task than estimating a single parameter. However, a theory of multiparameter estimation exists, and Cramer--Rao bounds for such tasks can be constructed.

One of the major differences when establishing Cramer--Rao bounds for multiparameter estimation, both in the classical and quantum cases, is to generalize the figure of merit (mean-squared error) and the Fisher information from scalars to matrices. That is, if there are $D > 0$ unknown parameters to be estimated, then the figure of merit and Fisher information both take on the form of a $D \times D$ matrix. The Cramer--Rao bound, too, then is generalized to a matrix inequality. We state this qualitatively now, and these notions will be expanded upon in detail in the subsequent chapters of this thesis. 

As we stated briefly earlier, parameter estimation can allow for a quantum advantage in certain cases, such as in phase estimation using optical interferometry \cite{Braunstein1992, Dowling1998, DemkowiczDobrzanski2015} and gravitational wave detection \cite{Caves1981, Yurke1986, Berry2000, DemkowiczDobrzanski2013}. We will now elaborate further on this, and introduce quantum metrology in the process. First, we recall Heisenberg's uncertainty principle, a fundamental property in quantum mechanics which demarcates it cleanly from the classical world. It imposes a fundamental limit on the accurary with which two non-commuting physical observables can be estimated. For example, the more precisely the position of a particle can be determined, the less precisely is our knowledge of its momentum. Such a concept has no classical analog and immediately suggests that quantum metrology is markedly different from classcial metrology. \cite{Helstrom1976} developed a theory of quantum detection and estimation.

In the classical Cramer--Rao bound stated above, the RHS scales with $n$ (the number of samples of data) as $1/n$. The scaling of the estimator variance as $1/n$ is known as the shot-noise limit.
It is the fundamental limit to the precision achievable by any estimation protocol using only classical resources. However, in the case of quantum metrology, there exist some tasks for which the estimator variance can be made to scale with $n$ as $1/n^2$ instead of $1/n$. This yields a lower estimator error than the classical case, especially in the case of large $n$. 

The scaling of estimator variance as $1/n^2$ is denoted as the Heisenberg limit (or sometimes as Heisenberg scaling), as the Cramer--Rao bound in such cases takes on a form reminiscent of the Heisenberg uncertainty principle. The Heisenberg limit is the best possible scaling of error that one can attain using quantum resources, and identifying estimation strategies and tasks for which it is attainable is one of the goals of quantum metrology. 

The objects of interest in classical estimation are probability distributions $p_\theta$ in which the unknown parameter $\theta$ is encoded. In the case of quantum parameter estimation, though, we are interested in quantum states and channels. The state of a quantum system is not completely described by a probability distribution. Pure quantum states are represented by complex-valued vectors, and the more general mixed quantum states are operators (they can be represented as positive semidefinite matrices with unit trace). Quantum channels are dynamical maps that take quantum states to quantum states. Objects in quantum theory do not generally commute, and quantum states can be in a linear superposition of certain fixed basis states, both of which result in quantum metrology having a more mathematically involved theory than classical metrology.

The noncommutativity of quantum mechanics results in an infinite number of quantum generalizations of the classical Fisher information. The two best-studied quantum Fisher informations are the symmetric logarithmic derivative (SLD) Fisher information \cite{Helstrom1976} and the right logarithmic derivative (RLD) Fisher information \cite{Yuen1973}. These quantities can be defined for both quantum state and quantum channel families, both of which we use extensively in this thesis. Further, each of them can be used to yield Cramer--Rao bounds for estimation of both quantum states and channels.

Our goal in this thesis is to study fundamental limits to estimating one or more parameters encoded in an unknown quantum channel. This is a well-studied problem with literature stretching back to the 1970s \cite{Helstrom1976, Yuen1973} and a number of other prior works \cite{Sasaki2002, Fujiwara2003, Fujiwara2004, Ji2008, Fujiwara2008, Matsumoto2010, Hayashi2011,DemkowiczDobrzanski2012, Kolodynski2013, DemkowiczDobrzanski2014, Sekatski2017, DemkowiczDobrzanski2017, Zhou2018, Zhou2019, Zhou2019a,Yang2020b}. The most general setting for the channel estimation problem is the sequential setting, where the unknown channel is processed $n$ successive times while allowing for adaptive control operations between channel uses \cite{Giovannetti2006, Dam2007, DemkowiczDobrzanski2014, Yuan2017}.

A subset of sequential (also known as adaptive) strategies is the set of parallel strategies, where the $n$ uses of the channel happen simultaneously. This is a practical setting of interest. Since parallel strategies are a subset of sequential ones, by design they are less powerful. For some special cases, e.g., unitary channels \cite{Giovannetti2006}, parallel strategies are just as powerful as sequential ones. It is an important and fruitful line of inquiry in quantum information to identify when sequential strategies offer an advantage over parallel ones and when they do not. This question remains a topic of interest, and has been studied recently in the context of various channel distinguishability tasks \cite{Berta2018,Fang2020,Katariya2021_adaptive,Salek2021,Bavaresco2020}.

In this thesis, we also study this particular problem in the context of quantum channel estimation. Our goal, both in the single and multiparameter cases, is to establish Cramer--Rao bounds that apply for estimating parameters encoded in a quantum channel in the sequential setting. We do so by defining certain quantities for the quantum Fisher information inspired by recent developments in discrimination and distinguishability of quantum channels and processes.

The quantities we define are the \textit{generalized Fisher information} of quantum states and channels, inspired by generalized channel divergences introduced in \cite{Polyanskiy2010,Sharma2012}, and the \textit{amortized Fisher information}, which in turn is inspired by the amortized channel divergence introduced in \cite{Berta2018}. These quantities taken together allow for us to study and apply the SLD and RLD Fisher information to the task of sequential channel estimation. Our approaches to establishing Cramer--Rao bounds for the sequential setting for both the single and multiparameter estimation cases are similar.

The main ingredient in our Cramer--Rao bounds are what are known as \textit{amortization collapses}. An amortization collapse is when, for a certain Fisher information in question, the amortized Fisher information is equal to the Fisher information itself. It further means that catalysis cannot help to increase the Fisher information of a channel family to a value more than its inherent value. These facts mean that for quantities that undergo an amortization collapse, sequential strategies offer no advantage over parallel ones.

Finally, we connect the amortized Fisher information to the Fisher information achievable by a sequential estimation strategy. This is done by proving meta-converse theorems for both the single and multiparameter estimation tasks, inspired by the meta-converse theorem of \cite{Berta2018}. With this in place, we prove various Cramer--Rao bounds for single parameter estimation, and an RLD-based one for the case of multiparameter channel estimation. This builds on prior work in establishing Cramer--Rao bounds for channel estimation, both in the parallel setting and the sequential one \cite{Hayashi2011, Yuan2017, Zhou2018, Zhou2020}.

Our bounds have a number of desirable properties, which we state briefly here. Our bounds in Chapter~\ref{ch:single} for single parameter estimation are single-letter, a fact that arises due to the amortization collapses we prove. ``Single-letter'' is a term from information theory, which means that the Fisher information in question is evaluated with respect to a single copy of the channel only even though the bound holds for general $n$-round sequential strategies. This makes them straightforward to evaluate, and further we provide various optimization problem characterizations for the SLD and RLD Fisher information of states and channels in Chapter~\ref{ch:single}.

 For single parameter estimation, our SLD-based Cramer--Rao bound for channel estimation recovers the fact that Heisenberg scaling is the best possible scaling attainable for channel estimation, even in sequential estimation protocols. This builds on work of \cite{Yuan2017}. Further, our RLD-based bound for single parameter estimation leads to the important corollary that, when the RLD Fisher information of a particular channel family is finite, then Heisenberg scaling (with respect to the number of channel uses) in error for estimating the channel family is unattainable.

In Chapter~\ref{ch:multi}, we use the RLD Fisher information of states and channels to establish Cramer--Rao bounds for the case of simultaneously estimating multiple parameters, in the vein of and continuing the results of Chaper~\ref{ch:single}. Our goal is to establish scalar Cramer--Rao bounds for multiparameter estimation of quantum channels. As we stated earlier in this chapter, Cramer--Rao bounds for multiparameter estimation take the form of matrix inequalities and the Fisher information too is a matrix.

Therefore, our first step is to define a scalar quantity, the RLD Fisher information value, for state and channel families. We show how the RLD Fisher information value of states can be used to establish a scalar Cramer--Rao bound, as we desired. We then follow a similar approach as in Chapter~\ref{ch:single}; i.e., for multiparameter channel estimation, we show an amortization collapse for the RLD Fisher information value of quantum channels.

With the amortization collapse for the RLD Fisher information value of channels in place, we are able to also establish a scalar Cramer--Rao bound for multiparameter estimation of quantum channels in the sequential setting. This bound, like the RLD-based one of Chapter~\ref{ch:single}, is also
\begin{itemize}
	\item single-letter; i.e., computing it requires computing the RLD Fisher information value of a single channel use even though the bound is applicable for $n$-round sequential procotols,
	\item universally applicable, in the sense that our bound applies to all quantum channels, and thus encompasses all admissible quantum dynamics, and
	\item efficiently computable via a semi-definite program. 
\end{itemize}

The necessary background for this thesis is familiarity with the basics of quantum mechanics, quantum information theory, and estimation theory. We point readers to the books~\cite{Wilde2017, Watrous2018, Khatri2020, Hayashi2006, Holevo2011} on quantum information theory and to the recent reviews \cite{Sidhu2020, Szczykulska2016, Albarelli2020a} on quantum estimation theory in the single and multiparameter settings.

This thesis is based on the following papers:
\begin{itemize}
	\item \textbf{Geometric distinguishability measures limit quantum channel estimation and discrimination} \cite{Katariya2021} \\
	Vishal Katariya and Mark M. Wilde \\
	Quantum Information Processing \textbf{20}, 78 (2021), arXiv:2004.10708 \\
	Chapter~\ref{ch:single}
	\item \textbf{RLD Fisher information bound for multiparameter estimation of quantum channels} \cite{Katariya2021a} \\
	Vishal Katariya and Mark M. Wilde \\
	New Journal of Physics \textbf{23}, 073040 (2021), arXiv:2008.11178 \\
	Chapter~\ref{ch:multi}
\end{itemize}

Other papers to which the author contributed during his Ph.D.:
\begin{itemize}
	\item \textbf{Entropic energy-time uncertainty relation} \cite{Coles2019} \\
	Patrick J. Coles, Vishal Katariya, Seth Lloyd, Iman Marvian, and Mark M. Wilde \\
	Physical Review Letters \textbf{122}, 100401 (2019), arXiv:1805.07772 
	\item \textbf{Evaluating the Advantage of Adaptive Strategies for Quantum Channel Distinguishability} \cite{Katariya2021_adaptive} \\
	Vishal Katariya and Mark M. Wilde \\
	Physical Review A \textbf{104}, 052406 (2021), arXiv:2001.05376
	\item \textbf{Guesswork with quantum side information} \cite{Hanson2021} \\
	Eric P. Hanson, Vishal Katariya, Nilanjana Datta and Mark M. Wilde \\
	\textit{to appear: IEEE Transactions on Information Theory}, arXiv:2001.03958, and
	\item \textbf{Quantum State Discrimination Circuits Inspired by Deutschian Closed Timelike Curves} \cite{Vairogs2021} \\
	Christopher Vairogs, Vishal Katariya and Mark M. Wilde \\
	arXiv:2109.11549.
\end{itemize}
\pagebreak

\chapter{Preliminaries}\label{ch:prelims}
\allowdisplaybreaks

\vspace{0.5em}

\section{Classical parameter estimation} 

The first step towards studying parameter estimation using quantum resources is to understand parameter estimation in the classical setting; i.e., estimating one or more unknown parameters encoded in a probability distribution. We begin by discussing the task we introduced briefly in Chapter~\ref{ch:intro}, that of estimating a single parameter encoded in a parameterized probability distribution. This is the framework that we will build on later, both to study estimation tasks involving quantum states and channels, as well as to simultaneously estimate multiple parameters.

The fundamental task in classical estimation is to estimate a parameter $\theta$ encoded in a probability distribution $p_\theta (x)$ with associated random variable $X$. Each probability distribution belongs to a parameterized family $\{p_{\theta}(x)\}_{\theta}$\ of probability distributions. Each distribution is a function of the unknown parameter $\theta\in\Theta\subseteq\mathbb{R} $, and the goal is to produce an estimate $\hat{\theta}(X)$ of $\theta$ from $n$ independent samples of the distribution $p_{\theta}(x)$. Suppose that the family $\{p_{\theta}(x)\}_{\theta}$ is differentiable with respect to the parameter $\theta$, so that $\partial_{\theta}p_{\theta}(x)$ exists for all values of $\theta$ and $x$, where $\partial_{\theta}\equiv\frac{\partial}{\partial\theta}$.

The figure of merit that we will use to quantify performance of an estimator $\hat{\theta}$ is the mean-squared error (MSE), defined as follows:
\begin{equation}
	\Delta^2\hat{\theta} := \mathbb{E}[(\hat{\theta}(X) - \theta)^2],
\end{equation}
The quantitative goal of parameter estimation is to minimize this quantity. Throughout this thesis, we focus exclusively on unbiased estimators; i.e., estimators $\hat{\theta}$ for which
\begin{equation}
    \mathbb{E}[\hat{\theta}(X) ] = \theta~~\forall~\theta \in \Theta. \label{eq:prelims-unbiasedness}
\end{equation}
For an unbiased estimator, the MSE is equal to the variance; i.e.,
\begin{equation}
    \mathbb{E}[(\hat{\theta}(X) - \theta)^2] = \text{Var} (\hat{\theta}).
\end{equation}
That is, we assume that the estimate is accurate and proceed to quantify and place limits on its precision.

The fundamental theoretical tool in estimation theory is the Cramer--Rao bound (henceforth denoted often as CRB). It is a lower bound on the mean-squared error of an unbiased estimator; i.e., its variance:
\begin{equation}
    \text{Var}(\hat{\theta}(X)) \geq \frac{1}{I_{F}(\theta; \{p_{\theta}\}_\theta)},
\end{equation}
where $I_{F}(\theta; \{p_{\theta}\}_\theta)$ is the Fisher information of the family of distributions $\{ p_\theta \}_\theta$. By association with the Cramer--Rao bound, the Fisher information takes on its operational meaning in estimation theory. The Fisher information of a family of parameterized probability distributions is defined as follows:

\begin{definition}[Fisher information]
For a parameterized family of probability distributions $\{p_{\theta}(x)\}_{\theta}$, the Fisher information is defined as follows:
	\begin{equation}
        I_{F}(\theta;\{p_{\theta}\}_{\theta}):=\left\{
        \begin{array}
            [c]{cc}
            \int_{\Omega}dx\ \frac{1}{p_{\theta}(x)}\left(  \partial_{\theta}p_{\theta
            }(x)\right)  ^{2} & \text{if }\operatorname{supp}(\partial_{\theta}p_{\theta
            })\subseteq\operatorname{supp}(p_{\theta})\\
            +\infty & \text{otherwise}
            \end{array}
        \right.  , \label{eq:CFI}
    \end{equation}
    where $\Omega$ is the sample space for the probability density function $p_{\theta}(x)$. The support of  distribution $p_\theta$ is the smallest closed set $R \in \mathbb{R}$ such that $p_\theta(x \in R) = 1$. 

    Alternatively, when the support condition
    \begin{equation}
        \operatorname{supp}(\partial_{\theta}p_{\theta})\subseteq\operatorname{supp}(p_{\theta})
    \end{equation}
    is satisfied (understood as \textquotedblleft essential support\textquotedblright), the Fisher information has the following expression:
    \begin{equation}
        I_{F}(\theta;\{p_{\theta}\}_{\theta})    =\int_{\Omega}dx\ p_{\theta
        }(x)\left(  \partial_{\theta}\ln p_{\theta}(x)\right)  ^{2}
          =\mathbb{E}[\left(  \partial_{\theta}\ln p_{\theta}(X)\right)
        ^{2}], \label{eq:CFI-for-RLD}
    \end{equation}
    interpreted as the variance of the surprisal rate $\partial_{\theta}[- \ln p_{\theta}(x)]$. The quantity $\partial_{\theta}[\ln p_{\theta}(x)]$ is known as the logarithmic derivative.
\end{definition}

If one generates $n$ independent samples $x^n \equiv x_1, \ldots, x_n$ of $p_{\theta}(x)$, described by the random sequence $X^n \equiv X_1, \ldots, X_n$, and forms an unbiased estimator $\hat{\theta}(x^n)$, then the Fisher information increases linearly with $n$ and the CRB becomes as follows:
\begin{equation}
    \text{Var}(\hat{\theta}(X^n)) \geq \frac{1}{n I_{F}(\theta; \{p_{\theta}\}_\theta)}.
\end{equation}

This scaling of the Cramer--Rao bound as $\frac{1}{n}$ is often termed the shot-noise limit. It is a fundamental limit that applies to estimation tasks when using classical resources. However, as we stated briefly in Chapter~\ref{ch:intro}, it is possible to perform better than the shot noise limit in certain cases when using quantum resources; i.e., when using quantum probe states and quantum measurements. In the asymptotic limit of large $n$, the Cramer--Rao bound above is attained by the maximum likelihood estimator. Specifically, in certain cases of interest when using quantum resources, the variance of an unbiased estimator can be made to scale as $\frac{1}{n^2}$. We expand on this further in Section~\ref{sec:prelims-state-estimation} of this chapter.

Finally, we note that there are different paradigms of estimation theory, in line with the different interpretations of probability. These are the frequentist and Bayesian paradigms. In this thesis, we work in the frequentist paradigm, where the MSE and therefore the Cramer--Rao bound generally depend on the value of the unknown parameter $\theta$ itself, unlike in the Bayesian regime where this is not a concern. In the frequentist paradigm, though, this concern can be alleviated by enforcing the unbiasedness condition~\eqref{eq:prelims-unbiasedness}. 

\section{Quantum information preliminaries}

Before we can describe quantum parameter estimation, we need to introduce some quantum information preliminaries. This section briefly reviews the quantum information formalism and tools that we use in this thesis. We follow the convention and material in~\cite{Wilde2017}. 

\begin{definition}[Hilbert space]
    A Hilbert space is an inner product vector space over complex numbers $\mathbb{C}$. The inner product maps a pair of vectors $\ket{\psi}$ and $\ket{\phi}$ to an element of $\mathbb{C}$, and has the following properties:
    \begin{itemize}
        \item Positivity: $\langle \psi | \psi \rangle \geq 0$. The equality is satisfied if and only if $\ket{\psi} = 0$.
        \item Linearity: $\langle \phi | \lambda_1 \psi_1 + \lambda_2 \psi_2 \rangle = \lambda_1 \langle \phi | \psi_1 \rangle + \lambda_2 \langle \phi | \psi_2 \rangle$ where $\lambda_1, \lambda_2 \in \mathbb{C}$ and $\ket{\psi_1}, \ket{\psi_2}, \ket{\phi}$ are vectors in the Hilbert space $\mathcal{H}$.
        \item Skew-symmetry: $\langle\phi | \psi\rangle = \overline{\langle\psi | \phi\rangle}$ where $\overline{c}$ denotes the complex conjugate of complex number $c$.
    \end{itemize}
\end{definition}

\begin{definition}[Quantum states]
    A quantum state $\rho$ on Hilbert space $\mathcal{H}$ is a positive semidefinite, Hermitian operator with trace equal to one. That is, $\rho \geq 0$, $\rho = \rho^{\dag}$, and $\Tr[\rho]=1$. The set of quantum states on $\mathcal{H}$ is denoted as $\mathcal{D}(\mathcal{H})$.
\end{definition}

Quantum states as defined above are also known as density operators.
A special case of quantum states are what are known as \textit{pure} states $\ket{\phi}$, which are vectors on the Hilbert space $\mathcal{H}$ with norm equal to $1$.  

Quantum states are static objects that describe the state of any given quantum system. The evolution of a quantum system from an initial state to a final state is most generally described by a quantum channel. We denote the set of bounded operators acting on $\mathcal{H}$ as $\mathcal{B} (\mathcal{H})$. A bounded operator $M\in \mathcal{B}(\mathcal{H})$ is trace-class if $\Vert M \Vert_1 \leq \infty$. We denote the set of trace-class operators on $\mathcal{H}$ as $\mathcal{T} (\mathcal{H})$.
\def\H{\mathcal{H}}
\begin{definition}[Positive map]
    A linear map $\N_{A \to B}:\B ( \mathcal{H}_A) \to \B(\H_{B})$ is positive if $\N_{A \to B}(M_A) \geq 0$, for all $M_A\geq 0$, where $M_A\in \B(\H_A)$.
\end{definition}

\begin{definition}[Completely-positive map]
    A linear map $\N_{A\to B}:\B(\H_A) \to \B(\H_B)$ is completely positive if $\I_R \otimes \N_A$ is a positive map for all possible $\H_R$, where $\H_R$ represents a Hilbert space extending $\H_A$.  
\end{definition}

\begin{definition}[Trace-preserving map]
  A linear map $\N_{A\to B}:\T(\H_A) \to \T(\H_B)$ is trace preserving if 
    \begin{equation}
        \Tr(M_A) = \Tr(\N_{A\to B}(M_A))~,
    \end{equation}
    for all $M_a \in \T(\H_A)$.
\end{definition}

\begin{definition}[Quantum channel]
    A quantum channel $\N_{A\to B}:\T(\H_A) \to \T(\H_B)$ is a completely-positive and trace-preserving linear map.  
\end{definition}

\begin{definition}[Choi operator]
    The Choi operator $\Gamma_{RB}^{\mathcal{N}}$ of a quantum channel
$\mathcal{N}_{A\rightarrow B}$ is defined as
\begin{equation}
\Gamma_{RB}^{\mathcal{N}}:=\mathcal{N}_{A\rightarrow B}(\Gamma_{RA}),
\end{equation}
where
\begin{equation}
\Gamma_{RA}:=|\Gamma\rangle\!\langle\Gamma|_{RA}, \text{and}~~|\Gamma\rangle_{RA} := \sum_i |i\rangle_R |i\rangle_A
\end{equation}
denotes the unnormalized maximally entangled vector on systems $R$ and $A$. The sets $\{|i\rangle_{R}\}_{i}$ and $\{|i\rangle_{A}\}_{i}$ are orthonormal bases for the isomorphic Hilbert spaces $\mathcal{H}_{R}$ and $\mathcal{H}_{A}$. The Choi operator is positive semi-definite and satisfies the following property as a consequence of $\mathcal{N}_{A\rightarrow B}$ being trace preserving:
\begin{equation}
\operatorname{Tr}_{B}[\Gamma_{RB}^{\mathcal{N}}]=I_{R}.
\end{equation}
\end{definition}

A measurement is the method by which one can extract classical knowledge from the state of a quantum system. The information obtained may correspond to various properties of the quantum system, e.g., position, momentum, or spin of its state. 

\begin{definition}[Measurement] Let $\rho \in \D(\H)$ be a density operator. Let $\{M_k\}_k$ denote a set of measurement operators for which $\sum_k M_k^{\dag}M_k=I$, where $I$ is the identity operator. Then the probability of obtaining outcome $k$ after the measurement is given by 
\begin{equation}
    p_K(k) = \Tr(M_k^{\dag}M_k \rho)~,
\end{equation}
and the post-measurement state $\tilde{\rho}_k$ is given by 
\begin{equation}
    \tilde{\rho}_k = \frac{M_k \rho M_k^{\dag}}{p_K(k)}~.
\end{equation}
\end{definition}

If we are willing to forego knowledge of the state of the quantum system after the measurement, then the measurement can be more generally described by a positive-operator-valued measure (POVM).

\begin{definition}[POVM]
A positive operator-valued measure (POVM) is a set $\{\Lambda_j\}_j$ of operators that satisfy the following properties: 
\begin{equation}
    \Lambda_j \geq 0 \quad \text{and} \quad \sum_{j} \Lambda_j = I, ~\forall j~.
\end{equation}
The probability of observing the outcome $k$ when state $\rho$ is measured using the above POVM is given by $\Tr \left[ \Lambda_k \rho \right]$.  
\end{definition}

Finally, we provide three theoretical tools that we use in this thesis to prove the subsequent technical results of Chapters~\ref{ch:single} and~\ref{ch:multi}. 

\begin{proposition}
    A pure bipartite state $|\psi\rangle_{RA}$ can be written as $(X_{R}\otimes I_{A})|\Gamma\rangle_{RA}$ where $X_{R}$ is an operator satisfying $\operatorname{Tr}[X_{R}^{\dag}X_{R}]=1$.
\end{proposition}

\begin{proposition}
    For a linear operator~$M$, the following transpose trick identity holds:
    \begin{equation}
        \left(  I_{R}\otimes M_{A}\right)  |\Gamma\rangle_{RA}=\left(  M_{R}^{T}\otimes I_{A}\right)  |\Gamma\rangle_{RA}, \label{eq:transpose-trick}
    \end{equation}
    where $M^{T}$ denotes the transpose of $M$ with respect to the orthonormal basis $\{|i\rangle_{R}\}_{i}$. For a linear operator $K_{R}$, the following identity holds:
    \begin{equation}
        \langle\Gamma|_{RA}\left(  K_{R}\otimes I_{A}\right)  |\Gamma\rangle_{RA}=\operatorname{Tr}[K_{R}]. \label{eq:max-ent-partial-trace}
    \end{equation}
\end{proposition}

\begin{proposition}[Post-selected teleportation identity]
    The output of a quantum channel $\mathcal{N}_{A\rightarrow B}$\ on an
input quantum state $\rho_{RA}$ can be rewritten in the following way~\cite{Bennett2005}:
    \begin{equation}
    \mathcal{N}_{A\rightarrow B}(\rho_{RA})=\langle\Gamma|_{AS}\rho_{RA}\otimes\Gamma_{SB}^{\mathcal{N}}|\Gamma\rangle_{AS}\text{,} \label{eq:PS-TP-identity}
    \end{equation}
    where $S$ is a system isomorphic to the channel input system $A$.
\end{proposition}

\bigskip

\section{Quantum state estimation} \label{sec:prelims-state-estimation}

Now that we have introduced classical estimation as well as the quantum information preliminaries that we need, we move to quantum estimation and introduce and describe the task in detail.

First, we consider the task of estimating a single unknown parameter encoded in a quantum state, which can be seen as the quantum analog of estimating a single parameter in an unknown probability distribution. This is a stepping stone towards the more involved task of estimating a parameter encoded in a quantum channel. Analogous to a family of probability distributions introduced earlier, consider that we have a parameterized family $\{ \rho_\theta \}_{\theta}$ of quantum states into which the parameter $\theta \in \Theta \subseteq \mathbb{R}$ is encoded. In quantum state estimation, the unknown parameter $\theta$ is encoded in a quantum state $\rho_\theta$. Consider that we have $n$ copies of the state $\rho_\theta$. To invoke the theory and framework of classical estimation theory, the $n$ copies $\rho_\theta^{\otimes n}$ are subjected to a POVM $\{ \Lambda_x \}_x$. This yields a probability distribution according to the Born rule:
\begin{equation}
    p_\theta(x) = \Tr[ \Lambda_x \rho_\theta^{\otimes n} ].
\end{equation} 
The task now is to obtain a good estimate of parameter $\theta$ from the above probability distribution. We denote the estimate obtained as $\hat{\theta}$. The estimate $\hat{\theta}$ is a function of the measurement $\{ \Lambda_x \}_x$. Thus, each $\{\Lambda_x\}_x$ yields a probability distribution which yields its own Fisher information, and hence its own classical Cramer--Rao bound. To obtain the best or the \textit{most informative} CRB, we would like to identify and perform the best possible measurement. This makes quantum estimation a much more involved task theoretically and experimentally than classical estimation. In the case of estimating a single parameter encoded in a quantum state, the most informative CRB involves what is known as the \textit{SLD Fisher information}. To understand what it is and how it arises, we first describe how to generalize the Fisher information as defined in~\eqref{eq:CFI} to the quantum case. 

\subsection{Quantum Fisher information}

In classical estimation, we were able to define the logarithmic derivative $\ln p_\theta(x)$ of a parameterized probability distribution. Since quantum states are operators, the logarithmic derivative for quantum states also takes the form of an operator. We admit any logarithmic derivative operator as long as it collapses to the scalar logarithmic derivative in the classical case. The general noncommutativity of operators results in an infinite number of such logarithmic derivative operators, and therefore an infinite number of quantum generalizations of the classical Fisher information. 

Consider the parameterized logarithmic derivative operator $D_{\theta, p}$ to be defined implicitly via the following differential equation:
\begin{equation}
    \partial_\theta \rho_\theta = \left( p D_{\theta, p} \rho_\theta  + (1-p) \rho_\theta D_{\theta, p} \right). \label{eq:parameterized-log-der-operator}
\end{equation}

In the classical case, when the state $\rho_\theta$ is represented by a diagonal matrix, the $D_{\theta, p}$ operator is also a classical (diagonal) operator for all $0 \leq p \leq 1$. In general, if $p$ is set to $1/2$, the logarithmic derivative operator $D_{\theta, 1/2}$ is known as the symmetric logarithmic derivative (SLD) operator $L_\theta$. On the other hand, if $p$ is set to $0$, then the logarithmic derivative operator $D_{\theta, 0}$ is known as the right logarithmic derivative (RLD) operator $R_\theta$. That is, the SLD and RLD operators are defined implicitly via the following differential equations:
\begin{align}
    \partial_\theta \rho_\theta &= \frac{1}{2} \left( L_\theta \rho_\theta + \rho_\theta L_\theta \right),\\
    \partial_\theta \rho_\theta &= \rho_\theta R_\theta .
\end{align}

Now that we have defined the SLD and RLD operators, we are in a position to define the quantum Fisher information quantities that each of them yields.

\begin{definition}
[SLD\ Fisher information]\label{def:SLD-Fish-states}Let $\{\rho_{\theta
}\}_{\theta}$ be a differentiable family of quantum states. Then the
SLD\ Fisher information is defined as follows:
\begin{equation}
I_{F}(\theta;\{\rho_{\theta}\}_{\theta})=\left\{
\begin{array}
[c]{cc}
2\left\Vert \left(  \rho_{\theta}\otimes I+I\otimes\rho_{\theta}^{T}\right)
^{-\frac{1}{2}}\left(  (\partial_{\theta}\rho_{\theta})\otimes I\right)
|\Gamma\rangle\right\Vert _{2}^{2} & \text{if }\Pi_{\rho_{\theta}}^{\perp
}(\partial_{\theta}\rho_{\theta})\Pi_{\rho_{\theta}}^{\perp}=0\\
+\infty & \text{otherwise}
\end{array}
\right.  , \label{eq:basis-independent-formula-SLD}
\end{equation}
where $\Pi_{\rho_{\theta}}^{\perp}$ denotes the projection onto the kernel of
$\rho_{\theta}$, $|\Gamma\rangle=\sum_{i}|i\rangle|i\rangle$ is the
unnormalized maximally entangled vector, $\{|i\rangle\}_{i}$ is any
orthonormal basis, the transpose in \eqref{eq:basis-independent-formula-SLD}
is with respect to this basis, and the inverse is taken on the support of
$\rho_{\theta}\otimes I+I\otimes\rho_{\theta}^{T}$.
\end{definition}

When the finiteness condition $\Pi_{\rho_{\theta}}^{\perp}(\partial_{\theta
}\rho_{\theta})\Pi_{\rho_{\theta}}^{\perp}=0$ holds, the SLD Fisher information can alternatively be defined using the SLD operator $L_\theta$ as follows:
\begin{equation}
I_{F}(\theta;\{\rho_{\theta}\}_{\theta}):=\mathrm{Tr}[ L_{\theta}^{2}
\rho_{\theta}] =\operatorname{Tr}[L_{\theta}(\partial_{\theta}\rho_{\theta})].
\label{eq:SLD-FI}
\end{equation}

The SLD Fisher information can also be written in terms of the spectral decomposition of $\rho_{\theta}$. Let the spectral decomposition of $\rho_{\theta}$ be given as
\begin{equation}
\rho_{\theta}=\sum_{j}\lambda_{\theta}^{j}|\psi_{\theta}^{j}\rangle\!\langle
\psi_{\theta}^{j}|,
\end{equation}
which includes the indices for which $\lambda_{\theta}^{j}=0$. Then the
projection $\Pi_{\rho_{\theta}}^{\perp}$ onto the kernel of $\rho_{\theta}$ is
given by
\begin{equation}
\Pi_{\rho_{\theta}}^{\perp}:=\sum_{j:\lambda_{\theta}^{j}=0}|\psi_{\theta}
^{j}\rangle\!\langle\psi_{\theta}^{j}|.
\end{equation}
With this notation, the SLD\ quantum Fisher information can also be written as follows:
\begin{equation}
I_{F}(\theta;\{\rho_{\theta}\}_{\theta})=\left\{
\begin{array}
[c]{cc}
2\sum_{j,k:\lambda_{j}^{\theta}+\lambda_{k}^{\theta}>0}\frac{|\langle
\psi_{\theta}^{j}|(\partial_{\theta}\rho_{\theta})|\psi_{\theta}^{k}
\rangle|^{2}}{\lambda_{\theta}^{j}+\lambda_{\theta}^{k}} & \text{if }\Pi
_{\rho_{\theta}}^{\perp}(\partial_{\theta}\rho_{\theta})\Pi_{\rho_{\theta}
}^{\perp}=0\\
+\infty & \text{otherwise}
\end{array}
\right.  . \label{eq:SLD-Fish-info-formula}
\end{equation}

\medskip

\begin{definition}
[RLD\ Fisher information]\label{def:RLD-Fish-info-states}Let $\{\rho_{\theta
}\}_{\theta}$ be a differentiable family of quantum states. Then the
RLD\ Fisher information is defined as follows:
\begin{equation}
\widehat{I}_{F}(\theta;\{\rho_{\theta}\}_{\theta})=\left\{
\begin{array}
[c]{cc}
\operatorname{Tr}[(\partial_{\theta}\rho_{\theta})^{2}\rho_{\theta}^{-1}] &
\text{if }\operatorname{supp}(\partial_{\theta}\rho_{\theta})\subseteq
\operatorname{supp}(\rho_{\theta})\\
+\infty & \text{otherwise}
\end{array}
\right.  , \label{eq:RLD-FI}
\end{equation}
where the inverse $\rho_{\theta}^{-1}$ is taken on the support of
$\rho_{\theta}$.
\end{definition}

Similar to the case of SLD Fisher information, when the finiteness condition $\operatorname{supp}(\partial_{\theta}\rho_{\theta})\subseteq
\operatorname{supp}(\rho_{\theta})$ holds, the RLD Fisher information can be defined using the RLD operator $R_\theta$ as follows:
\begin{equation}
I_{F}(\theta;\{\rho_{\theta}\}_{\theta}):=\operatorname{Tr}[R_{\theta
}R_{\theta}^{\dag}\rho_{\theta}]=\operatorname{Tr}[(\partial_{\theta}
\rho_{\theta})R_{\theta}^{\dag}].
\end{equation}

Note that the support condition $\operatorname{supp}(\partial_{\theta}\rho_{\theta})\subseteq
\operatorname{supp}(\rho_{\theta})$ is equivalent to $\Pi^{\perp}_{\rho_\theta} \partial_\theta \rho_\theta =  \partial_\theta \rho_\theta \Pi^{\perp}_{\rho_\theta} = 0$, which implies that $\Pi^{\perp}_{\rho_\theta}\partial_\theta \rho_\theta \Pi^{\perp}_{\rho_\theta} = 0$. That is, the SLD Fisher information is finite whenever the RLD Fisher information is finite.

\subsection{Properties of the SLD and RLD Fisher information for quantum states}

Here we include some properties of the SLD and RLD Fisher information that will be useful to prove our results for quantum channel estimation (which we will introduce in Section~\ref{sec:prelims-state-estimation} of this chapter). The proofs of Propositions~\ref{prop:additivity-SLD-RLD-states}, \ref{prop:cq-decomp-SLD-RLD}, \ref{prop:physical-consistency-SLD-Fish-states} and \ref{prop:physical-consistency-RLD-Fish-states} can be found in \cite{Katariya2021}.

\subsubsection{Faithfulness}

\begin{proposition}
[Faithfulness]\label{prop:faithfulness-SLD-RLD-Fish}For a differentiable
family $\{\rho_{A}^{\theta}\}_{\theta}$ of quantum states, the SLD\ and
RLD\ Fisher informations are equal to zero:
\begin{equation}
I_{F}(\theta;\{\rho_{A}\}_{\theta})=\widehat{I}_{F}(\theta;\{\rho
_{A}\}_{\theta})=0\qquad\forall\theta\in\Theta,
\label{eq:SLD-RLD-Fish-faithful}
\end{equation}
if and only if $\rho_{A}^{\theta}$ has no dependence on the parameter $\theta$
(i.e., $\rho_{A}^{\theta}=\rho_{A}$ for all $\theta$).
\end{proposition}

\begin{proof}
The if-part follows directly from plugging into the definitions after
observing that $\partial_{\theta}\rho_{\theta}=0$ for a constant family. So we
now prove the only-if part. If $I_{F}(\theta;\{\rho_{A}\}_{\theta})=0$, then
it is necessary for the finiteness condition in
\eqref{eq:basis-independent-formula-SLD}\ to hold (otherwise we would
have a contradiction). Then this means that
\begin{align}
\Pi_{\rho_{\theta}}^{\perp}(\partial_{\theta}\rho_{\theta})\Pi_{\rho_{\theta}
}^{\perp}  &  =0,\\
2\sum_{j,k:\lambda_{j}^{\theta}+\lambda_{k}^{\theta}>0}\frac{|\langle
\psi_{\theta}^{j}|(\partial_{\theta}\rho_{\theta})|\psi_{\theta}^{k}
\rangle|^{2}}{\lambda_{\theta}^{j}+\lambda_{\theta}^{k}}  &  =0\qquad
\forall\theta.
\end{align}
By sandwiching the first equation by $\langle\psi_{\theta}^{j}|$ and
$|\psi_{\theta}^{k}\rangle$ for which $\lambda_{\theta}^{j},\lambda_{\theta
}^{k}=0$, we find that these matrix elements $\langle\psi_{\theta}
^{j}|(\partial_{\theta}\rho_{\theta})|\psi_{\theta}^{k}\rangle$\ of
$\partial_{\theta}\rho_{\theta}$ are equal to zero. Since $\lambda_{\theta
}^{j}+\lambda_{\theta}^{k}>0$ in the latter expression, the latter equality
implies the following
\begin{equation}
|\langle\psi_{\theta}^{j}|(\partial_{\theta}\rho_{\theta})|\psi_{\theta}
^{k}\rangle|^{2}=0
\end{equation}
for all $\lambda_{\theta}^{j}$ and $\lambda_{\theta}^{k}$ satisfying
$\lambda_{\theta}^{j}+\lambda_{\theta}^{k}>0$. This implies that these matrix
elements $\langle\psi_{\theta}^{j}|(\partial_{\theta}\rho_{\theta}
)|\psi_{\theta}^{k}\rangle$ of $\partial_{\theta}\rho_{\theta}$ are equal to
zero. These are all possible matrix elements, and so we conclude that
$\partial_{\theta}\rho_{\theta}=0$. This in turn implies that $\rho_{\theta}$
is a constant family (i.e., $\rho_{A}^{\theta}=\rho_{A}$ for all $\theta$). If
$\widehat{I}_{F}(\theta;\{\rho_{A}\}_{\theta})=0$, then by the inequality in
\eqref{eq:SLD<=RLD}, $I_{F}(\theta;\{\rho_{A}\}_{\theta})=0$. Then by what we
have just shown, $\rho_{\theta}$ is a constant family in this case also.
\end{proof}

This property is a basic one that we expect the SLD and RLD Fisher informations to obey. Since they quantify the amount of information about a parameter present in a family of quantum states, it is natural to expect them to be zero when the family of quantum states has no dependence on the parameter of interest.

\subsubsection{Data processing}

The SLD\ and RLD\ Fisher informations obey the following data-processing
inequalities:
\begin{align}
I_{F}(\theta;\{\rho_{A}^{\theta}\}_{\theta})  &  \geq I_{F}(\theta
;\{\mathcal{N}_{A\rightarrow B}(\rho_{A}^{\theta})\}_{\theta}
),\label{eq:DP-SLD}\\
\widehat{I}_{F}(\theta;\{\rho_{A}^{\theta}\}_{\theta})  &  \geq\widehat{I}
_{F}(\theta;\{\mathcal{N}_{A\rightarrow B}(\rho_{A}^{\theta})\}_{\theta}),
\label{eq:DP-RLD}
\end{align}
where $\mathcal{N}_{A\rightarrow B}$ is a quantum channel independent of the
parameter $\theta$ (more generally, these hold if $\mathcal{N}_{A\rightarrow
B}$ is a two-positive, trace-preserving map). The data-processing inequalities
for $I_{F}$ and $\widehat{I}_{F}$ were established in \cite{Petz1996}. In fact,
the inequality in \eqref{eq:DP-RLD}\ is an immediate consequence of
\cite[Proposition~4.1]{Choi1980}.

The fact that the SLD and RLD Fisher informations obey data-processing under quantum channels will prove to be a fundamental property that we will use in the rest of this chapter to establish bounds on state and channel estimation.

\subsubsection{Additivity}

\begin{proposition}
\label{prop:additivity-SLD-RLD-states}Let $\{\rho_{A}^{\theta}\}_{\theta}$ and
$\{\sigma_{A}^{\theta}\}_{\theta}$ be differentiable families of quantum
states. Then the SLD\ and RLD\ Fisher informations are additive in the
following sense:
\begin{align}
I_{F}(\theta;\{\rho_{A}^{\theta}\otimes\sigma_{B}^{\theta}\}_{\theta})  &
=I_{F}(\theta;\{\rho_{A}^{\theta}\}_{\theta})+I_{F}(\theta;\{\sigma
_{B}^{\theta}\}_{\theta}),\label{eq:additivity-SLD-states}\\
\widehat{I}_{F}(\theta;\{\rho_{A}^{\theta}\otimes\sigma_{B}^{\theta}
\}_{\theta})  &  =\widehat{I}_{F}(\theta;\{\rho_{A}^{\theta}\}_{\theta
})+\widehat{I}_{F}(\theta;\{\sigma_{B}^{\theta}\}_{\theta}).
\label{eq:additivity-RLD-states}
\end{align}

\end{proposition}

\subsubsection{Decomposition for classical--quantum families}

\begin{proposition}
\label{prop:cq-decomp-SLD-RLD}Let $\left\{  \rho_{XB}^{\theta}\right\}
_{\theta}$ be a differentiable family of classical--quantum states, where
\begin{equation}
\rho_{XB}^{\theta}:=\sum_{x}p_{\theta}(x)|x\rangle\!\langle x|_{X}\otimes
\rho_{\theta}^{x}.
\end{equation}
Then the following decompositions hold for the SLD\ and RLD\ Fisher
informations:
\begin{align}
I_{F}(\theta;\{\rho_{XB}^{\theta}\}_{\theta})  &  =I_{F}(\theta;\{p_{\theta
}\}_{\theta})+\sum_{x}p_{\theta}(x)I_{F}(\theta;\{\rho_{\theta}^{x}\}_{\theta
}),\\
\widehat{I}_{F}(\theta;\{\rho_{XB}^{\theta}\}_{\theta})  &  =I_{F}
(\theta;\{p_{\theta}\}_{\theta})+\sum_{x}p_{\theta}(x)\widehat{I}_{F}
(\theta;\{\rho_{\theta}^{x}\}_{\theta}).
\end{align}

\end{proposition}

\subsubsection{Physical consistency of SLD and RLD Fisher information}

The SLD and RLD Fisher information both are physically consistent; i.e., they are both the result of a limiting procedure in which some constant additive noise vanishes.

\begin{proposition}
\label{prop:physical-consistency-SLD-Fish-states}Let $\{\rho_{\theta
}\}_{\theta}$ be a differentiable family of quantum states. Then the
SLD\ Fisher information in \eqref{eq:basis-independent-formula-SLD}\ is given by the
following limit:
\begin{equation}
I_{F}(\theta;\{\rho_{\theta}\}_{\theta})=\lim_{\varepsilon\rightarrow0}
I_{F}(\theta;\{\rho_{\theta}^{\varepsilon}\}_{\theta}),
\label{eq:fish-limit-delta}
\end{equation}
where
\begin{equation}
\rho_{\theta}^{\varepsilon}:=\left(  1-\varepsilon\right)  \rho_{\theta
}+\varepsilon\pi_{d}, \label{eq:rho-delta-def}
\end{equation}
and $\pi_{d}:=I/d$ is the maximally mixed state, with $d$ large enough so that
$\operatorname{supp}(\rho_{\theta})\subseteq\operatorname{supp}(\pi)$ for all
$\theta$.
\end{proposition}

\begin{proposition}
\label{prop:physical-consistency-RLD-Fish-states}Let $\{\rho_{\theta
}\}_{\theta}$ be a differentiable family of quantum states. Then the
RLD\ Fisher information in \eqref{eq:RLD-FI}\ is given by the following limit:
\begin{equation}
\widehat{I}_{F}(\theta;\{\rho_{\theta}\}_{\theta})=\lim_{\varepsilon
\rightarrow0}\widehat{I}_{F}(\theta;\{\rho_{\theta}^{\varepsilon}\}_{\theta}),
\end{equation}
where
\begin{equation}
\rho_{\theta}^{\varepsilon}:=\left(  1-\varepsilon\right)  \rho_{\theta
}+\varepsilon\pi_{d},
\end{equation}
and $\pi_{d}:=I/d$ is the maximally mixed state, with $d$ large enough so that
$\operatorname{supp}(\rho_{\theta})\subseteq\operatorname{supp}(\pi)$ for all
$\theta$.
\end{proposition}

\subsubsection{Relation between SLD and RLD Fisher information of quantum states}

We next show an important relationship obeyed by the SLD and RLD Fisher informations of quantum states, namely that for single parameter quantum state estimation, the SLD Fisher information never exceeds the RLD Fisher information.
\begin{proposition} \label{prop:SLD<=RLD}
For a family of quantum states $\{ \rho_\theta \}_{\theta}$, the SLD Fisher information never exceeds its RLD Fisher information:
\begin{equation}
I_{F}(\theta;\{\rho_{\theta}\}_{\theta})\leq\widehat{I}_{F}(\theta
;\{\rho_{\theta}\}_{\theta}), \label{eq:SLD<=RLD}
\end{equation}
\end{proposition}
\begin{proof}
This can be seen from the operator convexity of the function $x^{-1}$ for
$x>0$. That is, for full-rank $\rho_{\theta}$, we have that
\begin{align}
2\left(  \rho_{\theta}\otimes I+I\otimes\rho_{\theta}^{T}\right)  ^{-1}  &
=\left(  \frac{1}{2}\rho_{\theta}\otimes I+\frac{1}{2}I\otimes\rho_{\theta
}^{T}\right)  ^{-1}\\
&  \leq\frac{1}{2}\left(  \rho_{\theta}\otimes I\right)  ^{-1}+\frac{1}
{2}\left(  I\otimes\rho_{\theta}^{T}\right)  ^{-1}\\
&  =\frac{1}{2}\left(  \rho_{\theta}^{-1}\otimes I\right)  +\frac{1}{2}\left(
I\otimes\rho_{\theta}^{-T}\right).
\end{align}
Then, we apply the transpose trick \eqref{eq:transpose-trick}, the identity \eqref{eq:max-ent-partial-trace}, the formula for SLD Fisher information \eqref{eq:basis-independent-formula-SLD} and the limit formulae in Propositions~\ref{prop:physical-consistency-SLD-Fish-states} and \ref{prop:physical-consistency-RLD-Fish-states} to get~\eqref{eq:SLD<=RLD}.
\end{proof}

\subsection{Cramer--Rao bounds for quantum state estimation}

Earlier, we stated that by optimizing over estimation strategies in the quantum setting (i.e., the selection of input probe state and final measurement), we can obtain the most informative lower bound on the variance of an unbiased estimator of the parameter $\theta$. For single parameter quantum state estimation, the most informative CRB is the one that arises from the SLD Fisher information:
\begin{equation}
\text{Var}(\hat{\theta})\geq\frac{1}{nI_{F}(\theta;\{\rho_{\theta}\}_{\theta
})}. \label{eq:QCRB}
\end{equation}

To obtain the above bound, we have also applied the additivity relation $I_{F}(\theta;\{\rho^{\otimes n}_{\theta}\}_{\theta
}) = nI_{F}(\theta;\{\rho_{\theta}\}_{\theta
})$. The lower bound in
\eqref{eq:QCRB} is achievable in the large $n$ limit of many copies of the
state $\rho_{\theta}$ \cite{Nagaoka1989,Braunstein1994}.

Further, the SLD Fisher information is never larger than the RLD Fisher information for single parameter estimation (Proposition~\ref{prop:SLD<=RLD}). This yields a straightforward Cramer--Rao bound that utilizes the RLD Fisher information:
\begin{equation}
\text{Var}(\hat{\theta})\geq\frac{1}{n\widehat{I}_{F}(\theta;\{\rho_{\theta}\}_{\theta })}. \label{eq:RLD-QCRB}
\end{equation}

Thus, for single parameter estimation, the Cramer--Rao bound involving the RLD Fisher information is always looser, or less useful, than the one involving the SLD Fisher information. Furthermore, if we have a pure state family $\{ \ket{\psi_\theta} \}_\theta$, the finiteness condition for the RLD Fisher information in~\eqref{eq:RLD-FI} is not satisfied and it yields an uninformative Cramer--Rao bound.

\subsection{Generalized Fisher information of quantum states}

As we stated earlier, there is an infinite number of ways to define the quantum Fisher information due to the noncommutative nature of states and operators in quantum information. We have since defined the SLD and RLD Fisher information using a parameterized logarithmic derivative operator in~\eqref{eq:parameterized-log-der-operator}. In a different vein now, we define the generalized Fisher information, wherein the only requirement is that the quantum Fisher information in question obeys the data-processing inequality, similar to how the SLD and RLD Fisher information do (see~\eqref{eq:DP-SLD} and \eqref{eq:DP-RLD}). Data processing is a fundamental and important tool in quantum information, and the motivation to introduce the generalized Fisher information is to use the data-processing inequality to establish as many of its properties as possible. We are further motivated to define it because of the earlier definitions of generalized distinguishability measures (or generalized divergences) \cite{Polyanskiy2010,Sharma2012} and their uses in quantum communication and distinguishability theory \cite{Sharma2012,Wilde2014,Gupta2015,Tomamichel2017,Wilde2017a,Leditzky2016,Kaur2018,Das2020,Kaur2019,Fang2019,Wang2019,Takeoka2016,Leditzky2018,Berta2018,Wang2019a}. We further define the generalized Fisher information of quantum channels later in this chapter.

\begin{definition}[Generalized Fisher information of quantum states]
\label{def:gen-fish-info-states} The generalized Fisher information $\mathbf{I}_{F}(\theta;\{\rho_{A}^{\theta}\}_{\theta})$ of a family $\{\rho_{A}^{\theta}\}_{\theta}$ of quantum states is a function $\mathbf{I}_{F} : \Theta\times\mathcal{D} \to\mathbb{R}$ that does not increase under the action of a parameter-independent quantum channel~$\mathcal{N}_{A\rightarrow B}$:
\begin{equation}
    \mathbf{I}_{F}(\theta;\{\rho_{A}^{\theta}\}_{\theta})\geq\mathbf{I}_{F}(\theta;\{\mathcal{N}_{A\rightarrow B}(\rho_{A}^{\theta})\}_{\theta}). \label{eq:gen-fisher-states}
\end{equation}
\end{definition}

Since both the SLD and the RLD Fisher information obey the data-processing inequality, they are particular examples of the generalized Fisher information of quantum states.  

An immediate consequence of Definition~\ref{def:gen-fish-info-states} is that the generalized Fisher information is equal to a constant, minimal value for a state family that has no dependence on the parameter $\theta$:
\begin{equation}
    \mathbf{I}_{F}(\theta;\{\rho_{A}\}_{\theta})=c. \label{eq:constant-value-gen-Fish-states}
\end{equation}
This follows because one can get from one fixed family $\{\rho_{A}\}_{\theta}$ to another $\{\sigma_{A}\}_{\theta}$ by means of a trace and replace channel $(\cdot)\rightarrow\operatorname{Tr}[(\cdot)]\sigma_{A}$, and then we apply the data-processing inequality. If this constant $c$\ is equal to zero, then we say that the generalized Fisher information is \textit{weakly faithful}.

A generalized Fisher information obeys the \textit{direct-sum property} if the following equality holds
\begin{equation}
    \mathbf{I}_{F}\!\left(  \theta;\left\{  \sum_{x}p(x)|x\rangle\!\langle x|\otimes\rho_{\theta}^{x}\right\}  _{\theta}\right)  =\sum_{x}p(x)\mathbf{I}_{F}(\theta;\left\{  \rho_{\theta}^{x}\right\}  _{\theta}), \label{eq:direct-sum-prop-gen-fish}
\end{equation}
where, for each $x$, the family $\left\{  \rho_{\theta}^{x}\right\}  _{\theta}$ of quantum states is differentiable. Observe that the probability distribution $p(x)$ has no dependence on the parameter $\theta$. If a generalized Fisher information obeys the direct-sum property, then it is also convex in the following sense:
\begin{equation}
    \sum_{x}p(x)\mathbf{I}_{F}(\theta;\left\{  \rho_{\theta}^{x}\right\}_{\theta})\geq\mathbf{I}_{F}(\theta;\left\{  \overline{\rho}_{\theta}\right\}_{\theta}), \label{eq:gen-fish-convex}
\end{equation}
where $\overline{\rho}_{\theta}:=\sum_{x}p(x)\rho_{\theta}^{x}$. This follows by applying the direct-sum property \eqref{eq:direct-sum-prop-gen-fish} and the data-processing inequality with a partial trace over the classical register. Thus, due to their respective data-processing inequalities \eqref{eq:DP-SLD} and \eqref{eq:DP-RLD}, and Proposition~\ref{prop:cq-decomp-SLD-RLD}, the SLD and RLD Fisher informations are convex.

\bigskip

\section{Quantum channel estimation} \label{sec:prelims-channel-estimation}

The next rung in the hierarchical ladder of probability distributions and quantum states is quantum channels. Our next step, therefore, is to study the fundamental limits to estimation of parameters encoded in a quantum channel. In contrast to quantum states, quantum channels are dynamical objects which take quantum states to quantum states. Therefore, studying parameter estimation for quantum channels is a more mathematically involved task than the corresponding task for quantum states. 

In the previous section, we assumed that the $n$ copies of the quantum state $\rho_\theta$ were available in a tensor product form; i.e., $\rho_\theta^{\otimes n}$. However, since quantum channels are dynamic rather than static objects, they allow for more general interactions and probing. The most general way to use, or process, $n$ copies of a quantum channel is via a sequential or adaptive strategy, which we now describe.

\subsection{Sequential channel estimation setting}

\begin{figure}[ptb]
\centering
\includegraphics[width=\linewidth]{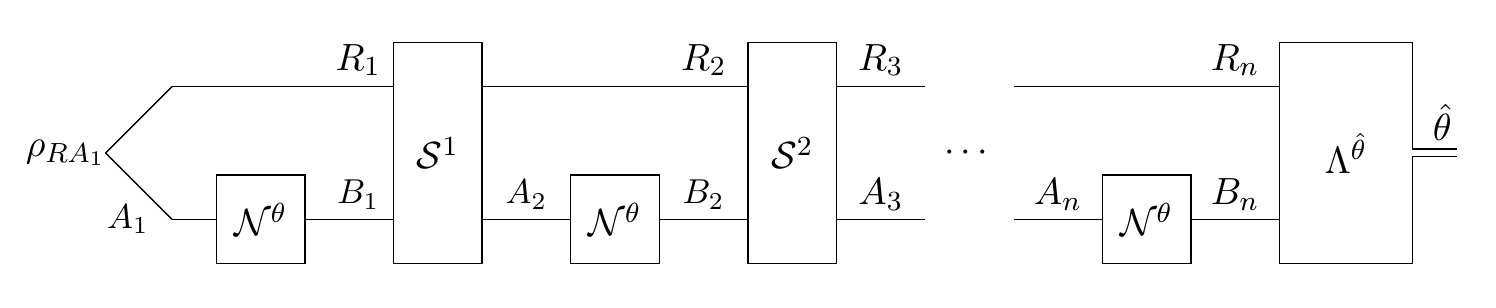}
\caption{Processing $n$ uses of channel $\mathcal{N}^{\theta}$ in a sequential manner.}{Processing $n$ uses of channel $\mathcal{N}^{\theta}$ in a sequential or adaptive manner is the most general approach to channel parameter estimation or discrimination. The $n$ uses of the channel are interleaved with $n-1$ quantum channels $\mathcal{S}^{1}$ through $\mathcal{S}^{n-1}$, which can also share memory systems with each other. The final measurement's outcome is then used to obtain an estimate of the unknown parameter $\theta$.}
\label{fig:adaptive-scheme}
\end{figure}

\begin{figure}[ptb]
\centering
\includegraphics[width=\linewidth]{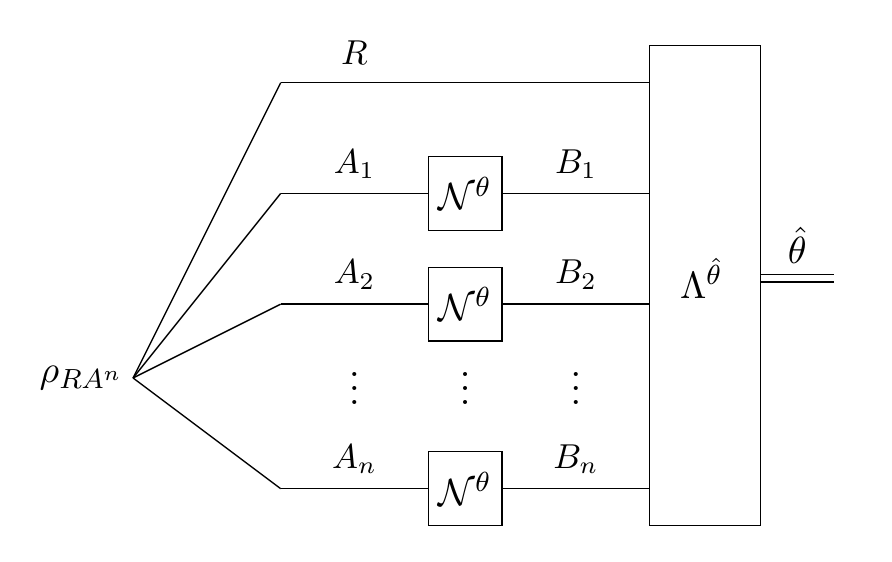}
\caption{Processing $n$ uses of channel $\mathcal{N}^{\theta}$ in a parallel manner.}{Processing $n$ uses of channel $\mathcal{N}^{\theta}$ in a parallel
manner. The $n$ channels are called in parallel, allowing for entanglement to
be shared among input systems $A_{1}$ through $A_{n}$, along with a quantum
memory system $R$. A collective measurement is made, with its outcome being an
estimate $\hat{\theta}$ for the unknown parameter $\theta$. Parallel
strategies form a special case of sequential ones, and therefore parallel
strategies are no more powerful than sequential ones.}
\label{fig:parallel-scheme}
\end{figure}

Let $\left\{  \mathcal{N}_{A\rightarrow B}^{\theta}\right\}_{\theta}$ denote a family of quantum channels with input system $A$ and output system $B$, such that each channel in the family is parameterized by a single real parameter $\theta\in\Theta\subseteq\mathbb{R}$, where $\Theta$ is the parameter set. The problem we consider is this: given a particular unknown channel $\mathcal{N}^{\theta}_{A \to B}$, how well can we estimate $\theta$ when allowed to probe the channel $n$ times? There are various ways that one can probe the quantum channel $n$ times such that each such procedure results in a probability distribution $p_{\theta}(x)$ for a final measurement outcome $x$, with corresponding random variable $X$. This distribution $p_{\theta}(x)$ depends on the unknown parameter $\theta$. Using the measurement outcome $x$, one formulates an estimate $\hat{\theta}(x)$ of the unknown parameter. An unbiased estimator satisfies $\mathbb{E}[\hat{\theta}(X)]=\theta$. We continue limiting our study to unbiased estimators.

The most general channel estimation procedure is depicted in Figure~\ref{fig:adaptive-scheme}. A sequential or adaptive strategy that makes $n$ calls to the channel is specified in terms of an input quantum state $\rho_{R_{1}A_{1}}$, a set of interleaved channels $\{\mathcal{S}_{R_{i} B_{i}\rightarrow R_{i+1}A_{i+1}}^{i}\}_{i=1}^{n-1}$, and a final quantum measurement $\{\Lambda_{R_{n}B_{n}}^{\hat{\theta}}\}_{\hat{\theta}}$ that outputs an estimate $\hat{\theta}$ of the unknown parameter (here we incorporate any classical post-processing of a preliminary measurement outcome $x$ to generate the estimate $\hat{\theta}$ as part of the final measurement). Note that any particular strategy $\{\rho_{R_{1}A_{1}},\{\mathcal{S}_{R_{i}B_{i}\rightarrow R_{i+1}A_{i+1}}^{i}\}_{i=1}^{n-1},\{\Lambda_{R_{n}B_{n}}^{\hat{\theta}}\}_{\hat{\theta}}\}$\ employed does not depend on the actual value of the unknown parameter $\theta$. We make the following abbreviation for a fixed strategy in what follows:
\begin{equation}
    \{\mathcal{S}^{(n)},\Lambda^{\hat{\theta}}\}\equiv\{\rho_{R_{1}A_{1}},\{\mathcal{S}_{R_{i}B_{i}\rightarrow R_{i+1}A_{i+1}}^{i}\}_{i=1}^{n-1},\{\Lambda_{R_{n}B_{n}}^{\hat{\theta}}\}_{\hat{\theta}}\}.
\end{equation}
The strategy begins with the estimator preparing the input quantum state $\rho_{R_{1}A_{1}}$ and sending the $A_{1}$ system into the channel $\mathcal{N}_{A_{1}\rightarrow B_{1}}^{\theta}$. The first channel $\mathcal{N}_{A_{1}\rightarrow B_{1}}^{\theta}$ outputs the system $B_{1}$, which is then available to the estimator. The resulting state is
\begin{equation}
    \rho_{R_{1}B_{1}}^{\theta}:=\mathcal{N}_{A_{1}\rightarrow B_{1}}^{\theta}(\rho_{R_{1}A_{1}}).
\end{equation}
The estimator adjoins the system $B_{1}$ to system $R_{1}$ and applies the channel $\mathcal{S}_{R_{1}B_{1}\rightarrow R_{2}A_{2}}^{1}$, leading to the state
\begin{equation}
    \rho_{R_{2}A_{2}}^{\theta}:=\mathcal{S}_{R_{1}B_{1}\rightarrow R_{2}A_{2}}^{1}(\rho_{R_{1}B_{1}}^{\theta}).
\end{equation}
The channel $\mathcal{S}_{R_{1}B_{1}\rightarrow R_{2}A_{2}}^{1}$ can take an action conditioned on information in the system $B_{1}$, which itself might contain some partial information about the unknown parameter~$\theta$. The estimator then inputs the system $A_{2}$ into the second use of the channel $\mathcal{N}_{A_{2}\rightarrow B_{2}}^{\theta}$, which outputs a system $B_{2}$ and gives the state
\begin{equation}
    \rho_{R_{2}B_{2}}^{\theta}:=\mathcal{N}_{A_{2}\rightarrow B_{2}}^{\theta}(\rho_{R_{2}A_{2}}^{\theta}).
\end{equation}
This process repeats $n-2$ more times, for which we have the intermediate states
\begin{align}
    \rho_{R_{i}B_{i}}^{\theta}  &  :=\mathcal{N}_{A_{i}\rightarrow B_{i}}^{\theta}(\rho_{R_{i}A_{i}}^{\theta}),\\
    \rho_{R_{i}A_{i}}^{\theta}  &  :=\mathcal{S}_{R_{i-1}B_{i-1}\rightarrow R_{i}A_{i}}^{i-1}(\rho_{R_{i-1}B_{i-1}}^{\theta}),
\end{align}
for $i\in\left\{  3,\ldots,n\right\}  $, and at the end, the estimator has systems $R_{n}$ and $B_{n}$. We define $\omega_{R_{n}B_{n}}^{\theta}$ to be the final state of the estimation protocol before the final measurement $\{\Lambda_{R_{n}B_{n}}^{\hat{\theta}}\}_{\hat{\theta}}$:
\begin{equation}
    \omega_{R_{n}B_{n}}^{\theta}:=(\mathcal{N}_{A_{n}\rightarrow B_{n}}^{\theta}\circ\mathcal{S}_{R_{n-1}B_{n-1}\rightarrow R_{n}A_{n}}^{n-1}\circ\cdots \circ\mathcal{S}_{R_{1}B_{1}\rightarrow R_{2}A_{2}}^{1}\circ\mathcal{N}
_{A_{1}\rightarrow B_{1}}^{\theta})(\rho_{R_{1}A_{1}}). \label{eq:estimation-final-state}
\end{equation}
The estimator finally performs a measurement $\{\Lambda_{R_{n}B_{n}}^{\hat{\theta}}\}_{\hat{\theta}}$ that outputs an estimate $\hat{\theta}$ of the unknown parameter $\theta$. The conditional probability for the estimate $\hat{\theta}$\ given the unknown parameter $\theta$ is determined by the Born rule:
\begin{equation}
    p_{\theta}(\hat{\theta})=\operatorname{Tr}[\Lambda_{R_{n}B_{n}}^{\hat{\theta}}\omega_{R_{n}B_{n}}^{\theta}]. \label{eq:cond-prob-adaptive}
\end{equation}
As we stated above, any particular strategy does not depend on the value of the unknown parameter $\theta$, but the states at each step of the protocol do depend on $\theta$ through the successive probings of the underlying channel $\mathcal{N}_{A\rightarrow B}^{\theta}$.

Note that such a sequential strategy contains a parallel or non-adaptive strategy as a special case: the system $R_{1}$ can be arbitrarily large and divided into subsystems, with the only role of the interleaved channels $\mathcal{S}_{R_{i}B_{i}\rightarrow R_{i+1}A_{i+1}}^{i}$ being that they redirect these subsystems to be the inputs of future calls to the channel (as would be the case in any non-adaptive strategy for estimation or discrimination). Figure~\ref{fig:parallel-scheme}\ depicts a parallel or non-adaptive channel estimation strategy.

\subsection{Quantum Fisher information of a channel family}

The first step we take towards defining the quantum Fisher information for quantum channels is extending the definition of generalized Fisher information from quantum states to quantum channels.

\begin{definition}[Generalized Fisher information of quantum channels]
\label{def:gen-fish-channels}The generalized Fisher information of a family $\{\mathcal{N}_{A\rightarrow B}^{\theta}\}_{\theta}$ of quantum channels is defined in terms of the following optimization:
\begin{equation}
    \mathbf{I}_{F}(\theta;\{\mathcal{N}_{A\rightarrow B}^{\theta}\}_{\theta}):=\sup_{\rho_{RA}}\mathbf{I}_{F}(\theta;\{\mathcal{N}_{A\rightarrow B}^{\theta}(\rho_{RA})\}_{\theta}). \label{eq:gen-fisher-channels}
\end{equation}
In the above definition, we take the supremum over arbitrary states $\rho_{RA}$\ with unbounded reference system $R$.
\end{definition}

\begin{remark} \label{rem:restrict-to-pure-bipartite}
As is the case for all information measures that obey the data-processing inequality, we can employ the data-processing inequality in \eqref{eq:gen-fisher-states} with respect to the partial trace operation and the Schmidt decomposition theorem to conclude that it suffices to perform the optimization in \eqref{eq:gen-fisher-channels} with respect to pure bipartite states $\psi_{RA}$ with system $R$ isomorphic to system $A$, so that
\begin{equation}
    \mathbf{I}_{F}(\theta;\{\mathcal{N}_{A\rightarrow B}^{\theta}\}_{\theta})=\sup_{\psi_{RA}}\mathbf{I}_{F}(\theta;\{\mathcal{N}_{A\rightarrow B}^{\theta}(\psi_{RA})\}_{\theta}).
\end{equation}
\end{remark}

\begin{proposition}
Let $\{\mathcal{N}_{A\rightarrow B}\}_{\theta}$ be a family of quantum channels that has no dependence on the parameter $\theta$, and suppose that the underlying generalized Fisher information is weakly faithful. Then
\begin{equation}
    \mathbf{I}_{F}(\theta;\{\mathcal{N}_{A\rightarrow B}\}_{\theta})=0.
\end{equation}
\end{proposition}

\begin{proof}
This follows as an immediate consequence of the definition~\eqref{eq:gen-fisher-channels},
\eqref{eq:constant-value-gen-Fish-states}, and the weak faithfulness assumption.
\end{proof}

\begin{proposition}
[Reduction to states]Let $\{\rho_{B}^{\theta}\}_{\theta}$ be a family of quantum states, and define the family $\{\mathcal{R}_{A\rightarrow B}^{\theta}\}_{\theta}$ of replacer channels as
\begin{equation}
    \mathcal{R}_{A\rightarrow B}^{\theta}(\omega_{A})=\operatorname{Tr}[\omega_{A}]\rho_{B}^{\theta}.
\end{equation}
Then
\begin{equation}
    \mathbf{I}_{F}(\theta;\{\mathcal{R}_{A\rightarrow B}^{\theta}\}_{\theta})=\mathbf{I}_{F}(\theta;\{\rho_{B}^{\theta}\}_{\theta}).
\end{equation}
\end{proposition}

\begin{proof}
This follows from the definition and the data-processing inequality. Consider that
\begin{align}
    \mathbf{I}_{F}(\theta;\{\mathcal{R}_{A\rightarrow B}^{\theta}\}_{\theta})  &=\sup_{\psi_{RA}}\mathbf{I}_{F}(\theta;\{\mathcal{R}_{A\rightarrow B}^{\theta }(\psi_{RA})\}_{\theta})\\
    &  =\sup_{\psi_{RA}}\mathbf{I}_{F}(\theta;\{\psi_{R}\otimes\rho_{B}^{\theta}\}_{\theta})\\
    &  =\mathbf{I}_{F}(\theta;\{\rho_{B}^{\theta}\}_{\theta}).
\end{align}
The last equality follows because
\begin{align}
    \mathbf{I}_{F}(\theta;\{\rho_{B}^{\theta}\}_{\theta})  &  \geq\mathbf{I}_{F}(\theta;\{\psi_{R}\otimes\rho_{B}^{\theta}\}_{\theta}),\\
    \mathbf{I}_{F}(\theta;\{\rho_{B}^{\theta}\}_{\theta})  &  \leq\mathbf{I}_{F}(\theta;\{\psi_{R}\otimes\rho_{B}^{\theta}\}_{\theta}),
\end{align}
with the first inequality following from the fact that there is a parameter-independent preparation channel such that $\rho_{B}^{\theta}\rightarrow\psi_{R}\otimes\rho_{B}^{\theta}$, while the second inequality follows from data-processing under partial trace over the reference system $R$.
\end{proof}


The SLD\ Fisher information of quantum channels was defined in \cite{Fujiwara2001} and the RLD\ Fisher information of quantum channels in \cite{Hayashi2011}; these are special cases of the generalized Fisher information of quantum channels \eqref{eq:gen-fisher-channels}.

\begin{definition}[SLD Fisher information of a channel family]
	\begin{equation}
    I_{F}(\theta;\{\mathcal{N}_{A\rightarrow B}^{\theta}\}_{\theta})=\left\{
    \begin{array}
    [c]{cc}
    \sup_{\psi_{RA}}I_{F}(\theta;\{\mathcal{N}_{A\rightarrow B}^{\theta}(\psi
    _{RA})\}_{\theta}) & \text{if }\Pi_{\Gamma_{RB}^{\mathcal{N}^{\theta}}}
    ^{\perp}(\partial_{\theta}\Gamma_{RB}^{\mathcal{N}^{\theta}})\Pi_{\Gamma
    _{RB}^{\mathcal{N}^{\theta}}}^{\perp}=0\\
    +\infty & \text{otherwise.}
    \end{array}
    \right.  , \label{eq:finiteness-condition-SLD-fish-ch}
    \end{equation}
\end{definition}

The following formula for the RLD\ Fisher information of quantum channels is known from \cite{Hayashi2011}:
\begin{definition}[RLD Fisher information of a channel family]
	\begin{multline} \label{eq:RLD-Fish-ch}
        \widehat{I}_{F}(\theta;\{\mathcal{N}_{A\rightarrow B}^{\theta}\}_{\theta})=\\
        \left\{
        \begin{array}
        [c]{cc}
        \left\Vert \operatorname{Tr}_{B}[(\partial_{\theta}\Gamma_{RB}^{\mathcal{N}^{\theta}})(\Gamma_{RB}^{\mathcal{N}^{\theta}})^{-1}(\partial_{\theta}\Gamma_{RB}^{\mathcal{N}^{\theta}})]\right\Vert _{\infty} & \text{if}\operatorname{supp}(\partial_{\theta}\Gamma_{RB}^{\mathcal{N}^{\theta}})\subseteq\operatorname{supp}(\Gamma_{RB}^{\mathcal{N}^{\theta}})\\
        +\infty & \text{otherwise.}
    \end{array}
    \right.  ,
    \end{multline}
    where $\Gamma^{\mathcal{N}^{\theta}}_{RB}$ is the Choi operator of the channel $\mathcal{N}^{\theta}_{A\to B}$.
\end{definition}

In Chapter~\ref{ch:single}, we will use the SLD and RLD Fisher information of channel families to establish Cramer--Rao bounds for parameter estimation in the sequential setting.

\section{Multiparameter estimation}

Finally, we generalize the formalism of estimating a single parameter encoded in a probability distribution, quantum state or quantum channel to their analogous multiparameter tasks.

The task of simultaneously estimating multiple parameters is a much more involved task than estimating a single parameter, both in the classical and quantum settings. However, Cramer--Rao bounds can still be constructed. A major difference when it comes to estimating multiple parameters simultaneously is that the figure of merit, which was earlier the mean-squared error (MSE) of the estimator of a single parameter, becomes the covariance matrix of the estimator of multiple parameters. That is, if the goal is to simultaneously estimate $D$ parameters, then the quantity of interest is a $D \times D$ covariance matrix. Secondly, the Fisher information is no longer a scalar quantity. It too, like the covariance matrix, takes on the form of a $D \times D$ matrix.

In the quantum case, an additional complication is that the optimal measurements for each parameter may not be compatible. Quantum multiparameter estimation has an extensive literature and a number of important recent results \cite{Helstrom1967, Holevo1972, Yuen1973, Belavkin1976, Bagan2006, Holevo2011, Monras2011, Humphreys2013, Yue2015, Ragy2016, Sidhu2019, Albarelli2019, Tsang2019a, Yang2019, Albarelli2020, DemkowiczDobrzanski2020, Friel2020, Gorecki2020}. See \cite{Szczykulska2016, Albarelli2020a} for recent reviews on multiparameter estimation in the quantum setting.

Consider that the $D$ parameters that need to be estimated are encoded in a vector $\bm{\theta} \coloneqq [\theta_1 ~ \theta_2 ~ \cdots ~ \theta_D]^T$. In the classical case, the parameterized family of probability distributions of interest is $\{ p_{\bm{\theta}}\}_{\bm{\theta}}$. As we did in the case of single parameter estimation, we will assume that the estimators used are unbiased. The covariance matrix is given by
\begin{equation}
    \text{Cov}(\bm{\theta}) := \mathbb{E} \left[ (\widehat{\bm{\theta}} - \bm{\theta}) ( \widehat{\bm{\theta}} - \bm{\theta})^T \right] \label{eq:cov-matrix-def}
\end{equation}
where $\widehat{\bm{\theta}}$ is an unbiased estimator for parameters in $\bm{\theta}$. 

In this thesis, we limit the study of multiparameter estimation to multiparameter Cramer--Rao bounds involving the RLD Fisher information.

\subsection{Multiparameter estimation of quantum states}

Suppose that there are $D$ parameters to be estimated, which are encoded in the vector $\bm{\theta} \coloneqq [\theta_1 ~ \theta_2 ~ \cdots ~ \theta_D]^T$. Also suppose that we have a differentiable family of quantum states $\{ \rho_{\bm{\theta}} \}_{\bm{\theta}}$ in which the unknown parameters are encoded. We now proceed to define the RLD Fisher information matrix for this parameterized family of quantum states.

\begin{definition}[RLD Fisher information matrix of quantum states]
    Let $\{ \rho_{\bm{\theta}} \}_{\bm{\theta}}$ be a differential family of quantum states. The matrix elements of its RLD Fisher information matrix are defined as follows:
  \begin{equation}
    [\widehat{I}_F(\bm{\theta};\{\rho_{\bm{\theta}}\}_{\bm{\theta}})]_{j,k}\coloneqq \left\{ 
    \begin{array}
    [c]{cc} 
    \operatorname{Tr}[(\partial_{\theta_{j}}\rho_{\bm{\theta}})\rho_{\bm{\theta}}^{-1}(\partial_{\theta_{k}}\rho_{\bm{\theta}})] & \text{ if } (\partial_{\theta_{j}}\rho_{\bm{\theta}})
(\partial_{\theta_{k}}\rho_{\bm{\theta}})\Pi_{\rho_{\bm{\theta}}}^{\perp}
  =0 \\
    + \infty & \text{otherwise}
    \end{array}
    \right. \label{eq:RLD-FI-matrix-elements-def}
    \end{equation} 
    where $\Pi^{\perp}_{\rho_{\bm{\theta}}}$ denotes the projection onto the kernel of $\rho_{\bm{\theta}}$.

    Alternatively, the RLD Fisher information matrix is defined as follows:
    \begin{align}
    \widehat{I}_F(\bm{\theta};\{\rho_{\bm{\theta}}\}_{\bm{\theta}}) &
    \coloneqq \sum_{j,k=1}^{D}\operatorname{Tr}[(\partial_{\theta_{j}}\rho_{\bm{\theta }}%
    )\rho_{\bm{\theta}}^{-1}(\partial_{\theta_{k}}\rho_{\bm{\theta}}%
    )]|j\rangle\!\langle k|\\
    &  =\operatorname{Tr}_{2}\!\left[  \sum_{j,k=1}^{D}|j\rangle\!\langle
    k|\otimes(\partial_{\theta_{j}}\rho_{\bm{\theta}})\rho_{\bm{\theta}}%
    ^{-1}(\partial_{\theta_{k}}\rho_{\bm{\theta}})\right]  ,
    \end{align}
    where $\Tr_2[\dots]$ refers to tracing over the second subsystem.
\end{definition}

The RLD Fisher information matrix can then be used to establish an operator Cramer--Rao bound on the covariance matrix of any unbiased estimator of the parameters $\bm{\theta}$ \cite{Yuen1973}:
\begin{equation}
    \text{Cov}(\bm{\theta}) \geq \widehat{I}_F(\bm{\theta}; \{\rho_{\bm{\theta}}\}_{\bm{\theta}})^{-1}
    \label{eq:matrix-CRB-1}
\end{equation}
where $\text{Cov} (\bm{\theta})$ is the covariance matrix as defined in~\eqref{eq:cov-matrix-def}.

For evaluating the efficacy of an estimation strategy, it may be more convenient to have a single scalar Cramer--Rao bound than to use the matrix inequality provided above. Our approach to do so involves defining a scalar quantity from the RLD Fisher information matrix. This quantity is called the RLD Fisher information value, and it is defined using the help of a weight matrix $W$. The matrix $W$ should be positive semidefinite and have unit trace.

\begin{definition}[RLD Fisher information value of quantum states]
    Let $\{ \rho_{\bm{\theta}} \}_{\bm{\theta}}$ be a differential family of quantum states and let $W$ be a positive semidefinite weight matrix with unit trace. 
    If the finiteness condition
    \begin{equation}
    \left[  \sum_{j,k=1}^{D}\langle k|W|j\rangle(\partial_{\theta_{k}}%
    \rho_{\bm{\theta}})(\partial_{\theta_{j}}\rho_{\bm{\theta}})\right]  \Pi
    _{\rho_{\bm{\theta}}}^{\perp}  =0, \label{eq:app:state-RLD-finiteness-cond} 
    \end{equation}
    where $\Pi_{\rho_{\bm{\theta}}}^{\perp}$ is the projection onto the kernel of $\rho_{\bm{\theta}}$, holds, 
    the RLD\ Fisher information value is defined as%
    \begin{align}
    \widehat{I}_F(\bm{\theta},W;\{\rho_{\bm{\theta}}\}_{\bm{\theta}}) &
    \coloneqq \operatorname{Tr}[W\widehat{I}_F(\bm{\theta};\{\rho_{\bm{\theta}}\}_{\bm{\theta}})] \label{eq:states-rld-fisher-value} \\ 
    & = \operatorname{Tr}\!\left[  (W\otimes I_{d})\left(  \sum_{j,k=1}^{D}%
    |j\rangle\!\langle k|\otimes(\partial_{\theta_{j}}\rho_{\bm{\theta}}%
    )\rho_{\bm{\theta}}^{-1}(\partial_{\theta_{k}}\rho_{\bm{\theta}})\right)
    \right]  \\
    &  =\sum_{j,k=1}^{D}\langle k|W|j\rangle\operatorname{Tr}\!\left[
    (\partial_{\theta_{j}}\rho_{\bm{\theta}})\rho_{\bm{\theta}}^{-1}%
    (\partial_{\theta_{k}}\rho_{\bm{\theta}})\right]  .
    \end{align}
\end{definition}

In Chapter~\ref{ch:multi}, we show how to use the RLD Fisher information value to establish scalar Cramer--Rao bounds for multiparameter quantum state estimation.

\subsection{Multiparameter estimation of quantum channels}

As we did for estimation of a single parameter, we extend the framework of simultaneously estimating multiple parameters of a quantum state family to the case of quantum channels. The first step is to define the RLD Fisher information value of quantum channels.

\begin{definition}[RLD Fisher information value of quantum channels]
    For a differentiable family of quantum channels $\{\mathcal{N}_{A\rightarrow B}^{\bm{\theta}}\}_{\bm{\theta}}$, if the finiteness condition
    \begin{equation}
        \left[  \sum_{j,k=1}^{D}\langle k|W|j\rangle(\partial_{\theta_{k}}\Gamma
        _{RB}^{\mathcal{N}^{\bm{\theta}}})(\partial_{\theta_{j}}\Gamma_{RB}%
        ^{\mathcal{N}^{\bm{\theta}}})\right]  \Pi_{\Gamma^{\mathcal{N}^{\bm{\theta}}}%
        }^{\perp}  =0,
        \label{eq:app:channel-RLD-finiteness-cond}
    \end{equation}
    holds, the RLD Fisher information value is defined as
    \begin{equation}
        \widehat{I}_F(\bm{\theta},W;\{\mathcal{N}_{A\rightarrow B}^{\bm{\theta}} \}_{\bm{\theta}})\coloneqq \sup_{\rho_{RA}}\widehat{I}_F(\bm{\theta },W;\{\mathcal{N}_{A\rightarrow B}^{\bm{\theta}}(\rho_{RA})\}_{\bm{\theta}}), \label{eq:rld-value-def}
    \end{equation}    
    where the optimization is with respect to every bipartite state $\rho_{RA}$ with system $R$ arbitrarily large. However, note that, by a standard argument, it suffices to optimize over pure states $\psi_{RA}$ with system $R$ isomorphic to the channel input system $A$.

    If the finiteness condition in~\eqref{eq:app:channel-RLD-finiteness-cond} does not hold, the quantity evalutes to $+ \infty$. In the above, $\Pi_{\Gamma^{\mathcal{N}^{\bm{\theta}}}}^{\perp}$ is the projection onto the kernel of $\Gamma^{\mathcal{N}^{\bm{\theta}}}_{RB}$, with $\Gamma^{\mathcal{N}^{\bm{\theta}}}_{RB}$ the Choi operator of the channel $\mathcal{N}^{\bm{\theta}}_{A\to B}$.
\end{definition}

Further, the RLD Fisher information value of quantum channels has the following explicit form: 
\begin{proposition}
\label{prop:geo-fish-explicit-formula-1st-order}Let $\{\mathcal{N}%
_{A\rightarrow B}^{\bm{\theta}}\}_{\bm{\theta}}$ be a differentiable family of
quantum channels, and let $W$ be a $D\times D$ weight matrix. Suppose that the finiteness condition \eqref{eq:app:channel-RLD-finiteness-cond} holds. Then the 
RLD\ Fisher information value of quantum channels has the following explicit form:%
\begin{equation}
\widehat{I}_{F}(\bm{\theta},W;\{\mathcal{N}_{A\rightarrow B}^{\bm{\theta}}%
\}_{\bm{\theta}})=\left\Vert \sum_{j,k=1}^{D}\langle k|W|j\rangle
\operatorname{Tr}_{B}[(\partial_{\theta_{j}}\Gamma_{RB}^{\mathcal{N}%
^{\bm{\theta}}})(\Gamma_{RB}^{\mathcal{N}^{\bm{\theta}}})^{-1}(\partial
_{\theta_{k}}\Gamma_{RB}^{\mathcal{N}^{\bm{\theta}}})]\right\Vert _{\infty}.
\label{eq:RLD-value-channels-inf-norm}
\end{equation}
\end{proposition}

\begin{proof}
Recall that every pure state $\psi_{RA}$ can be written as%
\begin{equation}
\psi_{RA}=Z_{R}\Gamma_{RA}Z_{R}^{\dag},
\end{equation}
where $Z_{R}$ is a square operator satisfying $\operatorname{Tr}[Z_{R}^{\dag
}Z_{R}]=1$. This implies that%
\begin{align}
\mathcal{N}_{A\rightarrow B}^{\bm{\theta}}(\psi_{RA}) &  =\mathcal{N}%
_{A\rightarrow B}^{\bm{\theta}}(Z_{R}\Gamma_{RA}Z_{R}^{\dag})\\
&  =Z_{R}\mathcal{N}_{A\rightarrow B}^{\bm{\theta}}(\Gamma_{RA})Z_{R}^{\dag}\\
&  =Z_{R}\Gamma_{RB}^{\mathcal{N}^{\bm{\theta}}}Z_{R}^{\dag}.
\end{align}
It suffices to optimize over pure states $\psi_{RA}$ such that $\psi_{A}>0$
because these states are dense in the set of all pure bipartite states. Then
consider that%
\begin{align}
&  \sup_{\psi_{RA}}\widehat{I}_{F}(\bm{\theta},W;\{\mathcal{N}_{A\rightarrow
B}^{\bm{\theta}}(\psi_{RA})\}_{\bm{\theta}})\nonumber\\
&  =\sup_{\psi_{RA}}\sum_{j,k=1}^{D}\langle k|W|j\rangle\operatorname{Tr}%
\left[  (\partial_{\theta_{j}}\mathcal{N}_{A\rightarrow B}^{\bm{\theta}}%
(\psi_{RA}))(\mathcal{N}_{A\rightarrow B}^{\bm{\theta}}(\psi_{RA}%
))^{-1}(\partial_{\theta_{k}}\mathcal{N}_{A\rightarrow B}^{\bm{\theta}}%
(\psi_{RA}))\right]  \\
&  =\sup_{Z_{R}:\operatorname{Tr}[Z_{R}^{\dag}Z_{R}]=1}\sum_{j,k=1}^{D}\langle
k|W|j\rangle\operatorname{Tr}\!\left[  (\partial_{\theta_{j}}Z_{R}\Gamma
_{RB}^{\mathcal{N}^{\bm{\theta}}}Z_{R}^{\dag})(Z_{R}\Gamma_{RB}^{\mathcal{N}%
^{\bm{\theta}}}Z_{R}^{\dag})^{-1}(\partial_{\theta_{k}}Z_{R}\Gamma
_{RB}^{\mathcal{N}^{\bm{\theta}}}Z_{R}^{\dag})\right]  \\
&  =\sup_{Z_{R}:\operatorname{Tr}[Z_{R}^{\dag}Z_{R}]=1}\sum_{j,k=1}^{D}\langle
k|W|j\rangle\operatorname{Tr}\!\left[  Z_{R}(\partial_{\theta_{j}}\Gamma
_{RB}^{\mathcal{N}^{\bm{\theta}}})Z_{R}^{\dag}Z_{R}^{-\dag}(\Gamma
_{RB}^{\mathcal{N}^{\bm{\theta}}})^{-1}Z_{R}^{-1}Z_{R}(\partial_{\theta_{k}%
}\Gamma_{RB}^{\mathcal{N}^{\bm{\theta}}})Z_{R}^{\dag}\right]  \\
&  =\sup_{Z_{R}:\operatorname{Tr}[Z_{R}^{\dag}Z_{R}]=1}\sum_{j,k=1}^{D}\langle
k|W|j\rangle\operatorname{Tr}\!\left[  Z_{R}^{\dag}Z_{R}(\partial_{\theta_{j}%
}\Gamma_{RB}^{\mathcal{N}^{\bm{\theta}}})(\Gamma_{RB}^{\mathcal{N}%
^{\bm{\theta}}})^{-1}(\partial_{\theta_{k}}\Gamma_{RB}^{\mathcal{N}%
^{\bm{\theta}}})\right]  \\
&  =\sup_{Z_{R}:\operatorname{Tr}[Z_{R}^{\dag}Z_{R}]=1}\operatorname{Tr}\!
\left[  Z_{R}^{\dag}Z_{R}\sum_{j,k=1}^{D}\langle k|W|j\rangle\operatorname{Tr}%
_{B}\!\left[  (\partial_{\theta_{j}}\Gamma_{RB}^{\mathcal{N}^{\bm{\theta}}%
})(\Gamma_{RB}^{\mathcal{N}^{\bm{\theta}}})^{-1}(\partial_{\theta_{k}}%
\Gamma_{RB}^{\mathcal{N}^{\bm{\theta}}})\right]  \right]  \\
&  =\left\Vert \sum_{j,k=1}^{D}\langle k|W|j\rangle\operatorname{Tr}%
_{B}\!\left[  (\partial_{\theta_{j}}\Gamma_{RB}^{\mathcal{N}^{\bm{\theta}}%
})(\Gamma_{RB}^{\mathcal{N}^{\bm{\theta}}})^{-1}(\partial_{\theta_{k}}%
\Gamma_{RB}^{\mathcal{N}^{\bm{\theta}}})\right]  \right\Vert _{\infty}.%
\end{align}
The last equality is a consequence of the characterization of the infinity
norm of a positive semi-definite operator $Y$ as $\left\Vert Y\right\Vert
_{\infty}=\sup_{\rho>0,\operatorname{Tr}[\rho]=1}\operatorname{Tr}[Y\rho]$.
\end{proof}


\pagebreak

\chapter{Limits on Single Parameter Estimation of Quantum Channels}\label{ch:single}
\graphicspath{{figures/}}
\allowdisplaybreaks

\vspace{0.5em}


In this chapter, we present our results for estimating a single parameter encoded in an unknown quantum channel. Using the machinery developed in Chapter~\ref{ch:prelims}, we establish Cramer--Rao bounds which place limits on the variance of an unbiased estimator for quantum channel estimation in the sequential setting. 

Our first step will be to define the amortized Fisher information of quantum channels. Amortization involves allowing for a catalyst state family to increase the Fisher information of the channel family in question, while also subtracting off the Fisher information of the catalyst state family itself. It is inspired by the notion of amortized channel divergence, introduced in ~\cite{Berta2018}, which has been useful to study the power of sequential strategies when processing quantum channels for a variety of distinguishabililty tasks. In particular, it has been used in the analysis of feedback-assisted or sequential protocols in other areas of quantum information science~\cite{Bennett2003,BenDana2017,Rigovacca2018,Kaur2018,Berta2018a,Das2019,Wang2019a,Fang2019,Wang2019b}. Thus, we use the amortized Fisher information of channels to probe the power and limitations of channel estimation in the sequential setting.

The amortized Fisher information is defined for any generalized Fisher information of quantum states and channels, which we defined in Chapter~\ref{ch:prelims}; i.e., any Fisher information that obeys the data-processing inequality. For certain special cases, the amortized Fisher information in question undergoes what is known as an ``amortization collapse''. This, in simple language, means that amortization does not increase the Fisher information undergoing the collapse, and that the amortized Fisher information is strictly equal to the Fisher information itself. We show how such an amortization collapse occurs for the SLD Fisher information of classical-quantum channels, for the root SLD Fisher information of general quantum channels, and also for the RLD Fisher information of general quantum channels.

Amortization collapses are useful from both a qualitative and technical viewpoint. Qualitatively, they can be understood as the fact that catalysis with an ancillary state family cannot increase the Fisher information of the channel family in question. Technically, they mean that in a sequential estimation strategy, the Fisher information can only increase linearly with the number of channel uses. Further, they also mean that with respect to the quantity that undergoes the amortization collapse, parallel strategies are just as good as (the more general) sequential ones. Both of these conclusions arise by connecting the amortized Fisher information to the performance of a sequential estimation protocol, which we do by proving a meta-converse theorem in this chapter. 

Finally, after establishing the amortization collapses we just mentioned and the connection between amortized Fisher information and sequential estimation protocols, we are able to establish Cramer--Rao bounds for channel estimation in the sequential setting. We derive the following three Cramer--Rao bounds:
\begin{itemize}
    \item a Cramer--Rao bound in \eqref{eq:cq-channels-crb} for estimation of classical-quantum channels using the SLD Fisher information,
    \item a Cramer--Rao bound in \eqref{eq:Heisenberg-bnd-SLD-channels} for estimation of general quantum channels using the SLD Fisher information, which recovers the Heisenberg scaling limit of estimation of unitary channels, and
    \item a Cramer--Rao bound in \eqref{eq:RLD-bnd-ch} for estimation of general quantum channels using the RLD Fisher information, which has the important corollary that if the RLD Fisher information of a channel family is finite, then Heisenberg scaling with respect to the number of channel uses is unattainable. In other words, the finiteness condition for the RLD Fisher information provides a no-go for Heisenberg scaling.
\end{itemize}

Our bounds have a number of desirable characteristics, namely that they are
\begin{itemize}
    \item single-letter; i.e., computing them requires computing the Fisher information in question for a single channel use only, even though the bounds are applicable for $n$-round sequential procotols,
    \item universally applicable, in the sense that our root SLD Fisher information and RLD Fisher information bounds apply to all quantum channels, and thus encompass all admissible quantum dynamics, and
    \item computable via optimization problems. In particular, the RLD Fisher information bound for quantum channels admits a semi-definite program representation. 
\end{itemize}

We evaluate our RLD-based Cramer--Rao bound for the task of estimating the loss and noise parameters of a generalized amplitude damping channel and compare it with an achievable SLD-based Cramer--Rao bound. Lastly, we provide optimization problem formulations for the following quantities:
\begin{itemize}
    \item semi-definite program for the SLD Fisher information of quantum states,
    \item quadratically constrained optimization problem for root SLD Fisher information of quantum states,
    \item semi-definite program for the RLD Fisher information of quantum states,
    \item semi-definite program for the RLD Fisher information of quantum channels, and
    \item bilinear program for SLD Fisher information of quantum channels.
\end{itemize}   

\section{Amortized Fisher information}


The amortized Fisher information, like the amortized channel divergence, allows for deeper study of the power of sequential strategies for various channel processing tasks. The amortized Fisher information specifically allows for us to quantitatively evaluate the difference in power between sequential and parallel strategies for quantum channel estimation. We define it below:
\begin{definition}
[Amortized Fisher information of quantum channels]The amortized Fisher
information of a family $\{\mathcal{N}_{A\rightarrow B}^{\theta}\}_{\theta}$
of quantum channels is defined as follows:
\begin{equation}
\mathbf{I}_{F}^{\mathcal{A}}(\theta;\{\mathcal{N}_{A\rightarrow B}^{\theta
}\}_{\theta}):=\sup_{\{\rho_{RA}^{\theta}\}_{\theta}}\left[  \mathbf{I}
_{F}(\theta;\{\mathcal{N}_{A\rightarrow B}^{\theta}(\rho_{RA}^{\theta
})\}_{\theta})-\mathbf{I}_{F}(\theta;\{\rho_{RA}^{\theta})\}_{\theta})\right]
, \label{eq:amortized-fisher-info}
\end{equation}
where $\mathbf{I}_F$ is the generalized Fisher information as defined in~\eqref{eq:gen-fisher-states} and the supremum is with respect to arbitrary state families $\{\rho
_{RA}^{\theta}\}_{\theta}$\ with unbounded reference system $R$.
\end{definition}

We allow for a resource at the channel input in order to help with the estimation task, but then we subtract off the value of this resource in order to account for the amount of resource that is strictly present in the channel family. Furthermore, the presence of the state $\rho_{RA}^{\theta}$ whose Fisher information has been accounted for opens up the possibility that the state can, in some way, catalyze the Fisher information of the channel. We should indicate here that the amortized channel divergence of \cite{Berta2018} is a special case of the amortized Fisher information in which the parameter $\theta$ takes on only two values. We also remark that discriminating two quantum channels is a special case of channel estimation where the parameter $\theta$ takes on two values.

\begin{figure}[ptb]
\centering
\includegraphics[width=\linewidth]{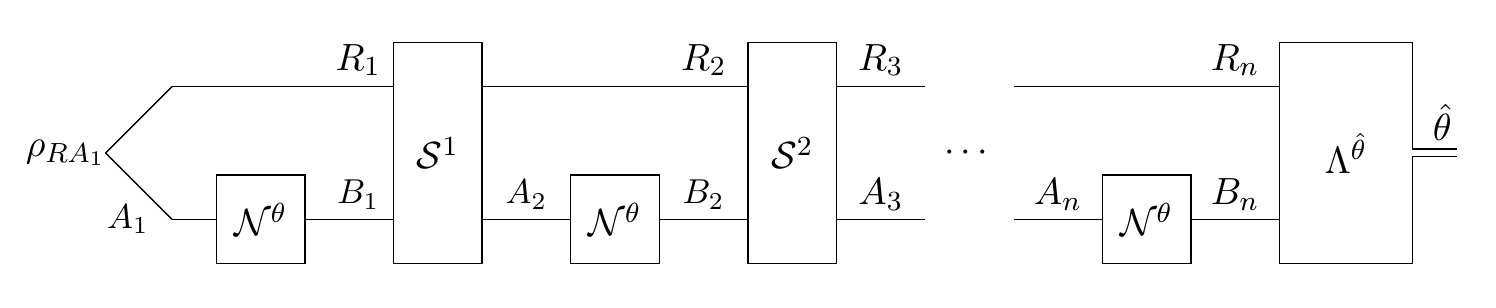}
\caption{{Processing $n$ uses of channel $\mathcal{N}^{\theta}$ in a sequential manner.}}{Processing $n$ uses of channel $\mathcal{N}^{\theta}$ in a sequential manner is the most general approach to channel parameter estimation or discrimination. The $n$ uses of the channel are interleaved with $n-1$ quantum channels $\mathcal{S}^{1}$ through $\mathcal{S}^{n-1}$, which can also share memory systems with each other. The final measurement's outcome is then used to obtain an estimate of the unknown parameter $\theta$.}
\label{fig:sequential-protocol-chapter3}
\end{figure}

For the sake of easy reference, we reproduce in Figure~\ref{fig:sequential-protocol-chapter3} the graphical depiction of a sequential estimation strategy from Chapter~\ref{ch:prelims}. An alternate way to understand the motivation behind introducing the amortized Fisher information is to consider that the goal of sequential channel estimation is to ``accumulate'' as much information about the parameter $\theta$ into the state being carried forward from one channel use to the other in a sequential estimation protocol, as shown in Figure~\ref{fig:sequential-protocol-chapter3}. The amortized Fisher information captures the marginal increase in Fisher information per channel use in such a scenario.

With the aid of the above qualitative reasoning and motivation, it may be simple to see that providing a catalyst state family and subtracting its Fisher information from the total, can never decrease the Fisher information of the channel in question; i.e., amortization does not decrease Fisher information. 

\begin{proposition}
\label{prop:amort->=-ch-Fish-gen}Let $\{\mathcal{N}_{A\rightarrow B}^{\theta
}\}_{\theta}$ be a family of quantum channels, and suppose that the underlying
generalized Fisher information is weakly faithful. Then the generalized Fisher
information does not exceed the amortized one:
\begin{equation}
\mathbf{I}_{F}^{\mathcal{A}}(\theta;\{\mathcal{N}_{A\rightarrow B}^{\theta
}\}_{\theta})\geq\mathbf{I}_{F}(\theta;\{\mathcal{N}_{A\rightarrow B}^{\theta
}\}_{\theta}). \label{eq:general-amort-ineq-obvi-dir}
\end{equation}

\end{proposition}

\begin{proof}
This follows because we can always pick the input family $\{\rho_{RA}^{\theta
}\}_{\theta}$ in \eqref{eq:amortized-fisher-info} to have no dependence on the
parameter $\theta$. Then we find that
\begin{align}
\mathbf{I}_{F}^{\mathcal{A}}(\theta;\{\mathcal{N}_{A\rightarrow B}^{\theta
}\}_{\theta})  &  \geq\mathbf{I}_{F}(\theta;\{\mathcal{N}_{A\rightarrow
B}^{\theta}(\rho_{RA})\}_{\theta})-\mathbf{I}_{F}(\theta;\{\rho_{RA}
)\}_{\theta})\\
&  =\mathbf{I}_{F}(\theta;\{\mathcal{N}_{A\rightarrow B}^{\theta}(\rho
_{RA})\}_{\theta}),
\end{align}
where we applied the weak faithfulness assumption to arrive at the equality.
Since the inequality holds for all input states $\rho_{RA}$, we conclude \eqref{eq:general-amort-ineq-obvi-dir}.
\end{proof}

\subsection{Amortization collapse}

For some particular choices of the generalized Fisher information, the
inequality in \eqref{eq:general-amort-ineq-obvi-dir}\ can be reversed, which is called an ``amortization collapse.'' The meaning of an amortization collapse for a Fisher information quantity, as we stated earlier as well, is that the Fisher information in question cannot be increased by using a catalyst. It also means that in a sequential channel estimation task, the quantity that undergoes an amortization collapse increases linearly, at best, with the number of channel uses.  

Later in the chapter, we prove Theorem~\ref{thm:meta-converse}, a meta-converse which connects sequential estimation to the amortized Fisher information. This makes an amortization collapse useful for establishing limits on the performance of sequential estimation protocols. For differentiable families $\{\mathcal{N}_{X\rightarrow B}^{\theta}\}_{\theta}$\ of classical--quantum channels, the following equality holds for the SLD Fisher information:
\begin{equation}
I_{F}^{\mathcal{A}}(\theta;\{\mathcal{N}_{X\rightarrow B}^{\theta}\}_{\theta
})=I_{F}(\theta;\{\mathcal{N}_{X\rightarrow B}^{\theta}\}_{\theta}). \label{eq:sld-cq-amort-collapse-intro}
\end{equation}

Further, we show that the following equalities hold for the root SLD and the RLD Fisher informations for all differentiable families $\{\mathcal{N}_{A\rightarrow B}^{\theta}\}_{\theta}$\ of quantum channels:
\begin{align}
\sqrt{I_{F}}^{\mathcal{A}}(\theta;\{\mathcal{N}_{A\rightarrow B}^{\theta
}\}_{\theta}) & =\sqrt{I_{F}}(\theta;\{\mathcal{N}_{A\rightarrow B}^{\theta
}\}_{\theta}), \label{eq:root-sld-amort-collapse-intro}
\\
\widehat{I}_{F}^{\mathcal{A}}(\theta;\{\mathcal{N}_{A\rightarrow B}^{\theta
}\}_{\theta}) & =\widehat{I}_{F}(\theta;\{\mathcal{N}_{A\rightarrow B}^{\theta
}\}_{\theta}). \label{eq:rld-amort-collapse-intro}
\end{align}

We will now provide explicit details on how we arrive at the particular amortization collapses stated above in \eqref{eq:sld-cq-amort-collapse-intro}, \eqref{eq:root-sld-amort-collapse-intro}, and \eqref{eq:rld-amort-collapse-intro}.

\subsection{Amortization collapse of SLD Fisher information for classical-quantum channels}

We first consider the special case of a family $\{\mathcal{N}_{X\rightarrow
B}^{\theta}\}_{\theta}$ of classical--quantum channels of the following form:
\begin{equation}
\mathcal{N}_{X\rightarrow B}^{\theta}(\sigma_{X}):=\sum_{x}\langle
x|_{X}\sigma_{X}|x\rangle_{X}\omega_{B}^{x,\theta}, \label{eq:cq-channel-fams}
\end{equation}
where $\{|x\rangle\}_{x}$ is an orthonormal basis and $\{\omega_{B}^{x,\theta
}\}_{x}$ is a collection of states prepared at the channel output conditioned
on the value of the unknown parameter $\theta$ and on the result of the
measurement of the channel input. The key aspect of these channels is that the
measurement at the input is the same regardless of the value of the parameter
$\theta$. We find the following amortization collapse for these channels:

\begin{proposition}
\label{thm:amort-collapse-cq}Let $\{\mathcal{N}_{X\rightarrow B}^{\theta
}\}_{\theta}$ be a family of differentiable classical--quantum channels. Then
the following amortization collapse occurs
\begin{equation}
I_{F}(\theta;\{\mathcal{N}_{X\rightarrow B}^{\theta}\}_{\theta})=I_{F}
^{\mathcal{A}}(\theta;\{\mathcal{N}_{X\rightarrow B}^{\theta}\}_{\theta}
)=\sup_{x}I_{F}(\theta;\{\omega_{B}^{x,\theta}\}_{\theta}).
\label{eq:amort-collapse}
\end{equation}

\end{proposition}

\begin{proof}
If the finiteness condition in
\eqref{eq:finiteness-condition-SLD-fish-ch}\ does not hold, then all
quantities are trivially equal to $+\infty$. So let us suppose that the
finiteness condition in \eqref{eq:finiteness-condition-SLD-fish-ch}\ holds.
Note that the finiteness condition is equivalent to
\begin{equation}
\Pi_{\omega_{B}^{x,\theta}}^{\perp}(\partial_{\theta}\omega_{B}^{x,\theta}
)\Pi_{\omega_{B}^{x,\theta}}^{\perp}=0\qquad\forall x.
\label{eq:finiteness-condition-cq-channels}
\end{equation}

First, consider that the following inequality holds
\begin{equation}
I_{F}(\theta;\{\mathcal{N}_{X\rightarrow B}^{\theta}\}_{\theta})\geq\sup
_{x}I_{F}(\theta;\{\omega_{B}^{x,\theta}\}_{\theta}) \label{eq:simple-way-cq}
\end{equation}
because we can input the state $|x\rangle\!\langle x|_{X}$ to the channel
$\mathcal{N}_{X\rightarrow B}^{\theta}$ and obtain the output state
$\mathcal{N}_{X\rightarrow B}^{\theta}(|x\rangle\!\langle x|_{X})=\omega
_{B}^{x,\theta}$. Then we can optimize over $x\in\mathcal{X}$ and obtain the
bound above.

We now prove the less trivial inequality
\begin{equation}
I_{F}^{\mathcal{A}}(\theta;\{\mathcal{N}_{X\rightarrow B}^{\theta}\}_{\theta
})\leq\sup_{x}I_{F}(\theta;\{\omega_{B}^{x,\theta}\}_{\theta}).
\label{eq:amort-dont-help-cq}
\end{equation}
Let $\{\rho_{RA}^{\theta}\}_{\theta}$ be a differentiable family of quantum
states. If the classical--quantum channel $\mathcal{N}_{X\rightarrow
B}^{\theta}$ acts on $\rho_{RA}^{\theta}$ (identifying $X=A$), the output
state is as follows:
\begin{equation}
\mathcal{N}_{X\rightarrow B}^{\theta}(\rho_{RA}^{\theta})=\sum_{x}p_{\theta
}(x)\rho_{R}^{x,\theta}\otimes\omega_{B}^{x,\theta},
\end{equation}
where
\begin{equation}
\rho_{R}^{x,\theta}:=\frac{1}{p_{\theta}(x)}\langle x|_{X}\rho_{RA}^{\theta
}|x\rangle_{X},\qquad p_{\theta}(x):=\operatorname{Tr}[\langle x|_{X}\rho
_{RA}^{\theta}|x\rangle_{X}].
\end{equation}
Then consider that
\begin{align}
&  I_{F}(\theta;\{\mathcal{N}_{X\rightarrow B}^{\theta}(\rho_{RA}^{\theta
})\}_{\theta})\nonumber\\
&  =I_{F}\!\left(  \theta;\left\{  \sum_{x}p_{\theta}(x)\rho_{R}^{x,\theta
}\otimes\omega_{B}^{x,\theta}\right\}  _{\theta}\right) \\
&  \leq I_{F}\!\left(  \theta;\left\{  \sum_{x}p_{\theta}(x)|x\rangle\!\langle
x|_{X}\otimes\rho_{R}^{x,\theta}\otimes\omega_{B}^{x,\theta}\right\}
_{\theta}\right) \\
&  =I_{F}(\theta;\{p_{\theta}\}_{\theta})+\sum_{x}p_{\theta}(x)I_{F}
(\theta;\{\rho_{R}^{x,\theta}\otimes\omega_{B}^{x,\theta}\}_{\theta})\\
&  =I_{F}(\theta;\{p_{\theta}\}_{\theta})+\sum_{x}p_{\theta}(x)I_{F}
(\theta;\{\rho_{R}^{x,\theta}\}_{\theta})+\sum_{x}p_{\theta}(x)I_{F}
(\theta;\{\omega_{B}^{x,\theta}\}_{\theta})\\
&  \leq I_{F}(\theta;\{p_{\theta}\}_{\theta})+\sum_{x}p_{\theta}
(x)I_{F}(\theta;\{\rho_{R}^{x,\theta}\}_{\theta})+\sup_{x}I_{F}(\theta
;\{\omega_{B}^{x,\theta}\}_{\theta})\\
&  =I_{F}\!\left(  \theta;\left\{  \sum_{x}p_{\theta}(x)|x\rangle\!\langle
x|_{X}\otimes\rho_{R}^{x,\theta}\right\}  _{\theta}\right)  +\sup_{x}
I_{F}(\theta;\{\omega_{B}^{x,\theta}\}_{\theta})\\
&  \leq I_{F}(\theta;\{\rho_{RA}^{\theta}\}_{\theta})+\sup_{x}I_{F}
(\theta;\{\omega_{B}^{x,\theta}\}_{\theta}).
\end{align}
The first inequality follows from the data-processing inequality for Fisher
information with respect to partial trace over the $X$ system. The second
equality follows from Proposition~\ref{prop:cq-decomp-SLD-RLD}. The third
equality follows from the additivity of SLD\ Fisher information for product
states (Proposition~\ref{prop:additivity-SLD-RLD-states}). The second
inequality follows from the fact that the average cannot exceed the maximum.
The last equality follows again from Proposition~\ref{prop:cq-decomp-SLD-RLD}.
The final inequality follows from the data-processing inequality under the
action of the measurement channel $(\cdot)\rightarrow\sum_{x}|x\rangle\!\langle
x|_{X}(\cdot)|x\rangle\!\langle x|_{X}$ on the state $\rho_{RA}$. Thus, the
following inequality holds for an arbitrary family $\{\rho_{RA}^{\theta
}\}_{\theta}$ of states:
\begin{equation}
I_{F}(\theta;\{\mathcal{N}_{X\rightarrow B}^{\theta}(\rho_{RA}^{\theta
})\}_{\theta})-I_{F}(\theta;\{\rho_{RA}^{\theta}\}_{\theta})\leq\sup_{x}
I_{F}(\theta;\{\omega_{B}^{x,\theta}\}_{\theta}).
\label{eq:amort-dont-help-cq-2}
\end{equation}
Since the inequality in \eqref{eq:amort-dont-help-cq-2} holds for an arbitrary
family $\{\rho_{RA}^{\theta}\}_{\theta}$ of states, we conclude
\eqref{eq:amort-dont-help-cq}. Combining \eqref{eq:simple-way-cq} and
\eqref{eq:amort-dont-help-cq}, along with the general inequality in
\eqref{eq:general-amort-ineq-obvi-dir}, we conclude \eqref{eq:amort-collapse}.
\end{proof}

\subsection{Amortization collapse of root SLD Fisher information for general channels} \label{subsec:root-sld-amort-collapse}

Next, we show an amortization collapse for the square root of the SLD Fisher information for all quantum channels. We begin by showing that the root SLD\ Fisher information obeys the following chain rule inequality and then show how the amortization collapse follows as a corollary of it.

\begin{proposition}
[Chain rule for root SLD Fisher information of quantum channels]\label{prop:chain-rule-root-SLD}Let $\{\rho_{\theta}\}_{\theta}$
be a differentiable family of quantum states, and let $\{\mathcal{N}
_{A\rightarrow B}^{\theta}\}_{\theta}$ be a differentiable family of quantum
channels. Then the following chain rule holds for the root SLD Fisher
information:
\begin{equation}
\sqrt{I_{F}}(\theta;\{\mathcal{N}_{A\rightarrow B}^{\theta}(\rho_{RA}^{\theta
})\}_{\theta})\leq\sqrt{I_{F}}(\theta;\{\mathcal{N}_{A\rightarrow B}^{\theta
}\}_{\theta})+\sqrt{I_{F}}(\theta;\{\rho_{RA}^{\theta}\}_{\theta}).
\label{eq:chain-rule-root-SLD}
\end{equation}

\end{proposition}

\begin{proof}
If the finiteness conditions in
\eqref{eq:basis-independent-formula-SLD}\ and\ \eqref{eq:finiteness-condition-SLD-fish-ch}
do not hold, then the inequality is trivially satisfied. So let us suppose
that the finiteness conditions
\eqref{eq:basis-independent-formula-SLD}\ and\ \eqref{eq:finiteness-condition-SLD-fish-ch}
hold.

By invoking the variational representation of the root SLD Fisher information (provided later in this chapter as Proposition~\ref{prop:root-SLD-opt-formula}) and also Remark~\ref{rem:restrict-to-pure-bipartite}, the root SLD\ Fisher information of channels has the following representation as an optimization:
\begin{align}
&  \frac{1}{\sqrt{2}}\sqrt{I_{F}}(\theta;\{\mathcal{N}_{A\rightarrow
B}^{\theta}\}_{\theta})\nonumber\\
&  =\frac{1}{\sqrt{2}}\sup_{\rho_{RA}}\sqrt{I_{F}}(\theta;\{\mathcal{N}
_{A\rightarrow B}^{\theta}(\rho_{RA})\}_{\theta})\\
&  =\sup_{\rho_{RA}}\sup_{X_{RB}}\left\{
\begin{array}
[c]{c}
\left\vert \operatorname{Tr}[X_{RB}(\partial_{\theta}\mathcal{N}_{A\rightarrow
B}^{\theta}(\rho_{RA}))]\right\vert :\\
\operatorname{Tr}[(X_{RB}X_{RB}^{\dag}+X_{RB}^{\dag}X_{RB})\mathcal{N}
_{A\rightarrow B}^{\theta}(\rho_{RA})]\leq1
\end{array}
\right\} \\
&  =\sup_{\rho_{RA},X_{RB}}\left\{
\begin{array}
[c]{c}
\left\vert \operatorname{Tr}[X_{RB}(\partial_{\theta}\mathcal{N}_{A\rightarrow
B}^{\theta})(\rho_{RA})]\right\vert :\\
\operatorname{Tr}[(X_{RB}X_{RB}^{\dag}+X_{RB}^{\dag}X_{RB})\mathcal{N}
_{A\rightarrow B}^{\theta}(\rho_{RA})]\leq1
\end{array}
\right\}  , \label{eq:root-SLD-opt-formula-channels}
\end{align}
where the distinction between the third and last line is that $\partial
_{\theta}\mathcal{N}_{A\rightarrow B}^{\theta}(\rho_{RA})=(\partial_{\theta
}\mathcal{N}_{A\rightarrow B}^{\theta})(\rho_{RA})$ (i.e., for fixed
$\rho_{RA}$, the state $\rho_{RA}$ is constant with respect to the partial derivative).

Now recall the post-selected teleportation identity from
\eqref{eq:PS-TP-identity}:
\begin{equation}
\mathcal{N}_{A\rightarrow B}^{\theta}(\rho_{RA}^{\theta})=\langle\Gamma
|_{AS}\rho_{RA}^{\theta}\otimes\Gamma_{SB}^{\mathcal{N}^{\theta}}
|\Gamma\rangle_{AS}\text{.}
\end{equation}
This implies that
\begin{align}
&  \partial_{\theta}(\mathcal{N}_{A\rightarrow B}^{\theta}(\rho_{RA}^{\theta
}))\nonumber\\
&  =\partial_{\theta}(\langle\Gamma|_{AS}\rho_{RA}^{\theta}\otimes\Gamma
_{SB}^{\mathcal{N}^{\theta}}|\Gamma\rangle_{AS})\\
&  =\langle\Gamma|_{AS}\partial_{\theta}(\rho_{RA}^{\theta}\otimes\Gamma
_{SB}^{\mathcal{N}^{\theta}})|\Gamma\rangle_{AS}\\
&  =\langle\Gamma|_{AS}[(\partial_{\theta}\rho_{RA}^{\theta})\otimes
\Gamma_{SB}^{\mathcal{N}^{\theta}}+\rho_{RA}^{\theta}\otimes(\partial_{\theta
}\Gamma_{SB}^{\mathcal{N}^{\theta}})]|\Gamma\rangle_{AS}\\
&  =\langle\Gamma|_{AS}[(\partial_{\theta}\rho_{RA}^{\theta})\otimes
\Gamma_{SB}^{\mathcal{N}^{\theta}}|\Gamma\rangle_{AS}+\langle\Gamma|_{AS}
\rho_{RA}^{\theta}\otimes(\partial_{\theta}\Gamma_{SB}^{\mathcal{N}^{\theta}
})|\Gamma\rangle_{AS}\\
&  =\mathcal{N}_{A\rightarrow B}^{\theta}(\partial_{\theta}\rho_{RA}^{\theta
})+(\partial_{\theta}\mathcal{N}_{A\rightarrow B}^{\theta})(\rho_{RA}^{\theta
}).
\end{align}
Let $X_{RB}$ be an arbitrary operator satisfying
\begin{equation}
\operatorname{Tr}[(X_{RB}X_{RB}^{\dag}+X_{RB}^{\dag}X_{RB})\mathcal{N}
_{A\rightarrow B}^{\theta}(\rho_{RA}^{\theta})]\leq1.
\label{eq:root-SLD-chain-arb-op}
\end{equation}
Working with the left-hand side of the inequality, we find that
\begin{align}
&  \operatorname{Tr}[(X_{RB}X_{RB}^{\dag}+X_{RB}^{\dag}X_{RB})\mathcal{N}
_{A\rightarrow B}^{\theta}(\rho_{RA}^{\theta})]\nonumber\\
&  =\operatorname{Tr}[(\mathcal{N}_{A\rightarrow B}^{\theta})^{\dag}
(X_{RB}X_{RB}^{\dag}+X_{RB}^{\dag}X_{RB})\rho_{RA}^{\theta}]\\
&  \geq\operatorname{Tr}[(Z_{RA}Z_{RA}^{\dag}+Z_{RA}^{\dag}Z_{RA})\rho
_{RA}^{\theta}],
\end{align}
where we set
\begin{equation}
Z_{RA}:=(\mathcal{N}_{A\rightarrow B}^{\theta})^{\dag}(X_{RB}).
\end{equation}
The equality follows because $(\mathcal{N}_{A\rightarrow B}^{\theta})^{\dag}$
is the Hilbert--Schmidt adjoint of $\mathcal{N}_{A\rightarrow B}^{\theta}$,
and the inequality follows because $\rho_{RA}^{\theta}\geq 0$ and
\begin{align}
(\mathcal{N}_{A\rightarrow B}^{\theta})^{\dag}(X_{RB}^{\dag})(\mathcal{N}
_{A\rightarrow B}^{\theta})^{\dag}(X_{RB})  &  \leq(\mathcal{N}_{A\rightarrow
B}^{\theta})^{\dag}(X_{RB}^{\dag}X_{RB}),\\
(\mathcal{N}_{A\rightarrow B}^{\theta})^{\dag}(X_{RB})(\mathcal{N}
_{A\rightarrow B}^{\theta})^{\dag}(X_{RB}^{\dag})  &  \leq(\mathcal{N}
_{A\rightarrow B}^{\theta})^{\dag}(X_{RB}X_{RB}^{\dag}),
\end{align}
which themselves follow from the Schwarz inequality for completely positive
unital maps \cite[Eq.~(3.14)]{Bhatia2007}. So we conclude that
\begin{equation}
\operatorname{Tr}[(Z_{RA}Z_{RA}^{\dag}+Z_{RA}^{\dag}Z_{RA})(\rho_{RA}^{\theta
})]\leq1. \label{eq:condition-on-adjoint-op-root-SLD-chain}
\end{equation}
Then consider that
\begin{align}
&  \left\vert \operatorname{Tr}[X_{RB}(\partial_{\theta}(\mathcal{N}
_{A\rightarrow B}^{\theta}(\rho_{RA}^{\theta})))]\right\vert \nonumber\\
&  =\left\vert \operatorname{Tr}[X_{RB}((\partial_{\theta}\mathcal{N}
_{A\rightarrow B}^{\theta})(\rho_{RA}^{\theta}))]+\operatorname{Tr}
[X_{RB}\mathcal{N}_{A\rightarrow B}^{\theta}(\partial_{\theta}
\rho_{RA}^{\theta})]\right\vert \\
&  =\left\vert \operatorname{Tr}[X_{RB}((\partial_{\theta}\mathcal{N}
_{A\rightarrow B}^{\theta})(\rho_{RA}^{\theta}))]+\operatorname{Tr}
[(\mathcal{N}_{A\rightarrow B}^{\theta})^{\dag}(X_{RB})(\partial_{\theta}
\rho_{RA}^{\theta})]\right\vert \\
&  \leq\left\vert \operatorname{Tr}[X_{RB}((\partial_{\theta}\mathcal{N}
_{A\rightarrow B}^{\theta})(\rho_{RA}^{\theta}))]\right\vert +\left\vert
\operatorname{Tr}[(\mathcal{N}_{A\rightarrow B}^{\theta})^{\dag}
(X_{RB})(\partial_{\theta}\rho_{RA}^{\theta})]\right\vert .
\end{align}
By applying the optimization representation \eqref{eq:root-SLD-opt-formula-channels}, we find that
\begin{equation}
\sqrt{2}\left\vert \operatorname{Tr}[X_{RB}((\partial_{\theta}\mathcal{N}
_{A\rightarrow B}^{\theta})(\rho_{RA}^{\theta}))]\right\vert \leq\sqrt{I_{F}
}(\theta;\{\mathcal{N}_{A\rightarrow B}^{\theta}\}_{\theta}).
\end{equation}
Since the operator $(\mathcal{N}_{A\rightarrow B}^{\theta})^{\dag}
(X_{RB})=Z_{RA}$ satisfies \eqref{eq:condition-on-adjoint-op-root-SLD-chain},
by applying the optimization in \eqref{eq:root-SLD-opt-formula}, we find that
\begin{equation}
\sqrt{2}\left\vert \operatorname{Tr}[(\mathcal{N}_{A\rightarrow B}^{\theta
})^{\dag}(X_{RB})(\partial_{\theta}\rho_{RA}^{\theta})]\right\vert \leq
\sqrt{I_{F}}(\theta;\{\rho_{RA}^{\theta}\}_{\theta}).
\end{equation}
So we conclude that
\begin{equation}
\sqrt{2}\left\vert \operatorname{Tr}[X_{RB}(\partial_{\theta}(\mathcal{N}
_{A\rightarrow B}^{\theta}(\rho_{RA}^{\theta})))]\right\vert \leq\sqrt{I_{F}
}(\theta;\{\mathcal{N}_{A\rightarrow B}^{\theta}\}_{\theta})+\sqrt{I_{F}
}(\theta;\{\rho_{RA}^{\theta}\}_{\theta}).
\end{equation}
Since $X_{RB}$ is an arbitrary operator satisfying the inequality \eqref{eq:root-SLD-chain-arb-op}, we can optimize over all such operators to conclude the chain rule inequality in \eqref{eq:chain-rule-root-SLD}.
\end{proof}
\medskip

We show now how the chain rule proved above leads straightforwardly to an amortization collapse for the root SLD Fisher information of channels: 
\begin{corollary}[Amortization collapse]
\label{cor:amort-collapse-root-SLD}Let $\{\mathcal{N}_{A\rightarrow B}
^{\theta}\}_{\theta}$ be a family of differentiable quantum channels. Then the
following amortization collapse occurs for the root SLD\ Fisher information of
quantum channels:
\begin{equation}
\sqrt{I_{F}}^{\mathcal{A}}(\theta;\{\mathcal{N}_{A\rightarrow B}^{\theta
}\}_{\theta})=\sqrt{I_{F}}(\theta;\{\mathcal{N}_{A\rightarrow B}^{\theta
}\}_{\theta}),
\end{equation}
where
\begin{equation}
\sqrt{I_{F}}^{\mathcal{A}}(\theta;\{\mathcal{N}_{A\rightarrow B}^{\theta
}\}_{\theta}):=\sup_{\{\rho_{RA}^{\theta}\}_{\theta}}\left[  \sqrt{I_{F}
}(\theta;\{\mathcal{N}_{A\rightarrow B}^{\theta}(\rho_{RA}^{\theta}
)\}_{\theta})-\sqrt{I_{F}}(\theta;\{\rho_{RA}^{\theta}\}_{\theta})\right]  .
\end{equation}

\end{corollary}

\begin{proof}
If the finiteness condition in\ \eqref{eq:finiteness-condition-SLD-fish-ch}
does not hold, then the equality trivially holds. So let us suppose that the
finiteness condition in \eqref{eq:finiteness-condition-SLD-fish-ch}\ holds.
The inequality $\geq$ follows from Proposition~\ref{prop:amort->=-ch-Fish-gen}
and the fact that the root SLD\ Fisher information is faithful (see
\eqref{eq:SLD-RLD-Fish-faithful}). The opposite inequality $\leq$ is a
consequence of the chain rule from Proposition~\ref{prop:chain-rule-root-SLD}.
Let $\{\rho_{RA}^{\theta}\}_{\theta}$ be a family of quantum states on
systems $RA$. Then it follows from the chain rule proved above (Proposition~\ref{prop:chain-rule-root-SLD}) that
\begin{equation}
\sqrt{I_{F}}(\theta;\{\mathcal{N}_{A\rightarrow B}^{\theta}(\rho_{RA}^{\theta
})\}_{\theta})-\sqrt{I_{F}}(\theta;\{\rho_{RA}^{\theta}\}_{\theta})\leq
\sqrt{I_{F}}(\theta;\{\mathcal{N}_{A\rightarrow B}^{\theta}\}_{\theta}).
\end{equation}
Since the family $\{\rho_{RA}^{\theta}\}_{\theta}$ is arbitrary, we can take a
supremum of the left-hand side over all such families, and conclude that
\begin{equation}
\sqrt{I_{F}}^{\mathcal{A}}(\theta;\{\mathcal{N}_{A\rightarrow B}^{\theta
}\}_{\theta})\leq\sqrt{I_{F}}(\theta;\{\mathcal{N}_{A\rightarrow B}^{\theta
}\}_{\theta}).
\end{equation}
This concludes the proof.
\end{proof}

Further, as another corollary of the chain rule, we show that the root SLD Fisher information is subadditive with respect to serial composition of quantum channels. 

\begin{corollary}
\label{cor:subadd-serial-concat-root-SLD-Fish}Let $\{\mathcal{N}_{A\rightarrow
B}^{\theta}\}_{\theta}$ and $\{\mathcal{M}_{B\rightarrow C}^{\theta}
\}_{\theta}$ be differentiable families of quantum channels. Then the root
SLD\ Fisher information of quantum channels is subadditive with respect to
serial composition, in the following sense:
\begin{equation}
\sqrt{I_{F}}(\theta;\{\mathcal{M}_{B\rightarrow C}^{\theta}\circ
\mathcal{N}_{A\rightarrow B}^{\theta}\}_{\theta})\leq\sqrt{I_{F}}
(\theta;\{\mathcal{N}_{A\rightarrow B}^{\theta}\}_{\theta})+\sqrt{I_{F}
}(\theta;\{\mathcal{M}_{B\rightarrow C}^{\theta}\}_{\theta}).
\label{eq:serial-composition-root-SLD-ch}
\end{equation}

\end{corollary}

\begin{proof}
If the finiteness condition in\ \eqref{eq:finiteness-condition-SLD-fish-ch}
does not hold for either channel, then the inequality trivially holds. So let
us suppose that the finiteness condition in
\eqref{eq:finiteness-condition-SLD-fish-ch}\ holds for both channels. Pick an
arbitrary input state $\omega_{RA}$. Now apply
Proposition~\ref{prop:chain-rule-root-SLD}\ to find that
\begin{align}
&  \sqrt{I_{F}}(\theta;\{\mathcal{M}_{B\rightarrow C}^{\theta}(\mathcal{N}
_{A\rightarrow B}^{\theta}(\omega_{RA}))\}_{\theta})\nonumber\\
&  \leq\sqrt{I_{F}}(\theta;\{\mathcal{N}_{A\rightarrow B}^{\theta}(\omega
_{RA})\}_{\theta})+\sqrt{I_{F}}(\theta;\{\mathcal{M}_{B\rightarrow C}^{\theta
}\}_{\theta})\\
&  \leq\sup_{\omega_{RA}}\sqrt{I_{F}}(\theta;\{\mathcal{N}_{A\rightarrow
B}^{\theta}(\omega_{RA})\}_{\theta})+\sqrt{I_{F}}(\theta;\{\mathcal{M}
_{B\rightarrow C}^{\theta}\}_{\theta})\\
&  =\sqrt{I_{F}}(\theta;\{\mathcal{N}_{A\rightarrow B}^{\theta}\}_{\theta
})+\sqrt{I_{F}}(\theta;\{\mathcal{M}_{B\rightarrow C}^{\theta}\}_{\theta}).
\end{align}
Since the inequality holds for all input states, we conclude that
\begin{equation}
\sup_{\omega_{RA}}\sqrt{I_{F}}(\theta;\{\mathcal{M}_{B\rightarrow C}^{\theta
}(\mathcal{N}_{A\rightarrow B}^{\theta}(\omega_{RA}))\}_{\theta})\leq
\sqrt{I_{F}}(\theta;\{\mathcal{N}_{A\rightarrow B}^{\theta}\}_{\theta}
)+\sqrt{I_{F}}(\theta;\{\mathcal{M}_{B\rightarrow C}^{\theta}\}_{\theta}),
\end{equation}
which implies \eqref{eq:serial-composition-root-SLD-ch}.
\end{proof}

\subsection{Amortization collapse for RLD Fisher information of general channels}

Finally, we show the amortization collapse for the RLD Fisher information for quantum channels. First, we recall the following additivity relation that was established in \cite{Hayashi2011}:

\begin{proposition}
\label{prop:add-RLD-Fish-ch}Let $\{\mathcal{N}_{A\rightarrow B}^{\theta
}\}_{\theta}$ and $\{\mathcal{M}_{C\rightarrow D}^{\theta}\}_{\theta}$ be
differentiable families of quantum channels. Then the RLD\ Fisher information
of quantum channels is additive in the following sense:
\begin{equation}
\widehat{I}_{F}(\theta;\{\mathcal{N}_{A\rightarrow B}^{\theta}\otimes
\mathcal{M}_{C\rightarrow D}^{\theta}\}_{\theta})=\widehat{I}_{F}
(\theta;\{\mathcal{N}_{A\rightarrow B}^{\theta}\}_{\theta})+\widehat{I}
_{F}(\theta;\{\mathcal{M}_{C\rightarrow D}^{\theta}\}_{\theta}).
\end{equation}

\end{proposition}

Just as we did for the root SLD Fisher information above, we will establish a chain rule for the RLD Fisher information that we will next use to establish the amortization collapse.

Before doing so, we need the following lemmas:
\begin{lemma}
\label{lem:min-XYinvX}Let $X$ be a linear operator and let $Y$ be a positive
definite operator. Then
\begin{equation}
X^{\dag}Y^{-1}X=\min\left\{  M:
\begin{bmatrix}
M & X^{\dag}\\
X & Y
\end{bmatrix}
\geq0\right\}  , \label{eq:G_-1_opt}
\end{equation}
where the ordering for the minimization is understood in the operator interval
sense (L\"{o}wner order).
\end{lemma}

\begin{proof}
This is a direct consequence of the Schur complement lemma, which states that
\begin{equation}
\begin{bmatrix}
M & X^{\dag}\\
X & Y
\end{bmatrix}
\geq0\qquad\Longleftrightarrow\qquad Y\geq0,\quad M\geq X^{\dag}Y^{-1}X.
\end{equation}
This concludes the proof.
\end{proof}

\begin{lemma}[Transformer inequality]
\label{lem:transformer-ineq-basic}Let $X$ be a linear square operator, let $Y$
be a positive definite operator, and let $L$ be a linear operator. Then
\begin{equation}
LX^{\dag}L^{\dag}(LYL^{\dag})^{-1}LXL^{\dag}\leq LX^{\dag}Y^{-1}XL^{\dag},
\label{eq:transformer-ineq}
\end{equation}
where the inverse on the left hand side is taken on the image of $L$. If $L$
is invertible, then the following equality holds
\begin{equation}
LX^{\dag}L^{\dag}(LYL^{\dag})^{-1}LXL^{\dag}=LX^{\dag}Y^{-1}XL^{\dag}.
\label{eq:transformer-eq}
\end{equation}

\end{lemma}

\begin{proof}
Fix an operator $M\geq0$ satisfying
\begin{equation}
\begin{bmatrix}
M & X^{\dag}\\
X & Y
\end{bmatrix}
\geq0. \label{eq:condition-on-M-Schur-comp}
\end{equation}
Since the maps $(\cdot)\rightarrow L(\cdot)L^{\dag}$ and $(\cdot
)\rightarrow\left(  I_{2}\otimes L\right)  (\cdot)\left(  I_{2}\otimes
L\right)  ^{\dag}$ are positive, the condition $M\geq0$ and that in
\eqref{eq:condition-on-M-Schur-comp}\ imply the following conditions:
\begin{align}
LML^{\dag}  &  \geq0,\\
\begin{bmatrix}
LML^{\dag} & LX^{\dag}L^{\dag}\\
LXL^{\dag} & LYL^{\dag}
\end{bmatrix}
&  =\left(  I_{2}\otimes L\right)
\begin{bmatrix}
M & X^{\dag}\\
X & Y
\end{bmatrix}
\left(  I_{2}\otimes L\right)  ^{\dag}\geq0.
\end{align}
Applying Lemma \ref{lem:min-XYinvX}, we conclude that
\begin{align}
LML^{\dag}  &  \geq\min\left\{  W\geq0:
\begin{bmatrix}
W & LX^{\dag}L^{\dag}\\
LXL^{\dag} & LYL^{\dag}
\end{bmatrix}
\geq0\right\} \\
&  =LX^{\dag}L^{\dag}\left(  LYL^{\dag}\right)  ^{-1}LXL^{\dag}.
\end{align}
Since $M$ is an arbitrary operator that satisfies $M\geq0$ and
\eqref{eq:condition-on-M-Schur-comp}, we can pick it to be the smallest and
set it to $X^{\dag}Y^{-1}X$. Thus we conclude \eqref{eq:transformer-ineq}.

If $L$ is invertible, then consider that
\begin{align}
LX^{\dag}L^{\dag}(LYL^{\dag})^{-1}LXL^{\dag}  &  =LX^{\dag}L^{\dag}L^{-\dag
}Y^{-1}L^{-1}LXL^{\dag}\\
&  =LX^{\dag}Y^{-1}XL^{\dag},
\end{align}
so that \eqref{eq:transformer-eq} follows.
\end{proof}

\begin{proposition}
[Chain rule for RLD Fisher information of quantum channels]\label{prop:chain-rule-RLD}Let $\{\mathcal{N}_{A\rightarrow
B}^{\theta}\}_{\theta}$ be a differentiable family of quantum channels, and
let $\{\rho_{RA}^{\theta}\}_{\theta}$ be a differentiable family of quantum
states on systems $RA$, with the system $R$ of arbitrary size. Then the
following chain rule holds
\begin{equation}
\widehat{I}_{F}(\theta;\{\mathcal{N}_{A\rightarrow B}^{\theta}(\rho
_{RA}^{\theta})\}_{\theta})\leq\widehat{I}_{F}(\theta;\{\mathcal{N}
_{A\rightarrow B}^{\theta}\}_{\theta})+\widehat{I}_{F}(\theta;\{\rho
_{RA}^{\theta}\}_{\theta}). \label{eq:single-RLD-chain-rule}
\end{equation} 

\end{proposition}

\begin{proof}
If the finiteness conditions in
\eqref{eq:RLD-FI}\ and\ \eqref{eq:RLD-Fish-ch}
do not hold, then the inequality is trivially satisfied. So let us suppose
that the finiteness conditions
\eqref{eq:RLD-FI}\ and\ \eqref{eq:RLD-Fish-ch}
hold. Recall the following post-selected teleportation identity from
\eqref{eq:PS-TP-identity}:
\begin{equation}
\mathcal{N}_{A\rightarrow B}^{\theta}(\rho_{RA}^{\theta})=\langle\Gamma
|_{AS}\rho_{RA}^{\theta}\otimes\Gamma_{SB}^{\mathcal{N}^{\theta}}
|\Gamma\rangle_{AS}. \label{eq:TP-identity}
\end{equation}
Then we can write
\begin{align}
&  \widehat{I}_{F}(\theta;\{\mathcal{N}_{A\rightarrow B}^{\theta}(\rho
_{RA}^{\theta})\}_{\theta})\nonumber\\
&  =\operatorname{Tr}[(\partial_{\theta}\mathcal{N}_{A\rightarrow B}^{\theta
}(\rho_{RA}^{\theta}))^{2}(\mathcal{N}_{A\rightarrow B}^{\theta}(\rho
_{RA}^{\theta}))^{-1}]\\
&  =\operatorname{Tr}[(\partial_{\theta}(\langle\Gamma|_{AS}\rho_{RA}^{\theta
}\otimes\Gamma_{SB}^{\mathcal{N}^{\theta}}|\Gamma\rangle_{AS}))^{2}
(\langle\Gamma|_{AS}\rho_{RA}^{\theta}\otimes\Gamma_{SB}^{\mathcal{N}^{\theta
}}|\Gamma\rangle_{AS})^{-1}]\\
&  =\operatorname{Tr}[((\langle\Gamma|_{AS}\partial_{\theta}(\rho_{RA}
^{\theta}\otimes\Gamma_{SB}^{\mathcal{N}^{\theta}})|\Gamma\rangle_{AS}
))^{2}(\langle\Gamma|_{AS}\rho_{RA}^{\theta}\otimes\Gamma_{SB}^{\mathcal{N}
^{\theta}}|\Gamma\rangle_{AS})^{-1}]\\
&  \leq\operatorname{Tr}[\langle\Gamma|_{AS}(\partial_{\theta}(\rho
_{RA}^{\theta}\otimes\Gamma_{SB}^{\mathcal{N}^{\theta}}))(\rho_{RA}^{\theta
}\otimes\Gamma_{SB}^{\mathcal{N}^{\theta}})^{-1}(\partial_{\theta}(\rho
_{RA}^{\theta}\otimes\Gamma_{SB}^{\mathcal{N}^{\theta}}))|\Gamma\rangle
_{AS}]\\
&  =\operatorname{Tr}_{RB}[\langle\Gamma|_{AS}(\partial_{\theta}(\rho
_{RA}^{\theta}\otimes\Gamma_{SB}^{\mathcal{N}^{\theta}}))(\rho_{RA}^{\theta
}\otimes\Gamma_{SB}^{\mathcal{N}^{\theta}})^{-1}(\partial_{\theta}(\rho
_{RA}^{\theta}\otimes\Gamma_{SB}^{\mathcal{N}^{\theta}}))|\Gamma\rangle
_{AS}]\\
&  =\langle\Gamma|_{AS}\operatorname{Tr}_{RB}[(\partial_{\theta}(\rho
_{RA}^{\theta}\otimes\Gamma_{SB}^{\mathcal{N}^{\theta}}))(\rho_{RA}^{\theta
}\otimes\Gamma_{SB}^{\mathcal{N}^{\theta}})^{-1}(\partial_{\theta}(\rho
_{RA}^{\theta}\otimes\Gamma_{SB}^{\mathcal{N}^{\theta}}))]|\Gamma\rangle_{AS}.
\label{eq:proof-chain-rule-final-line}
\end{align}
The second equality follows from applying \eqref{eq:TP-identity}, and the
inequality is a consequence of the transformer inequality in
Lemma~\ref{lem:transformer-ineq-basic}, with
\begin{align}
L  &  =\langle\Gamma|_{AS}\otimes I_{RB},\\
X  &  =\partial_{\theta}(\rho_{RA}^{\theta}\otimes\Gamma_{SB}^{\mathcal{N}
^{\theta}}),\\
Y  &  =\rho_{RA}^{\theta}\otimes\Gamma_{SB}^{\mathcal{N}^{\theta}}.
\end{align}
Now consider that
\[
\partial_{\theta}(\rho_{RA}^{\theta}\otimes\Gamma_{SB}^{\mathcal{N}^{\theta}
})=(\partial_{\theta}\rho_{RA}^{\theta})\otimes\Gamma_{SB}^{\mathcal{N}
^{\theta}}+\rho_{RA}^{\theta}\otimes(\partial_{\theta}\Gamma_{SB}
^{\mathcal{N}^{\theta}}).
\]
Right multiplying this by $(\rho_{RA}^{\theta}\otimes\Gamma_{SB}
^{\mathcal{N}^{\theta}})^{-1}$ gives
\begin{align}
&  (\partial_{\theta}(\rho_{RA}^{\theta}\otimes\Gamma_{SB}^{\mathcal{N}
^{\theta}}))(\rho_{RA}^{\theta}\otimes\Gamma_{SB}^{\mathcal{N}^{\theta}}
)^{-1}\nonumber\\
&  =(\partial_{\theta}\rho_{RA}^{\theta})(\rho_{RA}^{\theta})^{-1}
\otimes\Gamma_{SB}^{\mathcal{N}^{\theta}}(\Gamma_{SB}^{\mathcal{N}^{\theta}
})^{-1}+\rho_{RA}^{\theta}(\rho_{RA}^{\theta})^{-1}\otimes(\partial_{\theta
}\Gamma_{SB}^{\mathcal{N}^{\theta}})(\Gamma_{SB}^{\mathcal{N}^{\theta}}
)^{-1}\\
&  =(\partial_{\theta}\rho_{RA}^{\theta})(\rho_{RA}^{\theta})^{-1}\otimes
\Pi_{\Gamma^{\mathcal{N}^{\theta}}}+\Pi_{\rho_{RA}^{\theta}}\otimes
(\partial_{\theta}\Gamma_{SB}^{\mathcal{N}^{\theta}})(\Gamma_{SB}
^{\mathcal{N}^{\theta}})^{-1}.
\end{align}
Right multiplying the last line by $(\partial_{\theta}(\rho_{RA}^{\theta
}\otimes\Gamma_{SB}^{\mathcal{N}^{\theta}}))$ gives
\begin{align}
&  \left[  (\partial_{\theta}\rho_{RA}^{\theta})(\rho_{RA}^{\theta}
)^{-1}\otimes\Pi_{\Gamma^{\mathcal{N}^{\theta}}}+\Pi_{\rho_{RA}^{\theta}
}\otimes(\partial_{\theta}\Gamma_{SB}^{\mathcal{N}^{\theta}})(\Gamma
_{SB}^{\mathcal{N}^{\theta}})^{-1}\right]  (\partial_{\theta}(\rho
_{RA}^{\theta}\otimes\Gamma_{SB}^{\mathcal{N}^{\theta}}))\nonumber\\
&  =\left[  (\partial_{\theta}\rho_{RA}^{\theta})(\rho_{RA}^{\theta}
)^{-1}\otimes\Pi_{\Gamma^{\mathcal{N}^{\theta}}}+\Pi_{\rho_{RA}^{\theta}
}\otimes(\partial_{\theta}\Gamma_{SB}^{\mathcal{N}^{\theta}})(\Gamma
_{SB}^{\mathcal{N}^{\theta}})^{-1}\right] \nonumber\\
&  \qquad\times\left[  (\partial_{\theta}\rho_{RA}^{\theta})\otimes\Gamma
_{SB}^{\mathcal{N}^{\theta}}+\rho_{RA}^{\theta}\otimes(\partial_{\theta}
\Gamma_{SB}^{\mathcal{N}^{\theta}})\right] \\
&  =(\partial_{\theta}\rho_{RA}^{\theta})(\rho_{RA}^{\theta})^{-1}
(\partial_{\theta}\rho_{RA}^{\theta})\otimes\Gamma_{SB}^{\mathcal{N}^{\theta}
}+(\partial_{\theta}\rho_{RA}^{\theta})(\rho_{RA}^{\theta})^{-1}\rho
_{RA}^{\theta}\otimes\Pi_{\Gamma^{\mathcal{N}^{\theta}}}(\partial_{\theta
}\Gamma_{SB}^{\mathcal{N}^{\theta}})\nonumber\\
&  \qquad+\Pi_{\rho_{RA}^{\theta}}(\partial_{\theta}\rho_{RA}^{\theta}
)\otimes(\partial_{\theta}\Gamma_{SB}^{\mathcal{N}^{\theta}})(\Gamma
_{SB}^{\mathcal{N}^{\theta}})^{-1}\Gamma_{SB}^{\mathcal{N}^{\theta}}+\rho
_{RA}^{\theta}\otimes(\partial_{\theta}\Gamma_{SB}^{\mathcal{N}^{\theta}
})(\Gamma_{SB}^{\mathcal{N}^{\theta}})^{-1}(\partial_{\theta}\Gamma
_{SB}^{\mathcal{N}^{\theta}})\nonumber\\
&  =(\partial_{\theta}\rho_{RA}^{\theta})(\rho_{RA}^{\theta})^{-1}
(\partial_{\theta}\rho_{RA}^{\theta})\otimes\Gamma_{SB}^{\mathcal{N}^{\theta}
}+(\partial_{\theta}\rho_{RA}^{\theta})\Pi_{\rho_{RA}^{\theta}}\otimes
\Pi_{\Gamma^{\mathcal{N}^{\theta}}}(\partial_{\theta}\Gamma_{SB}
^{\mathcal{N}^{\theta}})\nonumber\\
&  \qquad+\Pi_{\rho_{RA}^{\theta}}(\partial_{\theta}\rho_{RA}^{\theta}
)\otimes(\partial_{\theta}\Gamma_{SB}^{\mathcal{N}^{\theta}})\Pi
_{\Gamma^{\mathcal{N}^{\theta}}}+\rho_{RA}^{\theta}\otimes(\partial_{\theta
}\Gamma_{SB}^{\mathcal{N}^{\theta}})(\Gamma_{SB}^{\mathcal{N}^{\theta}}
)^{-1}(\partial_{\theta}\Gamma_{SB}^{\mathcal{N}^{\theta}}).
\end{align}
Since the finiteness conditions $\Pi_{\rho_{RA}^{\theta}}^{\perp}
(\partial_{\theta}\rho_{RA}^{\theta})=(\partial_{\theta}\rho_{RA}^{\theta}
)\Pi_{\rho_{RA}^{\theta}}^{\perp}=0$ and $\Pi_{\Gamma^{\mathcal{N}^{\theta}}
}^{\perp}(\partial_{\theta}\Gamma_{SB}^{\mathcal{N}^{\theta}})=(\partial
_{\theta}\Gamma_{SB}^{\mathcal{N}^{\theta}})\Pi_{\Gamma^{\mathcal{N}^{\theta}
}}^{\perp}=0$ hold, we can \textquotedblleft add in\textquotedblright\ extra
zero terms to the two middle terms above to conclude that
\begin{multline}
(\partial_{\theta}(\rho_{RA}^{\theta}\otimes\Gamma_{SB}^{\mathcal{N}^{\theta}
}))(\rho_{RA}^{\theta}\otimes\Gamma_{SB}^{\mathcal{N}^{\theta}})^{-1}
(\partial_{\theta}(\rho_{RA}^{\theta}\otimes\Gamma_{SB}^{\mathcal{N}^{\theta}
}))=(\partial_{\theta}\rho_{RA}^{\theta})(\rho_{RA}^{\theta})^{-1}
(\partial_{\theta}\rho_{RA}^{\theta})\otimes\Gamma_{SB}^{\mathcal{N}^{\theta}
}\\
+2(\partial_{\theta}\rho_{RA}^{\theta})\otimes(\partial_{\theta}\Gamma
_{SB}^{\mathcal{N}^{\theta}})+\rho_{RA}^{\theta}\otimes(\partial_{\theta
}\Gamma_{SB}^{\mathcal{N}^{\theta}})(\Gamma_{SB}^{\mathcal{N}^{\theta}}
)^{-1}(\partial_{\theta}\Gamma_{SB}^{\mathcal{N}^{\theta}}).
\end{multline}
Now taking the partial trace over the systems $RB$, we find the following for each term:
\begin{align}
\operatorname{Tr}_{RB}[(\partial_{\theta}\rho_{RA}^{\theta})(\rho_{RA}
^{\theta})^{-1}(\partial_{\theta}\rho_{RA}^{\theta})\otimes\Gamma
_{SB}^{\mathcal{N}^{\theta}}]  &  =\operatorname{Tr}_{R}[(\partial_{\theta
}\rho_{RA}^{\theta})(\rho_{RA}^{\theta})^{-1}(\partial_{\theta}\rho
_{RA}^{\theta})]\otimes I_{S},\\
\operatorname{Tr}_{RB}[2(\partial_{\theta}\rho_{RA}^{\theta})\otimes
(\partial_{\theta}\Gamma_{SB}^{\mathcal{N}^{\theta}})]  &  =2\operatorname{Tr}
_{R}[(\partial_{\theta}\rho_{RA}^{\theta})]\otimes\operatorname{Tr}
_{B}[(\partial_{\theta}\Gamma_{SB}^{\mathcal{N}^{\theta}})]\\
&  =2\operatorname{Tr}_{R}[(\partial_{\theta}\rho_{RA}^{\theta})]\otimes
(\partial_{\theta}\operatorname{Tr}_{B}[\Gamma_{SB}^{\mathcal{N}^{\theta}}])\\
&  =2\operatorname{Tr}_{R}[(\partial_{\theta}\rho_{RA}^{\theta})]\otimes
(\partial_{\theta}(I_{S})])\\
&  =0, \text{and} \\
\operatorname{Tr}_{RB}[\rho_{RA}^{\theta}\otimes(\partial_{\theta}\Gamma
_{SB}^{\mathcal{N}^{\theta}})(\Gamma_{SB}^{\mathcal{N}^{\theta}}
)^{-1}(\partial_{\theta}\Gamma_{SB}^{\mathcal{N}^{\theta}})]  &  =\rho
_{A}^{\theta}\otimes\operatorname{Tr}_{B}[(\partial_{\theta}\Gamma
_{SB}^{\mathcal{N}^{\theta}})(\Gamma_{SB}^{\mathcal{N}^{\theta}}
)^{-1}(\partial_{\theta}\Gamma_{SB}^{\mathcal{N}^{\theta}})].
\end{align}
Now applying the sandwich $\langle\Gamma|_{AS}(\cdot)|\Gamma\rangle_{AS}$, the
first and last term become as follows:
\begin{align}
&  \langle\Gamma|_{AS}\operatorname{Tr}_{R}[(\partial_{\theta}\rho
_{RA}^{\theta})(\rho_{RA}^{\theta})^{-1}(\partial_{\theta}\rho_{RA}^{\theta
})]\otimes I_{S}|\Gamma\rangle_{AS}\nonumber\\
&  =\operatorname{Tr}[(\partial_{\theta}\rho_{RA}^{\theta})(\rho_{RA}^{\theta
})^{-1}(\partial_{\theta}\rho_{RA}^{\theta})]\\
&  =\operatorname{Tr}[(\partial_{\theta}\rho_{RA}^{\theta})^{2}(\rho
_{RA}^{\theta})^{-1}],
\end{align}
and
\begin{multline}
\langle\Gamma|_{AS}\rho_{A}^{\theta}\otimes\operatorname{Tr}_{B}
[(\partial_{\theta}\Gamma_{SB}^{\mathcal{N}^{\theta}})(\Gamma_{SB}
^{\mathcal{N}^{\theta}})^{-1}(\partial_{\theta}\Gamma_{SB}^{\mathcal{N}
^{\theta}})]|\Gamma\rangle_{AS}\\
=\operatorname{Tr}[(\rho_{S}^{\theta})^{T}\operatorname{Tr}_{B}[(\partial
_{\theta}\Gamma_{SB}^{\mathcal{N}^{\theta}})(\Gamma_{SB}^{\mathcal{N}^{\theta
}})^{-1}(\partial_{\theta}\Gamma_{SB}^{\mathcal{N}^{\theta}})]].
\end{multline}
Plugging back into \eqref{eq:proof-chain-rule-final-line}, we find that
\begin{align}
&  \langle\Gamma|_{AS}\operatorname{Tr}_{RB}[(\partial_{\theta}(\rho
_{RA}^{\theta}\otimes\Gamma_{SB}^{\mathcal{N}^{\theta}}))(\rho_{RA}^{\theta
}\otimes\Gamma_{SB}^{\mathcal{N}^{\theta}})^{-1}(\partial_{\theta}(\rho
_{RA}^{\theta}\otimes\Gamma_{SB}^{\mathcal{N}^{\theta}}))]|\Gamma\rangle
_{AS}\nonumber\\
&  =\operatorname{Tr}[(\partial_{\theta}\rho_{RA}^{\theta})^{2}(\rho
_{RA}^{\theta})^{-1}]+\operatorname{Tr}[(\rho_{S}^{\theta})^{T}
\operatorname{Tr}_{B}[(\partial_{\theta}\Gamma_{SB}^{\mathcal{N}^{\theta}
})(\Gamma_{SB}^{\mathcal{N}^{\theta}})^{-1}(\partial_{\theta}\Gamma
_{SB}^{\mathcal{N}^{\theta}})]]\\
&  \leq\operatorname{Tr}[(\partial_{\theta}\rho_{RA}^{\theta})^{2}(\rho
_{RA}^{\theta})^{-1}]+\left\Vert \operatorname{Tr}_{B}[(\partial_{\theta
}\Gamma_{SB}^{\mathcal{N}^{\theta}})(\Gamma_{SB}^{\mathcal{N}^{\theta}}
)^{-1}(\partial_{\theta}\Gamma_{SB}^{\mathcal{N}^{\theta}})]\right\Vert
_{\infty}\\
&  =\widehat{I}_{F}(\theta;\{\rho_{RA}^{\theta}\}_{\theta})+\widehat{I}
_{F}(\theta;\{\mathcal{N}_{A\rightarrow B}^{\theta}\}_{\theta}).
\end{align}
This concludes the proof.
\end{proof}

From the chain rule, we can make two conclusions: that amortization does not increase the RLD Fisher information of channels, and that the RLD Fisher information of channels is subadditive with respect to serial composition, or concatenation.
\begin{corollary}[Amortization collapse]
\label{cor:amort-collapse-RLD-fish}Let $\{\mathcal{N}_{A\rightarrow B}
^{\theta}\}_{\theta}$ be a differentiable family of quantum channels. Then
amortization does not increase the RLD\ Fisher information of quantum
channels, in the following sense:
\begin{equation}
\widehat{I}_{F}^{\mathcal{A}}(\theta;\{\mathcal{N}_{A\rightarrow B}^{\theta
}\}_{\theta})=\widehat{I}_{F}(\theta;\{\mathcal{N}_{A\rightarrow B}^{\theta
}\}_{\theta}).
\end{equation}

\end{corollary}

\begin{proof}
If the finiteness condition in \eqref{eq:RLD-Fish-ch}
does not hold, then the equality trivially holds. So let us suppose that the
finiteness condition in \eqref{eq:RLD-Fish-ch} holds. The
inequality $\geq$ follows from Proposition~\ref{prop:amort->=-ch-Fish-gen} and
the fact that the RLD\ Fisher information is faithful (see
\eqref{eq:SLD-RLD-Fish-faithful}). The opposite inequality $\leq$ is a
consequence of the chain rule from Proposition~\ref{prop:chain-rule-RLD}. Let
$\{\rho_{RA}^{\theta}\}_{\theta}$ be a family of quantum states on systems
$RA$. Then it follows from the chain rule that
\begin{equation}
\widehat{I}_{F}(\theta;\{\mathcal{N}_{A\rightarrow B}^{\theta}(\rho
_{RA}^{\theta})\}_{\theta})-\widehat{I}_{F}(\theta;\{\rho_{RA}^{\theta
}\}_{\theta})\leq\widehat{I}_{F}(\theta;\{\mathcal{N}_{A\rightarrow B}
^{\theta}\}_{\theta}).
\end{equation}
Since the family $\{\rho_{RA}^{\theta}\}_{\theta}$ is arbitrary, we can take a
supremum over the left-hand side over all such families, and conclude that
\begin{equation}
\widehat{I}_{F}^{\mathcal{A}}(\theta;\{\mathcal{N}_{A\rightarrow B}^{\theta
}\}_{\theta})\leq\widehat{I}_{F}(\theta;\{\mathcal{N}_{A\rightarrow B}
^{\theta}\}_{\theta}).
\end{equation}
This concludes the proof.
\end{proof}

\begin{corollary}
\label{cor:subadd-serial-concat-RLD-Fish}Let $\{\mathcal{N}_{A\rightarrow
B}^{\theta}\}_{\theta}$ and $\{\mathcal{M}_{B\rightarrow C}^{\theta}
\}_{\theta}$ be differentiable families of quantum channels. Then the
RLD\ Fisher information of quantum channels is subadditive with respect to
serial composition, in the following sense:
\begin{equation}
\widehat{I}_{F}(\theta;\{\mathcal{M}_{B\rightarrow C}^{\theta}\circ
\mathcal{N}_{A\rightarrow B}^{\theta}\}_{\theta})\leq\widehat{I}_{F}
(\theta;\{\mathcal{N}_{A\rightarrow B}^{\theta}\}_{\theta})+\widehat{I}
_{F}(\theta;\{\mathcal{M}_{B\rightarrow C}^{\theta}\}_{\theta}).
\label{eq:serial-composition-RLD-ch}
\end{equation}

\end{corollary}

\begin{proof}
If the finiteness condition in \eqref{eq:RLD-Fish-ch}
does not hold for both channels, then the inequality is trivially satisfied.
So let us suppose that the finiteness condition in
\eqref{eq:RLD-Fish-ch} holds for both channels. Pick an
arbitrary input state $\omega_{RA}$. Now apply
Proposition~\ref{prop:chain-rule-RLD}\ to find that
\begin{align}
&  \widehat{I}_{F}(\theta;\{\mathcal{M}_{B\rightarrow C}^{\theta}
(\mathcal{N}_{A\rightarrow B}^{\theta}(\omega_{RA}))\}_{\theta})\nonumber\\
&  \leq\widehat{I}_{F}(\theta;\{\mathcal{N}_{A\rightarrow B}^{\theta}
(\omega_{RA})\}_{\theta})+\widehat{I}_{F}(\theta;\{\mathcal{M}_{B\rightarrow
C}^{\theta}\}_{\theta})\\
&  \leq\sup_{\omega_{RA}}\widehat{I}_{F}(\theta;\{\mathcal{N}_{A\rightarrow
B}^{\theta}(\omega_{RA})\}_{\theta})+\widehat{I}_{F}(\theta;\{\mathcal{M}
_{B\rightarrow C}^{\theta}\}_{\theta})\\
&  =\widehat{I}_{F}(\theta;\{\mathcal{N}_{A\rightarrow B}^{\theta}\}_{\theta
})+\widehat{I}_{F}(\theta;\{\mathcal{M}_{B\rightarrow C}^{\theta}\}_{\theta}).
\end{align}
Since the inequality holds for all input states, we conclude that
\begin{equation}
\sup_{\omega_{RA}}\widehat{I}_{F}(\theta;\{\mathcal{M}_{B\rightarrow
C}^{\theta}(\mathcal{N}_{A\rightarrow B}^{\theta}(\omega_{RA}))\}_{\theta
})\leq\widehat{I}_{F}(\theta;\{\mathcal{N}_{A\rightarrow B}^{\theta}
\}_{\theta})+\widehat{I}_{F}(\theta;\{\mathcal{M}_{B\rightarrow C}^{\theta
}\}_{\theta}),
\end{equation}
which implies \eqref{eq:serial-composition-RLD-ch}.
\end{proof}

\section{Limits on channel estimation}

The goal of this section is to provide Cramer--Rao bounds for channel estimation in the sequential setting. That is, we wish to establish lower bounds on the mean-squared error of an unbiased estimator for quantum channel estimation in the sequential setting. The main ingredient of these bounds will be the amortization collapses proved above. However, before doing so, we need to connect the amortized Fisher information to sequential estimation, which we do via a meta-converse that generalizes the related meta-converse of~\cite{Berta2018}. 

First, we prove the following property obeyed by any generalized Fisher information of quantum channels:
\begin{proposition}
\label{prop:choi-state-bound-gen-fish}Let $\{\mathcal{N}_{A\rightarrow
B}^{\theta}\}_{\theta}$ be a family of quantum channels, and suppose that the
underlying generalized Fisher information is weakly faithful and obeys the
direct-sum property. Then the following inequalities hold
\begin{equation}
\mathbf{I}_{F}(\theta;\{\mathcal{N}_{A\rightarrow B}^{\theta}(\Phi
_{RA})\}_{\theta})\leq\mathbf{I}_{F}(\theta;\{\mathcal{N}_{A\rightarrow
B}^{\theta}\}_{\theta})\leq d\cdot\mathbf{I}_{F}(\theta;\{\mathcal{N}
_{A\rightarrow B}^{\theta}(\Phi_{RA})\}_{\theta}),
\label{eq:choi-state-bound-gen-fish}
\end{equation}
where $\Phi_{RA}$ is the maximally entangled state and $d$ is the dimension of
the channel input system $A$.
\end{proposition}

\begin{proof}
The first inequality is trivial, following from the definition in
\eqref{eq:gen-fisher-channels}. So we prove the second one and note that it
follows from a quantum steering or remote state preparation argument. Let
$\psi_{RA}$ be an arbitrary pure bipartite input state. To each such state,
there exists an operator $Z_{R}$ satisfying
\begin{align}
\psi_{RA}  &  =d\cdot Z_{R}\Phi_{RA}Z_{R}^{\dag},\\
\operatorname{Tr}[Z_{R}^{\dag}Z_{R}]  &  =1.
\end{align}
Let $\mathcal{P}_{R\rightarrow XR}$ denote the following steering quantum
channel:
\begin{equation}
\mathcal{P}_{R\rightarrow XR}(\omega_{R}):=|0\rangle\!\langle0|_{X}\otimes
Z_{R}\omega_{R}Z_{R}^{\dag}+|1\rangle\!\langle1|_{X}\otimes\sqrt{I_{R}
-Z_{R}^{\dag}Z_{R}}\omega_{R}\sqrt{I_{R}-Z_{R}^{\dag}Z_{R}},
\end{equation}
and consider that
\begin{equation}
\mathcal{P}_{R\rightarrow XR}(\Phi_{RA})=\frac{1}{d}|0\rangle\!\langle
0|_{X}\otimes\psi_{RA}+\left(  1-\frac{1}{d}\right)  |1\rangle\!\langle
1|_{X}\otimes\sigma_{RA},
\end{equation}
where
\begin{equation}
\sigma_{RA}:=\left(  1-\frac{1}{d}\right)  ^{-1}\sqrt{I_{R}-Z_{R}^{\dag}Z_{R}
}\Phi_{RA}\sqrt{I_{R}-Z_{R}^{\dag}Z_{R}}.
\end{equation}
This implies that
\begin{align}
&  \mathcal{P}_{R\rightarrow XR}(\mathcal{N}_{A\rightarrow B}^{\theta}
(\Phi_{RA}))\nonumber\\
&  =\mathcal{N}_{A\rightarrow B}^{\theta}(\mathcal{P}_{R\rightarrow XR}
(\Phi_{RA}))\\
&  =\frac{1}{d}|0\rangle\!\langle0|_{X}\otimes\mathcal{N}_{A\rightarrow
B}^{\theta}(\psi_{RA})+\left(  1-\frac{1}{d}\right)  |1\rangle\!\langle
1|_{X}\otimes\mathcal{N}_{A\rightarrow B}^{\theta}(\sigma_{RA}).
\label{eq:steering-ch-dim-bound}
\end{align}
Then we find that
\begin{align}
&  \mathbf{I}_{F}(\theta;\{\mathcal{N}_{A\rightarrow B}^{\theta}(\Phi
_{RA})\}_{\theta})\nonumber\\
&  \geq\mathbf{I}_{F}(\theta;\{\mathcal{P}_{R\rightarrow XR}(\mathcal{N}
_{A\rightarrow B}^{\theta}(\Phi_{RA}))\}_{\theta})\\
&  =\frac{1}{d}\mathbf{I}_{F}(\theta;\{\mathcal{N}_{A\rightarrow B}^{\theta
}(\psi_{RA})\}_{\theta})+\left(  1-\frac{1}{d}\right)  \mathbf{I}_{F}
(\theta;\{\mathcal{N}_{A\rightarrow B}^{\theta}(\sigma_{RA})\}_{\theta})\\
&  \geq\frac{1}{d}\mathbf{I}_{F}(\theta;\{\mathcal{N}_{A\rightarrow B}
^{\theta}(\psi_{RA})\}_{\theta}).
\end{align}
The first inequality follows from data processing. The equality follows from
\eqref{eq:steering-ch-dim-bound}\ and the direct-sum property in
\eqref{eq:direct-sum-prop-gen-fish}. The last inequality follows from the
assumption that $\mathbf{I}_{F}$ is weakly faithful, so that $\mathbf{I}
_{F}(\theta;\{\mathcal{N}_{A\rightarrow B}^{\theta}(\sigma_{RA})\}_{\theta
})\geq0$. Since the inequality holds for all pure bipartite states $\psi_{RA}$, we conclude the second inequality in \eqref{eq:choi-state-bound-gen-fish}.
\end{proof}
\medskip

We now have all the ingredients required to state and prove the required meta-converse:
\begin{theorem}
\label{thm:meta-converse}Consider a general sequential channel estimation
protocol of the form discussed in Chapter~\ref{ch:prelims} and reproduced in this chapter in Figure~\ref{fig:sequential-protocol-chapter3}. Suppose that
the generalized Fisher information $\mathbf{I}_{F}$ is weakly faithful. Then
the following inequality holds
\begin{equation}
\mathbf{I}_{F}(\theta;\{\omega_{R_{n}B_{n}}^{\theta}\}_{\theta})\leq
n\cdot \mathbf{I}_{F}^{\mathcal{A}}(\theta;\{\mathcal{N}_{A\rightarrow B}^{\theta
}\}_{\theta}),
\end{equation}
where $\omega_{R_{n}B_{n}}^{\theta}$ is the final state of the estimation
protocol in Figure~\ref{fig:sequential-protocol-chapter3}, as given in \eqref{eq:estimation-final-state}.
\end{theorem}

\begin{proof}
Consider that
\begin{align}
&  \mathbf{I}_{F}(\theta;\{\omega_{R_{n}B_{n}}^{\theta}\}_{\theta})\nonumber\\
&  =\mathbf{I}_{F}(\theta;\{\omega_{R_{n}B_{n}}^{\theta}\}_{\theta
})-\mathbf{I}_{F}(\theta;\{\rho_{R_{1}A_{1}}\}_{\theta})\\
&  =\mathbf{I}_{F}(\theta;\{\omega_{R_{n}B_{n}}^{\theta}\}_{\theta
})-\mathbf{I}_{F}(\theta;\{\rho_{R_{1}A_{1}}\}_{\theta})+\sum_{i=2}^{n}\left(
\mathbf{I}_{F}(\theta;\{\rho_{R_{i}A_{i}}^{\theta}\}_{\theta})-\mathbf{I}
_{F}(\theta;\{\rho_{R_{i}A_{i}}^{\theta}\}_{\theta})\right) \\
&  =\mathbf{I}_{F}(\theta;\{\omega_{R_{n}B_{n}}^{\theta}\}_{\theta
})-\mathbf{I}_{F}(\theta;\{\rho_{R_{1}A_{1}}\}_{\theta})\nonumber\\
&  \qquad+\sum_{i=2}^{n}\left(  \mathbf{I}_{F}(\theta;\{\mathcal{S}
_{R_{i-1}B_{i-1}\rightarrow R_{i}A_{i}}^{i-1}(\rho_{R_{i-1}B_{i-1}}^{\theta
})\}_{\theta})-\mathbf{I}_{F}(\theta;\{\rho_{R_{i}A_{i}}^{\theta}\}_{\theta
})\right) \\
&  \leq\mathbf{I}_{F}(\theta;\{\omega_{R_{n}B_{n}}^{\theta}\}_{\theta
})-\mathbf{I}_{F}(\theta;\{\rho_{R_{1}A_{1}}\}_{\theta})\nonumber\\
&  \qquad+\sum_{i=2}^{n}\left(  \mathbf{I}_{F}(\theta;\{\rho_{R_{i-1}B_{i-1}
}^{\theta}\}_{\theta})-\mathbf{I}_{F}(\theta;\{\rho_{R_{i}A_{i}}^{\theta
}\}_{\theta})\right) \\
&  =\sum_{i=1}^{n}\left(  \mathbf{I}_{F}(\theta;\{\rho_{R_{i}B_{i}}^{\theta
}\}_{\theta})-\mathbf{I}_{F}(\theta;\{\rho_{R_{i}A_{i}}^{\theta}\}_{\theta
}\right) \\
&  =\sum_{i=1}^{n}\left(  \mathbf{I}_{F}(\theta;\{\mathcal{N}_{A_{i}
\rightarrow B_{i}}^{\theta}(\rho_{R_{i}A_{i}}^{\theta})\}_{\theta}
)-\mathbf{I}_{F}(\theta;\{\rho_{R_{i}A_{i}}^{\theta}\}_{\theta}\right) \\
&  \leq n\cdot\sup_{\{\rho_{RA}^{\theta}\}_{\theta}}\left[  \mathbf{I}
_{F}(\theta;\{\mathcal{N}_{A\rightarrow B}^{\theta}(\rho_{RA}^{\theta
})\}_{\theta})-\mathbf{I}_{F}(\theta;\{\rho_{RA}^{\theta})\}_{\theta})\right]
\\
&  =n\cdot \mathbf{I}_{F}^{\mathcal{A}}(\theta;\{\mathcal{N}_{A\rightarrow
B}^{\theta}\}_{\theta}).
\end{align}
The first equality follows from the weak faithfulness assumption and because
the initial state of the protocol has no dependence on the parameter $\theta$.
The inequality follows from the data-processing inequality. The other steps
are straightforward manipulations.
\end{proof}

\subsection{SLD Fisher information limit on parameter estimation of classical-quantum channels}

The first Cramer--Rao bound we provide is that for the special case of classical-quantum channels, for which we proved the amortization collapse in Proposition~\ref{thm:amort-collapse-cq}.

\begin{proposition}
As a direct consequence of the QCRB\ in \eqref{eq:QCRB}, the amortization collapse from Proposition~\ref{thm:amort-collapse-cq}, and the meta-converse from Theorem~\ref{thm:meta-converse}, we conclude the following bound on the MSE\ of an unbiased estimator $\hat{\theta}$ for classical--quantum channel families defined in \eqref{eq:cq-channel-fams} and for which the finiteness condition in \eqref{eq:finiteness-condition-cq-channels} holds:
\begin{equation} \label{eq:cq-channels-crb}
\operatorname{Var}(\hat{\theta})\geq\frac{1}{n\sup_{x}I_{F}(\theta ;\{\omega_{B}^{x,\theta}\}_{\theta})}.
\end{equation}
\end{proposition}

Due to the amortization collapse, there is no advantage that sequential estimation strategies bring over parallel estimation strategies for this class of channels. In fact, an optimal parallel estimation strategy consists of picking the same optimal input letter $x$ to each channel use in order to estimate~$\theta$.

\subsection{SLD Fisher information limit on parameter estimation of general channels}

Our next contribution is to prove a Cramer--Rao bound for sequential channel estimation of general channels using the SLD Fisher information. Here, we see how it is a consequence of the QCRB\ in \eqref{eq:QCRB}, the amortization collapse from Corollary~\ref{cor:amort-collapse-root-SLD}, and the meta-converse from Theorem~\ref{thm:meta-converse}. The bound we provide in \eqref{eq:Heisenberg-bnd-SLD-channels} was reported recently in \cite{Yuan2017}. At the same time, our approach offers a technical improvement over the result of \cite{Yuan2017}, in that the families of quantum channels to which the bound applies need only be differentiable rather than second-order differentiable, the latter being required by the approach of \cite{Yuan2017}.

\begin{proposition}
As a direct consequence of the QCRB\ in \eqref{eq:QCRB}, the amortization collapse from
Corollary~\ref{cor:amort-collapse-root-SLD}, and the meta-converse
from Theorem~\ref{thm:meta-converse}, we conclude the following bound on the MSE\ of an unbiased estimator $\hat{\theta}$ for all differentiable
quantum channel families:
\begin{equation}
\operatorname{Var}(\hat{\theta})\geq\frac{1}{n^{2}I_{F}(\theta;\{\mathcal{N}
_{A\rightarrow B}^{\theta}\}_{\theta})}.
\label{eq:Heisenberg-bnd-SLD-channels}
\end{equation}
\end{proposition}

This bound thus poses a ``Heisenberg'' limitation on sequential estimation protocols for all differentiable quantum channel families satisfying the finiteness condition in \eqref{eq:finiteness-condition-SLD-fish-ch}. That is, in estimation tasks where the SLD Fisher information of a channel is achievable, the bound in~\eqref{eq:Heisenberg-bnd-SLD-channels} states that Heisenberg scaling of the estimator error is attainable using sequential estimation.

\subsection{RLD Fisher information limit on parameter estimation of general channels}

Finally, we provide a Cramer--Rao bound for sequential channel estimation using the RLD Fisher information. 
\begin{proposition}
\label{conc:RLD-bnd}As a direct consequence of the QCRB\ in
\eqref{eq:RLD-QCRB}, the amortization collapse from Corollary~\ref{cor:amort-collapse-RLD-fish}, and the meta-converse from Theorem~\ref{thm:meta-converse}
, we conclude the following bound on the MSE\ of an unbiased estimator
$\hat{\theta}$ for all quantum channel families $\{\mathcal{N}_{A\rightarrow
B}^{\theta}\}_{\theta}$:
\begin{equation}
\operatorname{Var}(\hat{\theta})\geq\frac{1}{n\widehat{I}_{F}(\theta
;\{\mathcal{N}_{A\rightarrow B}^{\theta}\}_{\theta})}. \label{eq:RLD-bnd-ch}
\end{equation}
\end{proposition}

Proposition~\ref{conc:RLD-bnd} strengthens one of the results of
\cite{Hayashi2011}. There, it was proved that the RLD\ Fisher information of
quantum channels is a limitation for parallel estimation protocols, but
Proposition~\ref{conc:RLD-bnd} establishes it as a limitation for the more general
sequential estimation protocols.

This bound thus poses a strong limitation on sequential estimation protocols for all differentiable quantum channel families satisfying the finiteness condition in \eqref{eq:RLD-Fish-ch}. That is, for a channel family with finite RLD Fisher information, no estimation protocol can attain Heisenberg scaling with respect to number of channel uses in the sequential setting. 

\section{Example: Estimating the parameters of the generalized amplitude damping channel} \label{subsec:estimating-gadc}

We now use the bounds established, in particular the one involving the RLD Fisher information in~\eqref{eq:RLD-bnd-ch}, to the example of the generalized amplitude damping channel \cite{Nielsen2000}. This channel has been studied previously in the context of quantum estimation theory \cite{Fujiwara2003, Fujiwara2004}, where its SLD\ Fisher information of quantum channels was computed and analyzed. Our goal now is to compute the RLD\ Fisher information of this channel with respect to its parameters.

A generalized amplitude damping channel is defined in terms of its loss $\gamma\in(0,1)$ and noise $N\in\left(  0,1\right)  $ as
\begin{equation}
\mathcal{A}_{\gamma,N}(\rho):=K_{1}\rho K_{1}^{\dag}+K_{2}\rho K_{2}^{\dag
}+K_{3}\rho K_{3}^{\dag}+K_{4}\rho K_{4}^{\dag}, \label{eq:GADC}
\end{equation}
where
\begin{align}
K_{1}  &  :=\sqrt{1-N}\left(  |0\rangle\!\langle0|+\sqrt{1-\gamma}
|1\rangle\!\langle1|\right)  ,\\
K_{2}  &  :=\sqrt{\gamma\left(  1-N\right)  }|0\rangle\!\langle1|,\\
K_{3}  &  :=\sqrt{N}\left(  \sqrt{1-\gamma}|0\rangle\!\langle0|+|1\rangle
\langle1|\right)  ,\\
K_{4}  &  :=\sqrt{\gamma N}|1\rangle\!\langle0|,
\end{align}
where $K_1$, $K_2$, $K_3$ and $K_4$ are its Kraus operators.
The Choi operator of the channel is then given by
\begin{align}
\Gamma_{RB}^{\mathcal{A}_{\gamma,N}}  &  :=(\operatorname{id}_{R}
\otimes\mathcal{A}_{\gamma,N})(\Gamma_{RA})\\
&  =\left(  1-\gamma N\right)  |00\rangle\!\langle00|+\sqrt{1-\gamma}\left(
|00\rangle\!\langle11|+|11\rangle\!\langle00|\right)  +\gamma N|01\rangle
\langle01|\nonumber\\
&  \qquad+\gamma\left(  1-N\right)  |10\rangle\!\langle10|+\left(
1-\gamma\left(  1-N\right)  \right)  |11\rangle\!\langle11|\\
&  =
\begin{bmatrix}
1-\gamma N & 0 & 0 & \sqrt{1-\gamma}\\
0 & \gamma N & 0 & 0\\
0 & 0 & \gamma\left(  1-N\right)  & 0\\
\sqrt{1-\gamma} & 0 & 0 & 1-\gamma\left(  1-N\right)
\end{bmatrix}
.
\end{align}

\subsection{Estimating loss}

The first task we study is to estimate the loss parameter $\gamma\in\left(  0,1\right)  $ of a generalized amplitude damping channel. By direct evaluation, we find that
\begin{equation}
\partial_{\gamma}\Gamma_{RB}^{\mathcal{A}_{\gamma,N}}=
\begin{bmatrix}
-N & 0 & 0 & -\frac{1}{2\sqrt{1-\gamma}}\\
0 & N & 0 & 0\\
0 & 0 & 1-N & 0\\
-\frac{1}{2\sqrt{1-\gamma}} & 0 & 0 & -\left(  1-N\right)
\end{bmatrix}
.
\end{equation}
Then we evaluate the expression for the RLD Fisher information of channels in \eqref{eq:RLD-Fish-ch}, which for our case with respect to parameter $\gamma$ is as follows:
\begin{equation}
\widehat{I}_{F}(\gamma;\{\mathcal{A}_{\gamma,N}\}_{\gamma})=\left\Vert
\operatorname{Tr}_{B}\left[  \left(  \partial_{\gamma}\Gamma_{RB}
^{\mathcal{A}_{\gamma,N}}\right)  \left(  \Gamma_{RB}^{\mathcal{A}_{\gamma,N}
}\right)  ^{-1}\left(  \partial_{\gamma}\Gamma_{RB}^{\mathcal{A}_{\gamma,N}
}\right)  \right]  \right\Vert _{\infty}.
\end{equation}
Using the fact that
\begin{equation}
\left(  \Gamma_{RB}^{\mathcal{A}_{\gamma,N}}\right)  ^{-1}=
\begin{bmatrix}
\frac{1-\gamma\left(  1-N\right)  }{\left(  1-N\right)  N\gamma^{2}} & 0 & 0 &
\frac{-\sqrt{1-\gamma}}{\left(  1-N\right)  N\gamma^{2}}\\
0 & \frac{1}{\gamma N} & 0 & 0\\
0 & 0 & \frac{1}{\gamma\left(  1-N\right)  } & 0\\
\frac{-\sqrt{1-\gamma}}{\left(  1-N\right)  N\gamma^{2}} & 0 & 0 &
\frac{1-\gamma N}{\left(  1-N\right)  N\gamma^{2}}
\end{bmatrix}
, \label{eq:GADC-inverse}
\end{equation}
we find that
\begin{equation}
\operatorname{Tr}_{B}\left[  \left(  \partial_{\gamma}\Gamma_{RB}
^{\mathcal{A}_{\gamma,N}}\right)  \left(  \Gamma_{RB}^{\mathcal{A}_{\gamma,N}
}\right)  ^{-1}\left(  \partial_{\gamma}\Gamma_{RB}^{\mathcal{A}_{\gamma,N}
}\right)  \right]  =
\begin{bmatrix}
f_{1}(\gamma,N) & 0\\
0 & f_{2}(\gamma,N)
\end{bmatrix}
,
\end{equation}
where
\begin{align}
f_{1}(\gamma,N)  &  :=\frac{\left(  4N-3\right)  N+\frac{1-N}{1-\gamma
}+4\left(  1-N\right)  N\left(  1-2N\right)  \gamma}{4N\left(  1-N\right)
\gamma^{2}},\\
f_{2}(\gamma,N)  &  :=\frac{8\gamma N+\frac{1}{N}+\frac{1}{\left(  1-N\right)
\left(  1-\gamma\right)  }-4\left(  1+\gamma\right)  }{4\gamma^{2}}.
\end{align}
Note that if $N\leq1/2$, then $f_{1}(\gamma,N)\geq f_{2}(\gamma,N)$, while if
$N>1/2$, then $f_{1}(\gamma,N)<f_{2}(\gamma,N)$. It then follows that
\begin{equation}
\widehat{I}_{F}(\gamma;\{\mathcal{A}_{\gamma,N}\}_{\gamma})=\left\{
\begin{array}
[c]{cc}
f_{1}(\gamma,N) & N\leq1/2\\
f_{2}(\gamma,N) & N>1/2.
\end{array}
\right.   \label{eq:estimate-gamma-GADC-RLD}
\end{equation}

It follows from the Cramer--Rao bound in \eqref{eq:RLD-bnd-ch} that the formula in \eqref{eq:estimate-gamma-GADC-RLD} provides a fundamental limitation on any
protocol that attempts to estimate the loss parameter $\gamma$. For the noise parameter $N$ equal to $0.2$ and $0.45$, Figure~\ref{fig:estimate-loss} depicts the logarithm of this bound, as well as the logarithm of the achievable bound from the SLD\ Fisher information of channels, corresponding to a parallel strategy that estimates $\gamma$. The RLD\ bound becomes better as $N$ approaches $1/2$, and we find numerically that the RLD\ and SLD\ bounds coincide at $N=1/2$.

\begin{figure}
\begin{subfigure}{.5\textwidth}
\centering
\includegraphics[width=.9\linewidth]{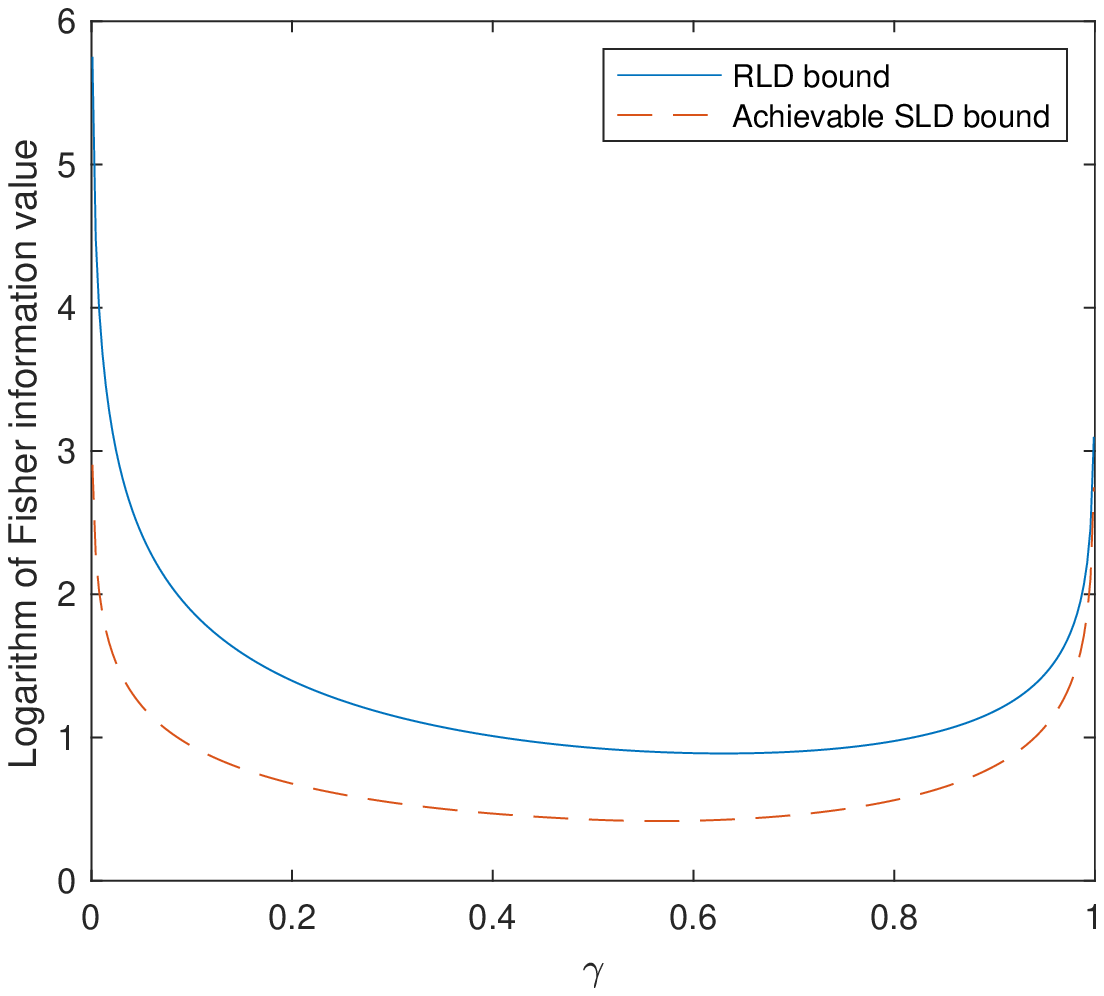}
\caption{}
\label{fig:estimate-loss-N-0.2}
\end{subfigure}
\begin{subfigure}{.5\textwidth}
\centering
\includegraphics[width=.9\linewidth]{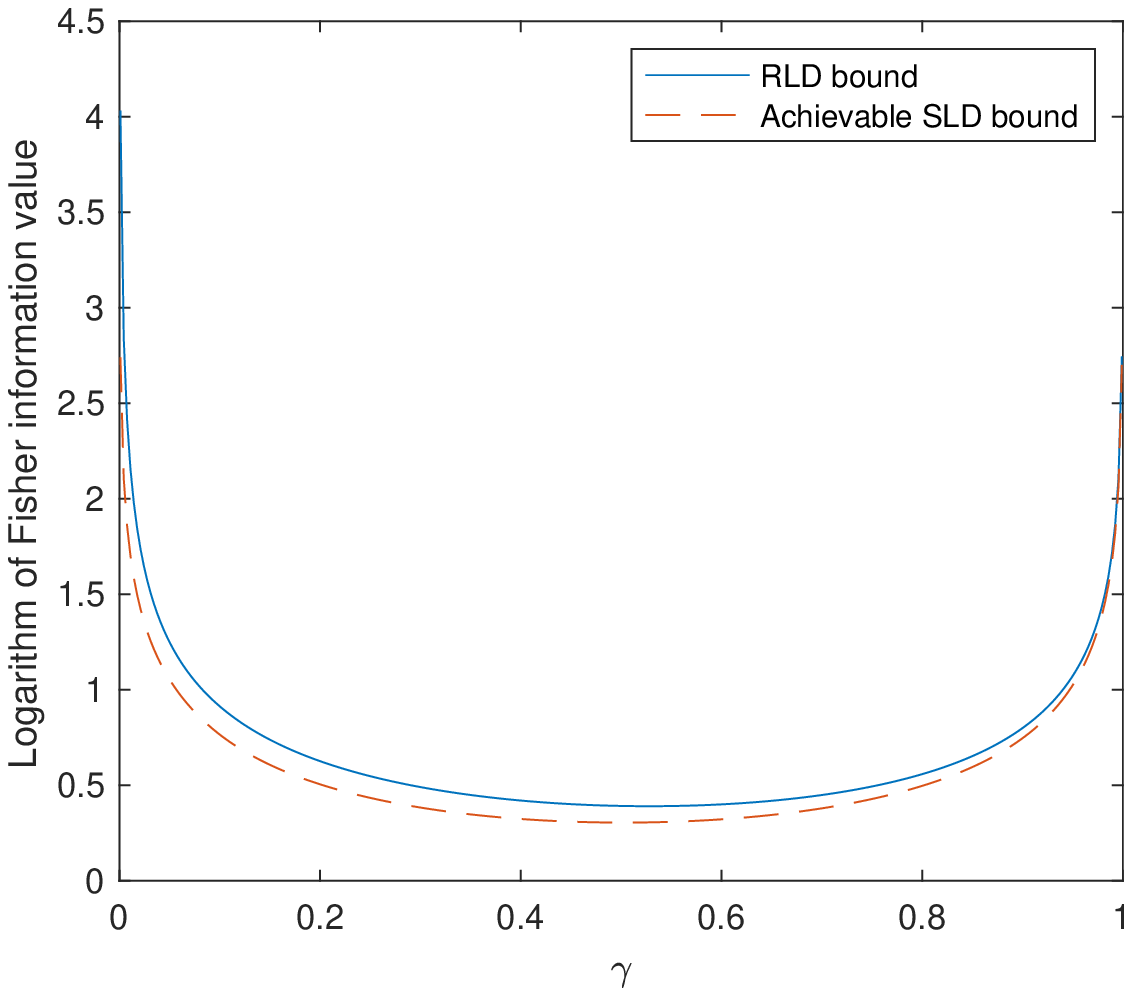}
\caption{}
\label{fig:estimate-loss-N-0.45}
\end{subfigure}
\caption{Comparing RLD and SLD bounds for loss estimation of a generalized amplitude damping channel}{(a)~Logarithm of RLD bound and achievable SLD bound versus loss
$\gamma$ for noise $N=0.2$, when estimating the loss $\gamma$. (b)~Logarithm
of the RLD bound and achievable SLD bound versus loss $\gamma$ for noise
$N=0.45$, when estimating the loss $\gamma$.} 
\label{fig:estimate-loss} 
\end{figure}

\subsection{Estimating noise}

Now suppose that we are interested in estimating the noise parameter $N$\ of a
generalized amplitude damping channel. We have that
\begin{equation}
\partial_{N}\Gamma_{RB}^{\mathcal{A}_{\gamma,N}}=-\gamma\left(  I_{2}
\otimes\sigma_{Z}\right)
\end{equation}
where $\sigma_z$ is the $2 \times 2$ matrix $\text{diag} (1,-1)$.
From \eqref{eq:GADC-inverse}, we have that
\begin{equation}
\operatorname{Tr}_{B}\left[  \left(  \partial_{N}\Gamma_{RB}^{\mathcal{A}
_{\gamma,N}}\right)  \left(  \Gamma_{RB}^{\mathcal{A}_{\gamma,N}}\right)
^{-1}\left(  \partial_{N}\Gamma_{RB}^{\mathcal{A}_{\gamma,N}}\right)  \right]
=
\begin{bmatrix}
\frac{1-\left(  1-2N\right)  \gamma}{\left(  1-N\right)  N} & 0\\
0 & \frac{1+\left(  1-2N\right)  \gamma}{\left(  1-N\right)  N}
\end{bmatrix}
.
\end{equation}
Thus, if $N>1/2$, then the first entry is the maximum, whereas if $N<1/2$,
then the second one is the maximum. We can summarize this as
\begin{equation}
\left\Vert \operatorname{Tr}_{B}\left[  \left(  \partial_{N}\Gamma
_{RB}^{\mathcal{A}_{\gamma,N}}\right)  \left(  \Gamma_{RB}^{\mathcal{A}
_{\gamma,N}}\right)  ^{-1}\left(  \partial_{N}\Gamma_{RB}^{\mathcal{A}
_{\gamma,N}}\right)  \right]  \right\Vert _{\infty}=\frac{1+\left\vert
1-2N\right\vert \gamma}{\left(  1-N\right)  N},
\end{equation}
so that
\begin{equation}
\widehat{I}_{F}(N;\{\mathcal{A}_{\gamma,N}\}_{N})=\frac{1+\left\vert
1-2N\right\vert \gamma}{\left(  1-N\right)  N}.
\end{equation}

For the loss parameter $\gamma$ equal to $0.5$ and $0.8$,
Figure~\ref{fig:estimate-noise} depicts the logarithm of the RLD bound, as
well as the logarithm of the achievable bound from the SLD\ Fisher information
of channels, corresponding to a parallel strategy that estimates $N$. The
RLD\ bound becomes better as $\gamma$ approaches $1$.

\begin{figure}[ptb]
\begin{subfigure}{.5\textwidth}
\centering
\includegraphics[width=.9\linewidth]{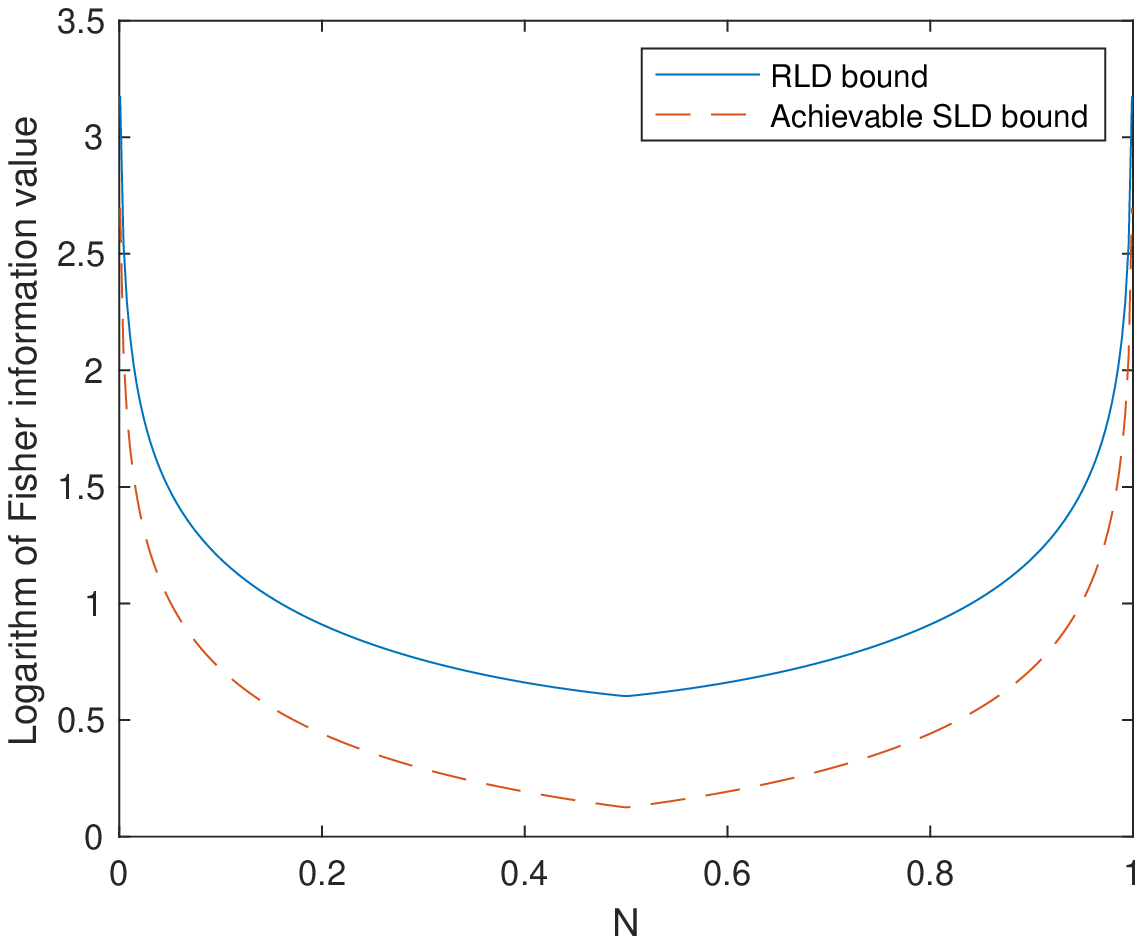}
\caption{}
\label{fig:estimate-noise-g-0.5}
\end{subfigure}
\begin{subfigure}{.5\textwidth}
\centering
\includegraphics[width=.9\linewidth]{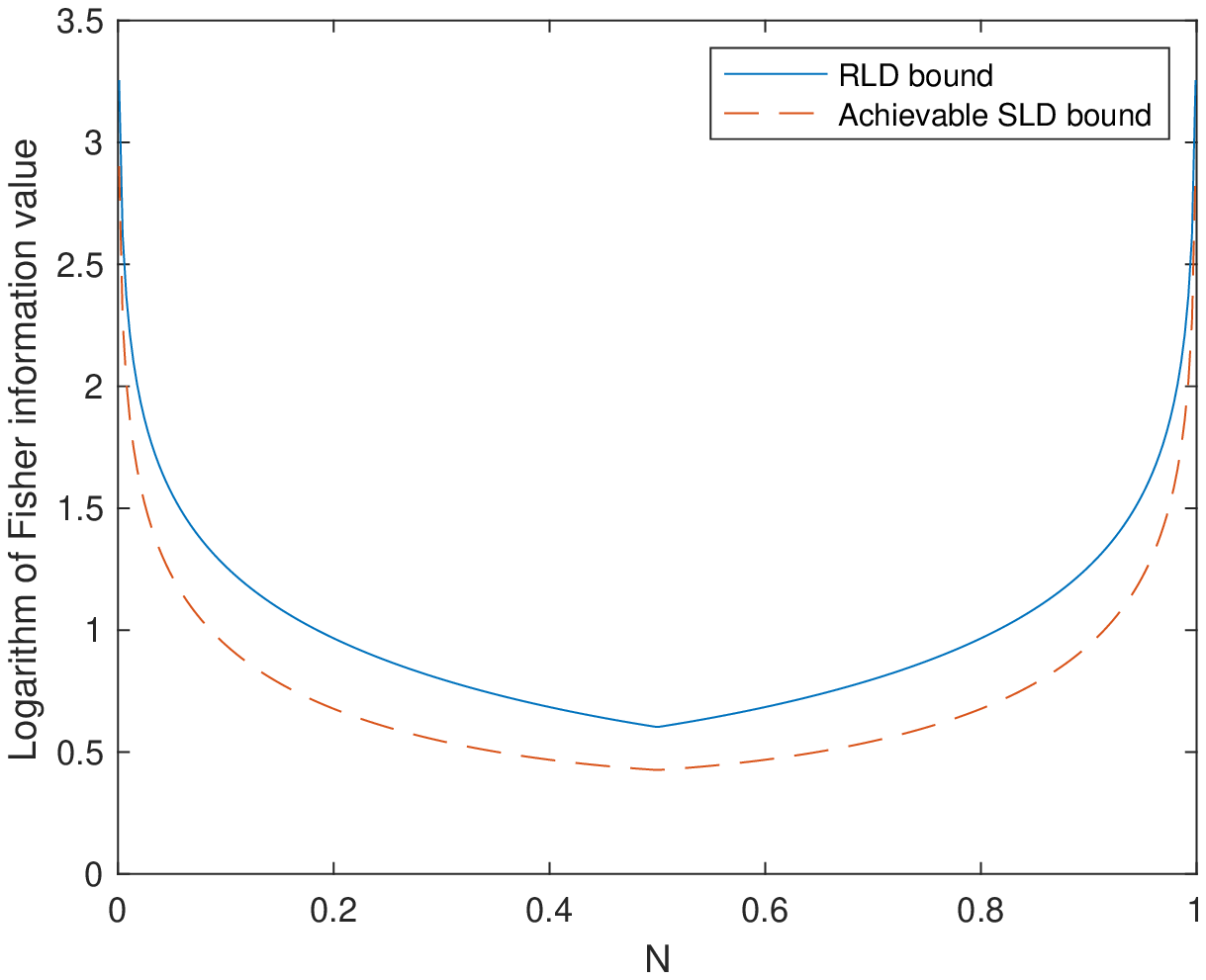}
\caption{}
\label{fig:estimate-loss-g-0.8}
\end{subfigure}
\caption{Comparing RLD and SLD bounds for noise estimation of a generalized amplitude damping channel}{(a)~Logarithm of RLD bound and achievable SLD bound versus noise $N$
for loss $\gamma=0.5$, when estimating the noise $N$. (b)~Logarithm of the RLD
bound and achievable SLD bound versus noise $N$ for loss $\gamma=0.8$, when
estimating the noise $N$.}
\label{fig:estimate-noise}
\end{figure}

\subsection{Estimating a phase in loss and noise}

Now let us suppose that we have a combination of a coherent process and the
generalized amplitude damping channel. In particular, let us suppose that a
phase $\phi$ is encoded in a unitary $e^{-i\phi\sigma_{Z}}$, and this is
followed by the generalized amplitude damping channel. Then this process is
\begin{equation}
\mathcal{A}_{\phi,\gamma,N}(\rho):=\mathcal{A}_{\gamma,N}(e^{-i\phi\sigma_{Z}
}\rho e^{i\phi\sigma_{Z}}).
\end{equation}
The goal is to estimate the phase $\phi$.

The Choi operator of the channel plus the phase shift is given by
\begin{equation}
\Gamma_{RB}^{\mathcal{A}_{\phi,\gamma,N}}:=
\begin{bmatrix}
1-\gamma N & 0 & 0 & e^{-i2\phi}\sqrt{1-\gamma}\\
0 & \gamma N & 0 & 0\\
0 & 0 & \gamma\left(  1-N\right)  & 0\\
e^{i2\phi}\sqrt{1-\gamma} & 0 & 0 & 1-\gamma\left(  1-N\right)
\end{bmatrix}
,
\end{equation}
and we find that
\begin{equation}
\partial_{\phi}\Gamma_{RB}^{\mathcal{A}_{\phi,\gamma,N}}=
\begin{bmatrix}
0 & 0 & 0 & -2ie^{-i2\phi}\sqrt{1-\gamma}\\
0 & 0 & 0 & 0\\
0 & 0 & 0 & 0\\
2ie^{i2\phi}\sqrt{1-\gamma} & 0 & 0 & 0
\end{bmatrix}
\end{equation}
Using the fact that
\begin{equation}
\left(  \Gamma_{RB}^{\mathcal{A}_{\phi}}\right)  ^{-1}=
\begin{bmatrix}
\frac{1-\gamma\left(  1-N\right)  }{\left(  1-N\right)  N\gamma^{2}} & 0 & 0 &
\frac{-e^{-2i\phi}\sqrt{1-\gamma}}{\left(  1-N\right)  N\gamma^{2}}\\
0 & \frac{1}{\gamma N} & 0 & 0\\
0 & 0 & \frac{1}{\gamma\left(  1-N\right)  } & 0\\
\frac{-e^{2i\phi}\sqrt{1-\gamma}}{\left(  1-N\right)  N\gamma^{2}} & 0 & 0 &
\frac{1-\gamma N}{\left(  1-N\right)  N\gamma^{2}}
\end{bmatrix}
,
\end{equation}
we find that
\begin{equation}
\operatorname{Tr}_{B}\left[  \left(  \partial_{\phi}\Gamma_{RB}^{\mathcal{A}
_{\phi,\gamma,N}}\right)  \left(  \Gamma_{RB}^{\mathcal{A}_{\phi,\gamma,N}
}\right)  ^{-1}\left(  \partial_{\phi}\Gamma_{RB}^{\mathcal{A}_{\phi,\gamma
,N}}\right)  \right]  =
\begin{bmatrix}
\frac{4\left(  1-\gamma\right)  \left(  1-\gamma N\right)  }{\left(
1-N\right)  N\gamma^{2}} & 0\\
0 & \frac{4\left(  1-\gamma\right)  \left(  1-\gamma\left(  1-N\right)
\right)  }{\left(  1-N\right)  N\gamma^{2}}
\end{bmatrix}
.
\end{equation}
Then if $N>1/2$, we have that
\begin{equation}
\left\Vert \operatorname{Tr}_{B}\left[  \left(  \partial_{\phi}\Gamma
_{RB}^{\mathcal{A}_{\phi,\gamma,N}}\right)  \left(  \Gamma_{RB}^{\mathcal{A}
_{\phi,\gamma,N}}\right)  ^{-1}\left(  \partial_{\phi}\Gamma_{RB}
^{\mathcal{A}_{\phi,\gamma,N}}\right)  \right]  \right\Vert _{\infty}
=\frac{4\left(  1-\gamma\right)  \left(  1-\gamma\left(  1-N\right)  \right)
}{\left(  1-N\right)  N\gamma^{2}},
\end{equation}
while if $N\leq1/2$, then
\begin{equation}
\left\Vert \operatorname{Tr}_{B}\left[  \left(  \partial_{\phi}\Gamma
_{RB}^{\mathcal{A}_{\phi,\gamma,N}}\right)  \left(  \Gamma_{RB}^{\mathcal{A}
_{\phi,\gamma,N}}\right)  ^{-1}\left(  \partial_{\phi}\Gamma_{RB}
^{\mathcal{A}_{\phi,\gamma,N}}\right)  \right]  \right\Vert _{\infty}
=\frac{4\left(  1-\gamma\right)  \left(  1-\gamma N\right)  }{\left(
1-N\right)  N\gamma^{2}}.
\end{equation}
So we conclude that
\begin{equation}
\widehat{I}_{F}(\phi;\{\mathcal{A}_{\phi,\gamma,N}\}_{\phi})=\frac{4\left(
1-\gamma\right)  \left(  1-\gamma\left(  N+\left(  1-2N\right)
u(2N-1)\right)  \right)  }{\left(  1-N\right)  N\gamma^{2}},
\end{equation}
where
\begin{equation}
u(x)=\left\{
\begin{array}
[c]{cc}
1 & x>0\\
0 & x\leq0
\end{array}
\right.  .
\end{equation}

For the noise parameter $N$ equal to $0.2$ and $0.45$,
Figure~\ref{fig:estimate-phase} depicts the logarithm of the RLD bound, as
well as the logarithm of the achievable bound from the SLD\ Fisher information
of channels, corresponding to a parallel strategy that estimates the phase
$\phi$ at $\phi=0.1$. The RLD\ bound becomes better as $\gamma$ approaches $1$.

\begin{figure}[ptb]
\begin{subfigure}{.5\textwidth}
\centering
\includegraphics[width=.9\linewidth]{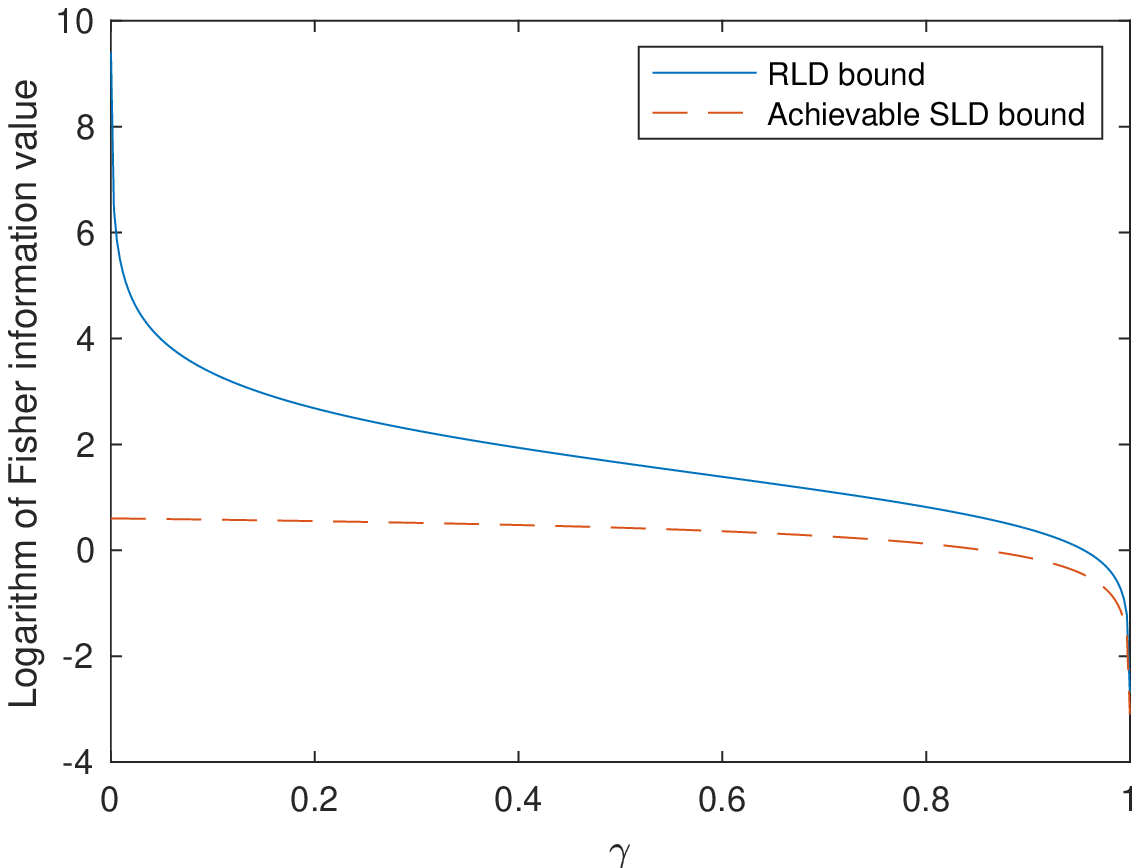}
\caption{}
\label{fig:estimate-phase-N-0.2}
\end{subfigure}
\begin{subfigure}{.5\textwidth}
\centering
\includegraphics[width=.9\linewidth]{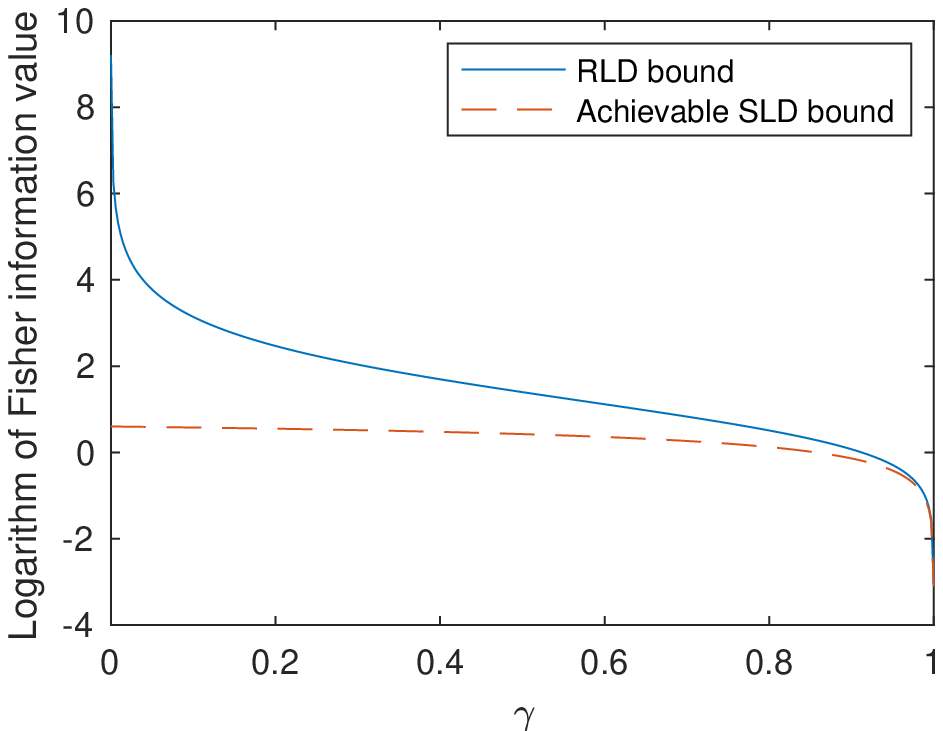}
\caption{}
\label{fig:estimate-phase-N-0.45}
\end{subfigure}
\caption{Comparing RLD and SLD bounds for estimating a phase in the presence of a generalized amplitude damping channel}{(a)~Logarithm of RLD bound and achievable SLD bound versus loss
$\gamma$ for noise $N=0.2$, when estimating the phase $\phi= 0.1$.
(b)~Logarithm of the RLD bound and achievable SLD bound versus loss $\gamma$
for noise $N=0.45$, when estimating the phase $\phi= 0.1$.}
\label{fig:estimate-phase}
\end{figure}

\section{Optimization/SDP forms of Fisher information quantities}

In the final section of this chapter, we provide optimization formulae for various quantities of interest in quantum estimation. In particular, we provide
\begin{itemize}
    \item a semi-definite program for the SLD Fisher information of quantum states,
    \item a quadratically constrained optimization problem for the root SLD Fisher information of quantum states,
    \item a semi-definite program for the RLD Fisher information of quantum states,
    \item a semi-definite program for the RLD Fisher information of quantum channels, and
    \item a bilinear program for the SLD Fisher information of quantum channels.
\end{itemize}  

First, we state the following technical lemma, which we use as a tool to derive the various optimization programs in this section.

\begin{lemma}
\label{lem:freq-used-SDP-primal-dual}Let $K$ and $Z$ be Hermitian operators,
and let $W$ be a linear operator. Then the dual of the following semi-definite
program
\begin{equation}
\inf_{M}\left\{  \operatorname{Tr}[KM]:
\begin{bmatrix}
M & W^{\dag}\\
W & Z
\end{bmatrix}
\geq0\right\}  ,
\end{equation}
with $M$ Hermitian, is given by
\begin{equation}
\sup_{P,Q,R}\left\{  2\operatorname{Re}(\operatorname{Tr}[W^{\dag
}Q])-\operatorname{Tr}[ZR]:P\leq K,
\begin{bmatrix}
P & Q^{\dag}\\
Q & R
\end{bmatrix}
\geq0\right\}  ,
\end{equation}
where $Q$ is a linear operator and $P$ and $R$ are Hermitian.
\end{lemma}

\begin{proof}
The standard forms of a\ primal and dual semi-definite program, for $A$ and
$B$ Hermitian and $\Phi$ a Hermiticity-preserving map, are respectively as
follows \cite{Watrous2018}:
\begin{align}
&  \inf_{Y\geq0}\left\{  \operatorname{Tr}[BY]:\Phi^{\dag}(Y)\geq A\right\}
,\\
&  \sup_{X\geq0}\left\{  \operatorname{Tr}[AX]:\Phi(X)\leq B\right\}  ,
\end{align}
where $\Phi^{\dag}$ is the Hilbert--Schmidt adjoint of $\Phi$. Noting that
\begin{equation}
\begin{bmatrix}
M & W^{\dag}\\
W & Z
\end{bmatrix}
\geq0\quad\Longleftrightarrow\quad
\begin{bmatrix}
M & -W^{\dag}\\
-W & Z
\end{bmatrix}
\geq0\quad\Longleftrightarrow\quad
\begin{bmatrix}
M & 0\\
0 & 0
\end{bmatrix}
\geq
\begin{bmatrix}
0 & W^{\dag}\\
W & -Z
\end{bmatrix}
,
\end{equation}
we conclude the statement of the lemma after making the following
identifications:
\begin{align}
B  &  =K,\quad Y=M,\quad\Phi^{\dag}(M)=
\begin{bmatrix}
M & 0\\
0 & 0
\end{bmatrix}
,\\
A  &  =
\begin{bmatrix}
0 & W^{\dag}\\
W & -Z
\end{bmatrix}
,\quad X=
\begin{bmatrix}
P & Q^{\dag}\\
Q & R
\end{bmatrix}
,\quad\Phi(X)=P.
\end{align}
This concludes the proof.
\end{proof}

\subsection{Semi-definite program for SLD\ Fisher information of quantum
states}

\label{sec:SDPs-for-FIs}We begin with the SLD\ Fisher information of quantum states,
establishing that it can be evaluated by means of a semi-definite program.

\begin{proposition}
\label{prop:SLD-Fish-states-SDP}The SLD\ Fisher information of a
differentiable family $\{\rho_{\theta}\}_{\theta}$ of states satisfying the
finiteness condition in \eqref{eq:basis-independent-formula-SLD} can be
evaluated by means of the following semi-definite program:
\begin{equation}
I_{F}(\theta;\{\rho_{\theta}\}_{\theta})=2\cdot\inf\left\{  \mu\in\mathbb{R}:
\begin{bmatrix}
\mu & \langle\Gamma|\left(  \partial_{\theta}\rho_{\theta}\otimes I\right) \\
\left(  \partial_{\theta}\rho_{\theta}\otimes I\right)  |\Gamma\rangle &
\rho_{\theta}\otimes I+I\otimes\rho_{\theta}^{T}
\end{bmatrix}
\geq0\right\}  .
\end{equation}
The dual semi-definite program is as follows:
\begin{equation}
2\cdot\sup_{\lambda,|\varphi\rangle,Z}2\operatorname{Re}[\langle
\varphi|\left(  \partial_{\theta}\rho_{\theta}\otimes I\right)  |\Gamma
\rangle]- \operatorname{Tr}[(\rho_{\theta}\otimes I+I\otimes\rho_{\theta}
^{T})Z],
\end{equation}
subject to $\lambda\in\mathbb{R}$, $|\varphi\rangle$ an arbitrary complex
vector, $Z$ Hermitian, and
\begin{equation}
\lambda\leq1,\qquad
\begin{bmatrix}
\lambda & \langle\varphi|\\
|\varphi\rangle & Z
\end{bmatrix}
\geq0.
\end{equation}

\end{proposition}

\begin{proof}
First, we recall the following alternate expression for the SLD Fisher information given in \cite{Safranek2018}:
\begin{align}
I_{F}(\theta;\{\rho_{\theta}\}_{\theta})  &  =2\sum_{j,k:\lambda_{j}^{\theta
}+\lambda_{k}^{\theta}>0}\frac{|\langle\psi_{\theta}^{j}|(\partial_{\theta
}\rho_{\theta})|\psi_{\theta}^{k}\rangle|^{2}}{\lambda_{\theta}^{j}
+\lambda_{\theta}^{k}}\label{eq:basis-dependent-SLD-formula}\\
&  =2\langle\Gamma|\left(  (\partial_{\theta}\rho_{\theta})\otimes I\right)
\left(  \rho_{\theta}\otimes I+I\otimes\rho_{\theta}^{T}\right)  ^{-1}\left(
(\partial_{\theta}\rho_{\theta})\otimes I\right)  |\Gamma\rangle
\label{eq:basis-independent-SLD-formula-extra}\\
&  =2\left\Vert \left(  \rho_{\theta}\otimes I+I\otimes\rho_{\theta}
^{T}\right)  ^{-\frac{1}{2}}\left(  (\partial_{\theta}\rho_{\theta})\otimes
I\right)  |\Gamma\rangle\right\Vert _{2}^{2}.
\label{eq:basis-independent-formula-SLD-2}
\end{align}
The primal semi-definite program is a direct consequence of the formula in
\eqref{eq:basis-independent-SLD-formula-extra}\ and Lemma~\ref{lem:min-XYinvX}. The dual program is a consequence of
Lemma~\ref{lem:freq-used-SDP-primal-dual}.
\end{proof}

\subsection{Root SLD Fisher information of quantum states as a
quadratically constrained optimization}

Next, we find that the root SLD\ Fisher information of quantum
states can be computed by means of a quadratically constrained optimization.
Such an optimization problem is difficult to solve in general, but heuristic
methods are available \cite{Park2017}. In any case, the particular optimization
formula we find in Proposition~\ref{prop:root-SLD-opt-formula}\ is helpful for
establishing the chain rule property of the root SLD\ Fisher information,
which we previously discussed in Section~\ref{subsec:root-sld-amort-collapse}.

\begin{proposition}
\label{prop:root-SLD-opt-formula}Let $\{\rho_{\theta}\}_{\theta}$ be a
differentiable family of quantum states. Then the root SLD\ Fisher information
can be written as the following optimization:
\begin{equation}
\sqrt{I_{F}}(\theta;\{\rho_{\theta}\}_{\theta})=\sqrt{2}\sup_{X}\left\{
\left\vert \operatorname{Tr}[X(\partial_{\theta}\rho_{\theta})]\right\vert
:\operatorname{Tr}[(XX^{\dag}+X^{\dag}X)\rho_{\theta}]\leq1\right\}  .
\label{eq:root-SLD-opt-formula}
\end{equation}
If the finiteness condition in \eqref{eq:basis-independent-formula-SLD}
is not satisfied, then the optimization formula evaluates to $+\infty$.
\end{proposition}

\begin{proof}
Let us begin by supposing that the finiteness condition in
\eqref{eq:basis-independent-formula-SLD} is satisfied (i.e., $\Pi
_{\rho_{\theta}}^{\perp}(\partial_{\theta}\rho_{\theta})\Pi_{\rho_{\theta}
}^{\perp}=0$). Recall from \eqref{eq:basis-independent-SLD-formula-extra} the
following formula for SLD\ Fisher information:
\begin{equation}
I_{F}(\theta;\{\rho_{\theta}\}_{\theta})=2\langle\Gamma|\left(  \partial
_{\theta}\rho_{\theta}\otimes I\right)  \left(  \rho_{\theta}\otimes
I+I\otimes\rho_{\theta}^{T}\right)  ^{-1}\left(  \partial_{\theta}\rho
_{\theta}\otimes I\right)  |\Gamma\rangle,
\end{equation}
so that
\begin{align}
&  \frac{1}{\sqrt{2}}\sqrt{I_{F}}(\theta;\{\rho_{\theta}\}_{\theta
})\nonumber\\
&  =\sqrt{\langle\Gamma|\left(  \partial_{\theta}\rho_{\theta}\otimes
I\right)  \left(  \rho_{\theta}\otimes I+I\otimes\rho_{\theta}^{T}\right)
^{-1}\left(  \partial_{\theta}\rho_{\theta}\otimes I\right)  |\Gamma\rangle}\\
&  =\left\Vert \left(  \rho_{\theta}\otimes I+I\otimes\rho_{\theta}
^{T}\right)  ^{-\frac{1}{2}}\left(  \partial_{\theta}\rho_{\theta}\otimes
I\right)  |\Gamma\rangle\right\Vert _{2}\\
&  =\sup_{|\psi\rangle:\left\Vert |\psi\rangle\right\Vert _{2}=1}\left\vert
\langle\psi|\left(  \rho_{\theta}\otimes I+I\otimes\rho_{\theta}^{T}\right)
^{-\frac{1}{2}}\left(  \partial_{\theta}\rho_{\theta}\otimes I\right)
|\Gamma\rangle\right\vert . \label{eq:sqrt-SLD-Fish-opt-alt}
\end{align}
Observe that the projection onto the support of $\rho_{\theta}\otimes
I+I\otimes\rho_{\theta}^{T}$ is
\begin{equation}
\Pi_{\rho_{\theta}}\otimes\Pi_{\rho_{\theta}^{T}}+\Pi_{\rho_{\theta}}^{\perp
}\otimes\Pi_{\rho_{\theta}^{T}}+\Pi_{\rho_{\theta}}\otimes\Pi_{\rho_{\theta
}^{T}}^{\perp}=I\otimes I-\Pi_{\rho_{\theta}}^{\perp}\otimes\Pi_{\rho_{\theta
}^{T}}^{\perp}.
\end{equation}
Thus, it suffices to optimize over $|\psi\rangle$ satisfying
\begin{equation}
|\psi\rangle=(I\otimes I-\Pi_{\rho_{\theta}}^{\perp}\otimes\Pi_{\rho_{\theta
}^{T}}^{\perp})|\psi\rangle
\end{equation}
because
\begin{multline}
\left(  \rho_{\theta}\otimes I+I\otimes\rho_{\theta}^{T}\right)  ^{-\frac
{1}{2}}\left(  \partial_{\theta}\rho_{\theta}\otimes I\right)  |\Gamma
\rangle\\
=(I\otimes I-\Pi_{\rho_{\theta}}^{\perp}\otimes\Pi_{\rho_{\theta}^{T}}^{\perp
})\left(  \rho_{\theta}\otimes I+I\otimes\rho_{\theta}^{T}\right)  ^{-\frac
{1}{2}}\left(  \partial_{\theta}\rho_{\theta}\otimes I\right)  |\Gamma\rangle.
\end{multline}
Now define
\begin{equation}
|\psi^{\prime}\rangle:=\left(  \rho_{\theta}\otimes I+I\otimes\rho_{\theta
}^{T}\right)  ^{-\frac{1}{2}}|\psi\rangle,
\end{equation}
which implies that
\begin{equation}
|\psi\rangle=(I\otimes I-\Pi_{\rho_{\theta}}^{\perp}\otimes\Pi_{\rho_{\theta
}^{T}}^{\perp})|\psi\rangle=\left(  \rho_{\theta}\otimes I+I\otimes
\rho_{\theta}^{T}\right)  ^{\frac{1}{2}}|\psi^{\prime}\rangle,
\end{equation}
because $I\otimes I-\Pi_{\rho_{\theta}}^{\perp}\otimes\Pi_{\rho_{\theta}^{T}
}^{\perp}$ is the projection onto the support of $\rho_{\theta}\otimes
I+I\otimes\rho_{\theta}^{T}$. Thus, the following equivalence holds
\begin{align}
\left\Vert |\psi\rangle\right\Vert _{2}=1\quad &  \Longleftrightarrow
\quad\left\Vert \left(  \rho_{\theta}\otimes I+I\otimes\rho_{\theta}
^{T}\right)  ^{\frac{1}{2}}|\psi^{\prime}\rangle\right\Vert _{2}=1\\
&  \Longleftrightarrow\quad\langle\psi^{\prime}|\left(  \rho_{\theta}\otimes
I+I\otimes\rho_{\theta}^{T}\right)  |\psi^{\prime}\rangle=1.
\end{align}
Now fix the operator $X$ such that
\begin{equation}
|\psi^{\prime}\rangle=\left(  X\otimes I\right)  |\Gamma\rangle.
\end{equation}
Then the last condition above is the same as the following:
\begin{align}
1  &  =\langle\Gamma|\left(  X^{\dag}\otimes I\right)  \left(  \rho_{\theta
}\otimes I+I\otimes\rho_{\theta}^{T}\right)  \left(  X\otimes I\right)
|\Gamma\rangle\\
&  =\langle\Gamma|\left(  X^{\dag}\rho_{\theta}X\otimes I+X^{\dag}X\otimes
\rho_{\theta}^{T}\right)  |\Gamma\rangle\\
&  =\langle\Gamma|\left(  X^{\dag}\rho_{\theta}X\otimes I+X^{\dag}
X\rho_{\theta}\otimes I\right)  |\Gamma\rangle\\
&  =\operatorname{Tr}[X^{\dag}\rho_{\theta}X]+\operatorname{Tr}[X^{\dag}
X\rho_{\theta}]\\
&  =\operatorname{Tr}[(XX^{\dag}+X^{\dag}X)\rho_{\theta}],
\end{align}
where we used \eqref{eq:transpose-trick}\ and
\eqref{eq:max-ent-partial-trace} from Chapter~\ref{ch:prelims}. So then the optimization problem in \eqref{eq:sqrt-SLD-Fish-opt-alt}\ is equal to the following:
\begin{align}
&  \sup_{X:\operatorname{Tr}[(XX^{\dag}+X^{\dag}X)\rho_{\theta}]=1}\left\vert
\langle\Gamma|\left(  X\otimes I\right)  \left(  \partial_{\theta}\rho
_{\theta}\otimes I\right)  |\Gamma\rangle\right\vert \nonumber\\
&  =\sup_{X:\operatorname{Tr}[(XX^{\dag}+X^{\dag}X)\rho_{\theta}]=1}\left\vert
\langle\Gamma|\left(  X(\partial_{\theta}\rho_{\theta})\otimes I\right)
|\Gamma\rangle\right\vert \\
&  =\sup_{X}\left\{  \left\vert \operatorname{Tr}[X(\partial_{\theta}
\rho_{\theta})]\right\vert :\operatorname{Tr}[(XX^{\dag}+X^{\dag}
X)\rho_{\theta}]=1\right\}  ,
\end{align}
where again we used \eqref{eq:max-ent-partial-trace}. Now suppose that
$\operatorname{Tr}[(XX^{\dag}+X^{\dag}X)\rho_{\theta}]=c$, with $c\in(0,1)$.
Then we can multiply $X$ by $\sqrt{1/c}$, and the new operator satisfies the
equality constraint while the value of the objective function increases. So we
can write
\begin{equation}
\sqrt{I_{F}}(\theta;\{\rho_{\theta}\}_{\theta})=\sqrt{2}\sup_{X}\left\{
\left\vert \operatorname{Tr}[X(\partial_{\theta}\rho_{\theta})]\right\vert
:\operatorname{Tr}[(XX^{\dag}+X^{\dag}X)\rho_{\theta}]\leq1\right\}  .
\end{equation}
Finally, in this form, note that we can trivially include $X=0$ as part of the
optimization because it leads to a generally suboptimal value of zero for the
objective function.

Suppose that $\Pi_{\rho_{\theta}}^{\perp}(\partial_{\theta}\rho_{\theta}
)\Pi_{\rho_{\theta}}^{\perp}\neq0$. Then we can pick $X=c\Pi_{\rho_{\theta}
}^{\perp}+dI$ where $c,d>0$ and $2d^{2}=1$. We find that
\begin{align}
\operatorname{Tr}[(XX^{\dag}+X^{\dag}X)\rho_{\theta}]  &  =2\operatorname{Tr}
[\left(  c\Pi_{\rho_{\theta}}^{\perp}+dI\right)  ^{2}\rho_{\theta}]\\
&  =2\operatorname{Tr}[\left(  \left[  c^{2}+2cd\right]  \Pi_{\rho_{\theta}
}^{\perp}+d^{2}I\right)  \rho_{\theta}]\\
&  =2d^{2}=1.
\end{align}
for this case, so that the constraint in \eqref{eq:root-SLD-opt-formula}\ is
satisfied. The objective function then evaluates to
\begin{align}
\left\vert \operatorname{Tr}[X(\partial_{\theta}\rho_{\theta})]\right\vert  &
=\left\vert \operatorname{Tr}[\left(  c\Pi_{\rho_{\theta}}^{\perp}+dI\right)
(\partial_{\theta}\rho_{\theta})]\right\vert \\
&  =\left\vert c\operatorname{Tr}[\Pi_{\rho_{\theta}}^{\perp}(\partial
_{\theta}\rho_{\theta})]+d\operatorname{Tr}[\partial_{\theta}\rho_{\theta
}]\right\vert \\
&  =c\left\vert \operatorname{Tr}[\Pi_{\rho_{\theta}}^{\perp}(\partial
_{\theta}\rho_{\theta})]\right\vert .
\end{align}
Then we can pick $c>0$ arbitrarily large to get that
\eqref{eq:root-SLD-opt-formula} evaluates to $+\infty$ in the case that
$\Pi_{\rho_{\theta}}^{\perp}(\partial_{\theta}\rho_{\theta})\Pi_{\rho_{\theta
}}^{\perp}\neq0$.
\end{proof}

\subsection{Semi-definite programs for RLD Fisher information of quantum
states and channels}

First we give a semi-definite program for the RLD Fisher information of quantum states:

\begin{proposition}
The RLD\ Fisher information of a differentiable family $\{\rho_{\theta
}\}_{\theta}$ of states satisfying the support condition in
\eqref{eq:RLD-FI} can be evaluated by means of the following
semi-definite program:
\begin{equation}
\widehat{I}_{F}(\theta;\{\rho_{A}^{\theta}\})=\inf\left\{  \operatorname{Tr}
[M]:M\geq0,
\begin{bmatrix}
M & \partial_{\theta}\rho_{\theta}\\
\partial_{\theta}\rho_{\theta} & \rho_{\theta}
\end{bmatrix}
\geq0\right\}  .
\end{equation}
The dual semi-definite program is as follows:
\begin{equation}
\sup_{X,Y,Z}2\operatorname{Re}[\operatorname{Tr}[Y(\partial_{\theta}
\rho_{\theta})]]-\operatorname{Tr}[Z\rho_{\theta}],
\end{equation}
subject to $X$ and $Y$ being Hermitian and
\begin{equation}
X\leq I,\qquad
\begin{bmatrix}
X & Y^{\dag}\\
Y & Z
\end{bmatrix}
\geq0.
\end{equation}

\end{proposition}

\begin{proof}
The primal semi-definite program is a direct consequence of the formula of the RLD Fisher information in
\eqref{eq:RLD-FI}\ and Lemma~\ref{lem:min-XYinvX}. The dual program is found
by applying Lemma~\ref{lem:freq-used-SDP-primal-dual}.
\end{proof}

Using the explicit formula of the RLD Fisher information of quantum channels as defined in~\eqref{eq:RLD-Fish-ch}, we find the following semi-definite program for the RLD\ Fisher
information of quantum channels:

\begin{proposition}
Let $\{\mathcal{N}_{A\rightarrow B}^{\theta}\}_{\theta}$ be a differentiable
family of quantum channels such that the support condition in
\eqref{eq:RLD-Fish-ch} holds. Then the RLD Fisher
information of quantum channels can be calculated by means of the following
semi-definite program:
\begin{equation}
\widehat{I}_{F}(\theta;\{\mathcal{N}_{A\rightarrow B}^{\theta}\}_{\theta
})=\inf\lambda\in\mathbb{R}^{+}, \label{eq:SDP-SLD-fish-channels}
\end{equation}
subject to
\begin{equation}
\lambda I_{R}\geq\operatorname{Tr}_{B}[M_{RB}],\qquad
\begin{bmatrix}
M_{RB} & \partial_{\theta}\Gamma_{RB}^{\mathcal{N}^{\theta}}\\
\partial_{\theta}\Gamma_{RB}^{\mathcal{N}^{\theta}} & \Gamma_{RB}
^{\mathcal{N}^{\theta}}
\end{bmatrix}
\geq0. \label{eq:SDP-SLD-fish-channels-2}
\end{equation}
The dual program is given by
\begin{equation}
\sup_{\rho_{R}\geq0,P_{RB},Z_{RB},Q_{RB}}2\operatorname{Re}[\operatorname{Tr}
[Z_{RB}(\partial_{\theta}\Gamma_{RB}^{\mathcal{N}^{\theta}}
)]]-\operatorname{Tr}[Q_{RB}\Gamma_{RB}^{\mathcal{N}^{\theta}}],
\end{equation}
subject to
\begin{equation}
\operatorname{Tr}[\rho_{R}]\leq1,\quad
\begin{bmatrix}
P_{RB} & Z_{RB}^{\dag}\\
Z_{RB} & Q_{RB}
\end{bmatrix}
\geq0,\quad P_{RB}\leq\rho_{R}\otimes I_{B}.
\end{equation}

\end{proposition}

\begin{proof}
The form of the primal program follows directly from \eqref{eq:RLD-Fish-ch},
Lemma~\ref{lem:min-XYinvX}, and from the following characterization of the
infinity norm of a positive semi-definite operator $W$:
\begin{equation}
\left\Vert W\right\Vert _{\infty}=\inf\left\{  \lambda\geq0:W\leq\lambda
I\right\}  .
\end{equation}

To arrive at the dual program, we use the standard forms of primal and dual
semi-definite programs for Hermitian operators $A$ and $B$ and a
Hermiticity-preserving map $\Phi$ \cite{Watrous2018}:
\begin{equation}
\sup_{X\geq0}\left\{  \operatorname{Tr}[AX]:\Phi(X)\leq B\right\}  ,
\qquad\inf_{Y\geq0}\left\{  \operatorname{Tr}[BY]:\Phi^{\dag}(Y)\geq
A\right\}  . \label{eq:standard-SDP-form-RLD-ch-helper}
\end{equation}
From \eqref{eq:SDP-SLD-fish-channels}--\eqref{eq:SDP-SLD-fish-channels-2}, we
identify
\begin{align}
B  &  =
\begin{bmatrix}
1 & 0\\
0 & 0
\end{bmatrix}
,\quad Y=
\begin{bmatrix}
\lambda & 0\\
0 & M_{RB}
\end{bmatrix}
,\quad\Phi^{\dag}(Y)=
\begin{bmatrix}
\lambda I_{R}-\operatorname{Tr}_{B}[M_{RB}] & 0 & 0\\
0 & M_{RB} & 0\\
0 & 0 & 0
\end{bmatrix}
,\\
A  &  =
\begin{bmatrix}
0 & 0 & 0\\
0 & 0 & -\partial_{\theta}\Gamma_{RB}^{\mathcal{N}^{\theta}}\\
0 & -\partial_{\theta}\Gamma_{RB}^{\mathcal{N}^{\theta}} & -\Gamma
_{RB}^{\mathcal{N}^{\theta}}
\end{bmatrix}
.
\end{align}
Setting
\begin{equation}
X=
\begin{bmatrix}
\rho_{R} & 0 & 0\\
0 & P_{RB} & Z_{RB}^{\dag}\\
0 & Z_{RB} & Q_{RB}
\end{bmatrix}
,
\end{equation}
we find that
\begin{align}
\operatorname{Tr}[X\Phi^{\dag}(Y)]  &  =\operatorname{Tr}\left[
\begin{bmatrix}
\rho_{R} & 0 & 0\\
0 & P_{RB} & Z_{RB}^{\dag}\\
0 & Z_{RB} & Q_{RB}
\end{bmatrix}
\begin{bmatrix}
\lambda I_{R}-\operatorname{Tr}_{B}[M_{RB}] & 0 & 0\\
0 & M_{RB} & 0\\
0 & 0 & 0
\end{bmatrix}
\right] \\
&  =\operatorname{Tr}[\rho_{R}(\lambda I_{R}-\operatorname{Tr}_{B}
[M_{RB}])]+\operatorname{Tr}[P_{RB}M_{RB}]\\
&  =\lambda\operatorname{Tr}[\rho_{R}]+\operatorname{Tr}[(P_{RB}-\rho
_{R}\otimes I_{B})M_{RB}]\\
&  =\operatorname{Tr}\left[
\begin{bmatrix}
\lambda & 0\\
0 & M_{RB}
\end{bmatrix}
\begin{bmatrix}
\operatorname{Tr}[\rho_{R}] & 0\\
0 & P_{RB}-\rho_{R}\otimes I_{B}
\end{bmatrix}
\right]  ,
\end{align}
which implies that
\begin{equation}
\Phi(X)=
\begin{bmatrix}
\operatorname{Tr}[\rho_{R}] & 0\\
0 & P_{RB}-\rho_{R}\otimes I_{B}
\end{bmatrix}
.
\end{equation}
Then plugging into the left-hand side of
\eqref{eq:standard-SDP-form-RLD-ch-helper}, we find that the dual is given by
\begin{equation}
\sup_{\rho_{R},P_{RB},Z_{RB},Q_{RB}}\operatorname{Tr}\left[
\begin{bmatrix}
0 & 0 & 0\\
0 & 0 & -\partial_{\theta}\Gamma_{RB}^{\mathcal{N}^{\theta}}\\
0 & -\partial_{\theta}\Gamma_{RB}^{\mathcal{N}^{\theta}} & -\Gamma
_{RB}^{\mathcal{N}^{\theta}}
\end{bmatrix}
\begin{bmatrix}
W_{R} & 0 & 0\\
0 & P_{RB} & Z_{RB}^{\dag}\\
0 & Z_{RB} & Q_{RB}
\end{bmatrix}
\right]  ,
\end{equation}
subject to
\begin{equation}
\begin{bmatrix}
\rho_{R} & 0 & 0\\
0 & P_{RB} & Z_{RB}^{\dag}\\
0 & Z_{RB} & Q_{RB}
\end{bmatrix}
\geq0,\qquad
\begin{bmatrix}
\operatorname{Tr}[\rho_{R}] & 0\\
0 & P_{RB}-\rho_{R}\otimes I_{B}
\end{bmatrix}
\leq
\begin{bmatrix}
1 & 0\\
0 & 0
\end{bmatrix}
.
\end{equation}
Upon making the swap $Z_{RB}\rightarrow-Z_{RB}$, which does not change the
optimal value, and simplifying, we find the following form:
\begin{equation}
\sup_{\rho_{R}\geq0,P_{RB},Z_{RB},Q_{RB}}2\operatorname{Re}[\operatorname{Tr}
[Z_{RB}(\partial_{\theta}\Gamma_{RB}^{\mathcal{N}^{\theta}}
)]]-\operatorname{Tr}[Q_{RB}\Gamma_{RB}^{\mathcal{N}^{\theta}}],
\end{equation}
subject to
\begin{equation}
\operatorname{Tr}[\rho_{R}]\leq1,\quad
\begin{bmatrix}
P_{RB} & -Z_{RB}^{\dag}\\
-Z_{RB} & Q_{RB}
\end{bmatrix}
\geq0,\quad P_{RB}\leq\rho_{R}\otimes I_{B}.
\end{equation}
Then we note that
\begin{equation}
\begin{bmatrix}
P_{RB} & -Z_{RB}^{\dag}\\
-Z_{RB} & Q_{RB}
\end{bmatrix}
\geq0 \quad\Longleftrightarrow\quad
\begin{bmatrix}
P_{RB} & Z_{RB}^{\dag}\\
Z_{RB} & Q_{RB}
\end{bmatrix}
\geq0
\end{equation}
This concludes the proof.
\end{proof}

\subsection{Bilinear program for SLD\ Fisher information of quantum
channels}

Finally, we provide a bilinear program for the SLD Fisher information of quantum channels. To do so, we use Proposition~\ref{prop:SLD-Fish-states-SDP} and a number of subseqeuent manipulations:

\begin{proposition}
\label{prop:SDP-SLD-Fish}The SLD Fisher information of a differentiable family
$\{\mathcal{N}_{A\rightarrow B}^{\theta}\}_{\theta}$ of channels satisfying
the finiteness condition in \eqref{eq:finiteness-condition-SLD-fish-ch} can be
evaluated by means of the following bilinear program:
\begin{multline}
I_{F}(\theta;\{\mathcal{N}_{A\rightarrow B}^{\theta}
\})=\label{eq:SDP-SLD-ch-Fish}\\
2\sup_{\substack{\lambda,|\varphi\rangle_{RBR^{\prime}B^{\prime}
},\\W_{RBR^{\prime}B^{\prime}},Y_{R},\sigma_{R}}}\left(  2\operatorname{Re}
[\langle\varphi|_{RBR^{\prime}B^{\prime}}(\partial_{\theta}\Gamma
_{RB}^{\mathcal{N}^{\theta}})|\Gamma\rangle_{RR^{\prime}BB^{\prime}
}]-\operatorname{Tr}[Y_{R}\Phi(W_{RBR^{\prime}B^{\prime}})]\right)
\end{multline}
subject to
\begin{equation}
\sigma_{R}\geq0,\quad\operatorname{Tr}[\sigma_{R}]=1,\quad\lambda\leq1,\quad
\begin{bmatrix}
\lambda & \langle\varphi|_{RBR^{\prime}B^{\prime}}\\
|\varphi\rangle_{RBR^{\prime}B^{\prime}} & W_{RBR^{\prime}B^{\prime}}
\end{bmatrix}
\geq0,\quad
\begin{bmatrix}
\sigma_{R} & I_{R}\\
I_{R} & Y_{R}
\end{bmatrix}
\geq0.
\end{equation}
where
\begin{align}
|\Gamma\rangle_{RR^{\prime}BB^{\prime}}  &  :=|\Gamma\rangle_{RR^{\prime}
}\otimes|\Gamma\rangle_{BB^{\prime}},\\
\Phi(W_{RBR^{\prime}B^{\prime}})  &  :=(\operatorname{Tr}_{BR^{\prime
}B^{\prime}}[\Gamma_{R^{\prime}B^{\prime}}^{\mathcal{N}^{\theta}}\left(
F_{RR^{\prime}}\otimes F_{BB^{\prime}}\right)  W_{RBR^{\prime}B^{\prime}
}\left(  F_{RR^{\prime}}\otimes F_{BB^{\prime}}\right)  ^{\dag}])^{T}
\nonumber\\
&  \qquad\qquad+\operatorname{Tr}_{BR^{\prime}B^{\prime}}[(\Gamma_{R^{\prime
}B^{\prime}}^{\mathcal{N}^{\theta}})^{T}W_{RBR^{\prime}B^{\prime}}],
\end{align}
and $F_{RR^{\prime}}$ is the flip or swap operator that swaps systems $R$ and
$R^{\prime}$, with a similar definition for $F_{BB^{\prime}}$ but for $B$ and
$B^{\prime}$.
\end{proposition}

The optimization above is a jointly constrained semi-definite bilinear program
\cite{Huber2018}\ because the variables $Y_{R}$ and $W_{RBR^{\prime}B^{\prime}}$
are operators involved in the optimization and they multiply each other in the
last expression in \eqref{eq:SDP-SLD-ch-Fish}. This kind of optimization can
be approached with a heuristic \textquotedblleft seesaw\textquotedblright
\ method, and more advanced methods are available in~\cite{Huber2018}.

\begin{proof}
Recall that the Fisher information of channels is defined as the following
optimization over pure state inputs:
\begin{equation}
I_{F}(\theta;\{\mathcal{N}_{A\rightarrow B}^{\theta}\})=\sup_{\psi_{RA}}
I_{F}(\theta;\{\mathcal{N}_{A\rightarrow B}^{\theta}(\psi_{RA})\}).
\end{equation}
It suffices to optimize over pure state inputs $\psi_{RA}$ such that the
reduced state $\psi_{R}>0$, because this set is dense in the set of all pure
bipartite states. Now consider a fixed input state $\psi_{RA}$, and recall
that it can be written as follows:
\begin{equation}
\psi_{RA}=Z_{R}\Gamma_{RA}Z_{R}^{\dag},
\end{equation}
where $Z_{R}$ is an invertible operator satisfying $\operatorname{Tr}
[Z_{R}^{\dag}Z_{R}]=1$. Then the output state is as follows:
\begin{equation}
\omega_{RB}^{\theta}:=\mathcal{N}_{A\rightarrow B}^{\theta}(\psi_{RA}
)=Z_{R}\Gamma_{RB}^{\mathcal{N}^{\theta}}Z_{R}^{\dag},
\end{equation}
and we find that
\begin{multline}
\frac{1}{2}I_{F}(\theta;\{\mathcal{N}_{A\rightarrow B}^{\theta}(\psi
_{RA})\})\label{eq:SDP-SLD-Fisher-fixed-state}\\
=\inf\left\{  \mu:
\begin{bmatrix}
\mu & \langle\Gamma|_{RR^{\prime}BB^{\prime}}\left(  \partial_{\theta}
\omega_{RB}^{\theta}\otimes I_{R^{\prime}B^{\prime}}\right) \\
\left(  \partial_{\theta}\omega_{RB}^{\theta}\otimes I_{R^{\prime}B^{\prime}
}\right)  |\Gamma\rangle_{RR^{\prime}BB^{\prime}} & \omega_{RB}^{\theta
}\otimes I_{R^{\prime}B^{\prime}}+I_{RB}\otimes(\omega_{R^{\prime}B^{\prime}
}^{\theta})^{T}
\end{bmatrix}
\geq0\right\}  ,
\end{multline}
by applying Proposition~\ref{prop:SLD-Fish-states-SDP}.

Now consider that
\begin{align}
&
\begin{bmatrix}
\mu & \langle\Gamma|_{RR^{\prime}BB^{\prime}}\left(  \partial_{\theta}
\omega_{RB}^{\theta}\otimes I_{R^{\prime}B^{\prime}}\right) \\
\left(  \partial_{\theta}\omega_{RB}^{\theta}\otimes I_{R^{\prime}B^{\prime}
}\right)  |\Gamma\rangle_{RR^{\prime}BB^{\prime}} & \omega_{RB}^{\theta
}\otimes I_{R^{\prime}B^{\prime}}+I_{RB}\otimes(\omega_{R^{\prime}B^{\prime}
}^{\theta})^{T}
\end{bmatrix}
\nonumber\\
&  =
\begin{bmatrix}
\mu & \langle\Gamma|_{RR^{\prime}BB^{\prime}}\left(  Z_{R}(\partial_{\theta
}\Gamma_{RB}^{\mathcal{N}^{\theta}})Z_{R}^{\dag}\otimes I_{R^{\prime}
B^{\prime}}\right) \\
\left(  Z_{R}(\partial_{\theta}\Gamma_{RB}^{\mathcal{N}^{\theta}})Z_{R}^{\dag
}\otimes I_{R^{\prime}B^{\prime}}\right)  |\Gamma\rangle_{RR^{\prime
}BB^{\prime}} & Z_{R}\Gamma_{RB}^{\mathcal{N}^{\theta}}Z_{R}^{\dag}\otimes
I_{R^{\prime}B^{\prime}}+I_{RB}\otimes\overline{Z}_{R^{\prime}}(\Gamma
_{R^{\prime}B^{\prime}}^{\mathcal{N}^{\theta}})^{T}Z_{R^{\prime}}^{T}
\end{bmatrix}
\label{eq:bilinear-SLD-Fish-ch-proof-1}\\
&  =
\begin{bmatrix}
1 & 0\\
0 & Z_{R}\otimes I_{B}\otimes\overline{Z}_{R^{\prime}}\otimes I_{B^{\prime}}
\end{bmatrix}
\times\nonumber\\
&  \qquad
\begin{bmatrix}
\mu & \langle\Gamma|_{RR^{\prime}BB^{\prime}}\left(  (\partial_{\theta}
\Gamma_{RB}^{\mathcal{N}^{\theta}})\otimes I_{R^{\prime}B^{\prime}}\right) \\
\left(  (\partial_{\theta}\Gamma_{RB}^{\mathcal{N}^{\theta}})\otimes
I_{R^{\prime}B^{\prime}}\right)  |\Gamma\rangle_{RR^{\prime}BB^{\prime}} &
\Gamma_{RB}^{\mathcal{N}^{\theta}}\otimes\sigma_{R^{\prime}}^{-T}\otimes
I_{B^{\prime}}+\sigma_{R}^{-1}\otimes I_{B}\otimes(\Gamma_{R^{\prime}
B^{\prime}}^{\mathcal{N}^{\theta}})^{T}
\end{bmatrix}
\times\nonumber\\
&  \qquad
\begin{bmatrix}
1 & 0\\
0 & Z_{R}\otimes I_{B}\otimes\overline{Z}_{R^{\prime}}\otimes I_{B^{\prime}}
\end{bmatrix}
^{\dag}, \label{eq:bilinear-SLD-Fish-ch-proof-2}
\end{align}
where we define
\begin{equation}
\sigma_{R}:=Z_{R}^{\dag}Z_{R},
\end{equation}
and we applied the following observations:
\begin{align}
&  \left(  Z_{R}(\partial_{\theta}\Gamma_{RB}^{\mathcal{N}^{\theta}}
)Z_{R}^{\dag}\otimes I_{R^{\prime}B^{\prime}}\right)  |\Gamma\rangle
_{RR^{\prime}BB^{\prime}}\nonumber\\
&  =\left(  Z_{R}(\partial_{\theta}\Gamma_{RB}^{\mathcal{N}^{\theta}}
)\otimes\overline{Z}_{R^{\prime}}\otimes I_{B^{\prime}}\right)  |\Gamma
\rangle_{RR^{\prime}BB^{\prime}}\\
&  =\left(  Z_{R}\otimes\overline{Z}_{R^{\prime}}\right)  \left(
(\partial_{\theta}\Gamma_{RB}^{\mathcal{N}^{\theta}})\otimes I_{R^{\prime
}B^{\prime}}\right)  |\Gamma\rangle_{RR^{\prime}BB^{\prime}},
\end{align}
\begin{align}
&  Z_{R}\Gamma_{RB}^{\mathcal{N}^{\theta}}Z_{R}^{\dag}\otimes I_{R^{\prime
}B^{\prime}}+I_{RB}\otimes\overline{Z}_{R^{\prime}}(\Gamma_{R^{\prime
}B^{\prime}}^{\mathcal{N}^{\theta}})^{T}Z_{R^{\prime}}^{T}\nonumber\\
&  =Z_{R}\Gamma_{RB}^{\mathcal{N}^{\theta}}Z_{R}^{\dag}\otimes\overline
{Z}_{R^{\prime}}\left(  \overline{Z}_{R^{\prime}}\right)  ^{-1}\left(
Z_{R^{\prime}}^{T}\right)  ^{-1}Z_{R^{\prime}}^{T}\otimes I_{B^{\prime}
}\nonumber\\
&  \quad+Z_{R}\left(  Z_{R}\right)  ^{-1}\left(  Z_{R}^{\dag}\right)
^{-1}Z_{R}^{\dag}\otimes I_{B}\otimes\overline{Z}_{R^{\prime}}(\Gamma
_{R^{\prime}B^{\prime}}^{\mathcal{N}^{\theta}})^{T}Z_{R^{\prime}}^{T}\\
&  =Z_{R}\Gamma_{RB}^{\mathcal{N}^{\theta}}Z_{R}^{\dag}\otimes\overline
{Z}_{R^{\prime}}\sigma_{R}^{-T}Z_{R^{\prime}}^{T}\otimes I_{B^{\prime}}
+Z_{R}\sigma_{R}^{-1}Z_{R}^{\dag}\otimes I_{B}\otimes\overline{Z}_{R^{\prime}
}(\Gamma_{R^{\prime}B^{\prime}}^{\mathcal{N}^{\theta}})^{T}Z_{R^{\prime}}
^{T}\\
&  =\left(  Z_{R}\otimes I_{B}\otimes\overline{Z}_{R^{\prime}}\otimes
I_{B^{\prime}}\right)  \left(  \Gamma_{RB}^{\mathcal{N}^{\theta}}\otimes
\sigma_{R}^{-T}\otimes I_{B^{\prime}}\right)  \left(  Z_{R}\otimes
I_{B}\otimes\overline{Z}_{R^{\prime}}\otimes I_{B^{\prime}}\right)  ^{\dag
}\nonumber\\
&  \quad+\left(  Z_{R}\otimes I_{B}\otimes\overline{Z}_{R^{\prime}}\otimes
I_{B^{\prime}}\right)  \left(  \sigma_{R}^{-1}\otimes I_{B}\otimes
(\Gamma_{R^{\prime}B^{\prime}}^{\mathcal{N}^{\theta}})^{T}\right)  \left(
Z_{R}\otimes I_{B}\otimes\overline{Z}_{R^{\prime}}\otimes I_{B^{\prime}
}\right)  ^{\dag}\\
&  =\left(  Z_{R}\otimes I_{B}\otimes\overline{Z}_{R^{\prime}}\otimes
I_{B^{\prime}}\right)  \left(  \Gamma_{RB}^{\mathcal{N}^{\theta}}\otimes
\sigma_{R}^{-T}\otimes I_{B^{\prime}}+\sigma_{R}^{-1}\otimes I_{B}
\otimes(\Gamma_{R^{\prime}B^{\prime}}^{\mathcal{N}^{\theta}})^{T}\right)
\nonumber\\
&  \quad\times\left(  Z_{R}\otimes I_{B}\otimes\overline{Z}_{R^{\prime}
}\otimes I_{B^{\prime}}\right)  ^{\dag}.
\end{align}
Since the first matrix in
\eqref{eq:bilinear-SLD-Fish-ch-proof-1}--\eqref{eq:bilinear-SLD-Fish-ch-proof-2}
above is positive semi-definite if and only if the last one is, the
semi-definite program in \eqref{eq:SDP-SLD-Fisher-fixed-state} becomes as
follows:
\begin{equation}
\inf\left\{  \mu:
\begin{bmatrix}
\mu & \langle\Gamma|_{RR^{\prime}BB^{\prime}}\left(  (\partial_{\theta}
\Gamma_{RB}^{\mathcal{N}^{\theta}})\otimes I_{R^{\prime}B^{\prime}}\right) \\
\left(  (\partial_{\theta}\Gamma_{RB}^{\mathcal{N}^{\theta}})\otimes
I_{R^{\prime}B^{\prime}}\right)  |\Gamma\rangle_{RR^{\prime}BB^{\prime}} &
\Gamma_{RB}^{\mathcal{N}^{\theta}}\otimes\sigma_{R^{\prime}}^{-T}\otimes
I_{B^{\prime}}+\sigma_{R}^{-1}\otimes I_{B}\otimes(\Gamma_{R^{\prime}
B^{\prime}}^{\mathcal{N}^{\theta}})^{T}
\end{bmatrix}
\geq0\right\}  . \label{eq:SDP-first-primal-SLD-Fish-ch}
\end{equation}
By invoking Lemma~\ref{lem:freq-used-SDP-primal-dual}, the dual of this
program is given by
\begin{multline}
\sup_{\lambda,|\varphi\rangle_{RBR^{\prime}B^{\prime}},W_{RBR^{\prime
}B^{\prime}}}2\operatorname{Re}[\langle\varphi|_{RBR^{\prime}B^{\prime}
}(\partial_{\theta}\Gamma_{RB}^{\mathcal{N}^{\theta}})|\Gamma\rangle
_{RR^{\prime}BB^{\prime}}]\label{eq:SDP-first-dual-SLD-Fish-ch}\\
-\operatorname{Tr}[(\Gamma_{RB}^{\mathcal{N}^{\theta}}\otimes\sigma
_{R^{\prime}}^{-T}\otimes I_{B^{\prime}}+\sigma_{R}^{-1}\otimes I_{B}
\otimes(\Gamma_{R^{\prime}B^{\prime}}^{\mathcal{N}^{\theta}})^{T}
)W_{RBR^{\prime}B^{\prime}}]
\end{multline}
subject to
\begin{equation}
\lambda\leq1,\qquad
\begin{bmatrix}
\lambda & \langle\varphi|_{RBR^{\prime}B^{\prime}}\\
|\varphi\rangle_{RBR^{\prime}B^{\prime}} & W_{RBR^{\prime}B^{\prime}}
\end{bmatrix}
\geq0. \label{eq:SDP-first-dual-SLD-Fish-ch-constraints}
\end{equation}
Strong duality holds, so that \eqref{eq:SDP-first-dual-SLD-Fish-ch} is equal
to \eqref{eq:SDP-first-primal-SLD-Fish-ch}, because we are free to choose
values $\lambda$, $|\varphi\rangle_{RBR^{\prime}B^{\prime}}$, and
$W_{RBR^{\prime}B^{\prime}}$ such that the constraints in
\eqref{eq:SDP-first-dual-SLD-Fish-ch-constraints}\ are strict. Employing the
unitary swap operators $F_{RR^{\prime}}$ and $F_{BB^{\prime}}$, we can rewrite
the second term in the objective function as follows:
\begin{align}
&  \operatorname{Tr}[(\Gamma_{RB}^{\mathcal{N}^{\theta}}\otimes\sigma
_{R^{\prime}}^{-T}\otimes I_{B^{\prime}}+\sigma_{R}^{-1}\otimes I_{B}
\otimes(\Gamma_{R^{\prime}B^{\prime}}^{\mathcal{N}^{\theta}})^{T}
)W_{RBR^{\prime}B^{\prime}}]\nonumber\\
&  =\operatorname{Tr}[(\Gamma_{RB}^{\mathcal{N}^{\theta}}\otimes
\sigma_{R^{\prime}}^{-T}\otimes I_{B^{\prime}})W_{RBR^{\prime}B^{\prime}
}]+\operatorname{Tr}[(\sigma_{R}^{-1}\otimes I_{B}\otimes(\Gamma_{R^{\prime
}B^{\prime}}^{\mathcal{N}^{\theta}})^{T})W_{RBR^{\prime}B^{\prime}}]\\
&  =\operatorname{Tr}[(\left(  F_{RR^{\prime}}\otimes F_{BB^{\prime}}\right)
(\sigma_{R}^{-T}\otimes I_{B}\otimes\Gamma_{R^{\prime}B^{\prime}}
^{\mathcal{N}^{\theta}})\left(  F_{RR^{\prime}}\otimes F_{BB^{\prime}}\right)
W_{RBR^{\prime}B^{\prime}}]\nonumber\\
&  \qquad+\operatorname{Tr}[\sigma_{R}^{-1}\operatorname{Tr}_{BR^{\prime
}B^{\prime}}[(\Gamma_{R^{\prime}B^{\prime}}^{\mathcal{N}^{\theta}}
)^{T}W_{RBR^{\prime}B^{\prime}}]]\\
&  =\operatorname{Tr}[((\sigma_{R}^{-T}\otimes I_{B}\otimes\Gamma_{R^{\prime
}B^{\prime}}^{\mathcal{N}^{\theta}})\left(  F_{RR^{\prime}}\otimes
F_{BB^{\prime}}\right)  W_{RBR^{\prime}B^{\prime}}\left(  F_{RR^{\prime}
}\otimes F_{BB^{\prime}}\right)  ]\nonumber\\
&  \qquad+\operatorname{Tr}[\sigma_{R}^{-1}\operatorname{Tr}_{BR^{\prime
}B^{\prime}}[(\Gamma_{R^{\prime}B^{\prime}}^{\mathcal{N}^{\theta}}
)^{T}W_{RBR^{\prime}B^{\prime}}]]\\
&  =\operatorname{Tr}[\sigma_{R}^{-T}\operatorname{Tr}_{BR^{\prime}B^{\prime}
}[\Gamma_{R^{\prime}B^{\prime}}^{\mathcal{N}^{\theta}}\left(  F_{RR^{\prime}
}\otimes F_{BB^{\prime}}\right)  W_{RBR^{\prime}B^{\prime}}\left(
F_{RR^{\prime}}\otimes F_{BB^{\prime}}\right)  ]]\nonumber\\
&  \qquad+\operatorname{Tr}[\sigma_{R}^{-1}\operatorname{Tr}_{BR^{\prime
}B^{\prime}}[(\Gamma_{R^{\prime}B^{\prime}}^{\mathcal{N}^{\theta}}
)^{T}W_{RBR^{\prime}B^{\prime}}]]\\
&  =\operatorname{Tr}[\sigma_{R}^{-1}(\operatorname{Tr}_{BR^{\prime}B^{\prime
}}[\Gamma_{R^{\prime}B^{\prime}}^{\mathcal{N}^{\theta}}\left(  F_{RR^{\prime}
}\otimes F_{BB^{\prime}}\right)  W_{RBR^{\prime}B^{\prime}}\left(
F_{RR^{\prime}}\otimes F_{BB^{\prime}}\right)  ])^{T}]\nonumber\\
&  \qquad+\operatorname{Tr}[\sigma_{R}^{-1}\operatorname{Tr}_{BR^{\prime
}B^{\prime}}[(\Gamma_{R^{\prime}B^{\prime}}^{\mathcal{N}^{\theta}}
)^{T}W_{RBR^{\prime}B^{\prime}}]]\\
&  =\operatorname{Tr}[\sigma_{R}^{-1}K_{R}],
\end{align}
where
\begin{multline}
K_{R}=(\operatorname{Tr}_{BR^{\prime}B^{\prime}}[\Gamma_{R^{\prime}B^{\prime}
}^{\mathcal{N}^{\theta}}\left(  F_{RR^{\prime}}\otimes F_{BB^{\prime}}\right)
W_{RBR^{\prime}B^{\prime}}\left(  F_{RR^{\prime}}\otimes F_{BB^{\prime}
}\right)  ])^{T}\\
+\operatorname{Tr}_{BR^{\prime}B^{\prime}}[(\Gamma_{R^{\prime}B^{\prime}
}^{\mathcal{N}^{\theta}})^{T}W_{RBR^{\prime}B^{\prime}}].
\end{multline}
So the SDP\ in \eqref{eq:SDP-first-dual-SLD-Fish-ch}\ can be written as
\begin{equation}
\sup_{\lambda,|\varphi\rangle_{RBR^{\prime}B^{\prime}},W_{RBR^{\prime
}B^{\prime}}}2\operatorname{Re}[\langle\varphi|_{RBR^{\prime}B^{\prime}
}(\partial_{\theta}\Gamma_{RB}^{\mathcal{N}^{\theta}})|\Gamma\rangle
_{RR^{\prime}BB^{\prime}}]-\operatorname{Tr}[\sigma_{R}^{-1}K_{R}]
\label{eq:SDP-first-dual-SLD-Fish-ch-2}
\end{equation}
subject to
\begin{equation}
\lambda\leq1,\qquad
\begin{bmatrix}
\lambda & \langle\varphi|_{RBR^{\prime}B^{\prime}}\\
|\varphi\rangle_{RBR^{\prime}B^{\prime}} & W_{RBR^{\prime}B^{\prime}}
\end{bmatrix}
\geq0.
\end{equation}
Now noting from Lemma~\ref{lem:min-XYinvX} that
\begin{equation}
\sigma_{R}^{-1}=\inf\left\{  Y_{R}:
\begin{bmatrix}
\sigma_{R} & I_{R}\\
I_{R} & Y_{R}
\end{bmatrix}
\geq0\right\}  ,
\end{equation}
and that $\sigma_{R}^{-1}$ and $K_{R}$ are positive semi-definite, we can
rewrite the SDP in \eqref{eq:SDP-first-dual-SLD-Fish-ch-2}\ as
\begin{multline}
\sup_{\lambda,|\varphi\rangle_{RBR^{\prime}B^{\prime}},W_{RBR^{\prime
}B^{\prime}}}\left(  2\operatorname{Re}[\langle\varphi|_{RBR^{\prime}
B^{\prime}}(\partial_{\theta}\Gamma_{RB}^{\mathcal{N}^{\theta}})|\Gamma
\rangle_{RR^{\prime}BB^{\prime}}]-\inf_{Y_{R}}\operatorname{Tr}[Y_{R}
K_{R}]\right) \label{eq:almost-there-SLD-Fisher-channels-SDP}\\
=\sup_{\lambda,|\varphi\rangle_{RBR^{\prime}B^{\prime}},W_{RBR^{\prime
}B^{\prime}},Y_{R}}\left(  2\operatorname{Re}[\langle\varphi|_{RBR^{\prime
}B^{\prime}}(\partial_{\theta}\Gamma_{RB}^{\mathcal{N}^{\theta}}
)|\Gamma\rangle_{RR^{\prime}BB^{\prime}}]-\operatorname{Tr}[Y_{R}
K_{R}]\right)
\end{multline}
subject to
\begin{equation}
\lambda\leq1,\quad
\begin{bmatrix}
\lambda & \langle\varphi|_{RBR^{\prime}B^{\prime}}\\
|\varphi\rangle_{RBR^{\prime}B^{\prime}} & W_{RBR^{\prime}B^{\prime}}
\end{bmatrix}
\geq0,\quad
\begin{bmatrix}
\sigma_{R} & I_{R}\\
I_{R} & Y_{R}
\end{bmatrix}
\geq0.
\end{equation}
Then we can finally include the maximization over input states $\sigma_{R}$
(satisfying $\sigma_{R}\geq0$ and $\operatorname{Tr}[\sigma_{R}]=1$) to arrive
at the form given in \eqref{eq:SDP-SLD-ch-Fish}.
\end{proof}

\pagebreak

\chapter{Limits on Multiparameter Estimation of Quantum Channels}\label{ch:multi}
\graphicspath{{figures/}}

\vspace{0.5em}

In this chapter, we present our results for the task of simultaneously estimating multiple parameters encoded in a quantum channel. We generalize some of the results of Chapter~\ref{ch:single}, but we also develop new tools required to establish the more non-trivial Cramer--Rao bounds for the case of multiparameter estimation. For the multiparameter case, we will focus exclusively on the RLD Fisher information value of quantum states and channels and the Cramer--Rao bounds we can establish using them. 

First, we show how to convert the canonical matrix inequality Cramer--Rao bound for state estimation
\begin{equation}
	\text{Cov}(\bm{\theta}) \geq \widehat{I}_F(\bm{\theta}; \{\rho_{\bm{\theta}}\}_{\bm{\theta}})^{-1}
\end{equation}
into a scalar bound for ease of use and experimental application. To do so, we use the RLD Fisher information value of quantum states which we defined in Chapter~\ref{ch:prelims}.

Next, we continue in the vein of Chapter~\ref{ch:single} and define the amortized RLD Fisher information value of quantum channels. We prove a chain rule inequality obeyed by this quantity for all quantum channels, which, as we saw in Chapter~\ref{ch:single}, leads to an amortization collapse for it. 

Again, analogous to Chapter~\ref{ch:single}, we prove a meta-converse theorem which connects the amortized RLD Fisher information value to the RLD Fisher information value of a sequential estimation protocol. The amortization collapse can then directly be used to establish a single-letter Cramer--Rao bound, which we do.

Just as in the case of our RLD-based Cramer--Rao bound for single parameter estimation of quantum channels, our bound for multiparameter estimation is
\begin{itemize}
	\item single-letter, i.e. computing it requires computing the RLD Fisher information value of a single channel use even though the bound is applicable for $n$-round sequential procotols,
	\item universally applicable, in the sense that our bound applies to all quantum channels, and thus encompasses all admissible quantum dynamics, and
	\item efficiently computable via a semi-definite program. 
\end{itemize}

After providing our bound, we comment on how it offers a no-go condition for Heisenberg scaling of the estimation error. That is, if the finiteness condition for the RLD Fisher information value is met, then Heisenberg scaling of the error with respect to $n$, the number of channel uses, is unattainable.

We then evaluate our RLD-based Cramer--Rao bound for simultaneously estimating the two parameters of a generalized amplitude damping channel, generalizing the example from Chapter~\ref{ch:single}. Lastly, we provide semi-definite programs to evaluate the RLD Fisher information value of quantum states and quantum channels.

\section{Limits on multiparameter channel estimation using RLD Fisher information value of quantum states}

Before we can establish our lower bounds for estimating multiple parameters in a quantum channel, we show how we establish scalar Cramer--Rao bounds for multiparameter state estimation. We introduced the RLD Fisher information value of quantum states in~\eqref{eq:states-rld-fisher-value} in Chapter~\ref{ch:prelims}, and we will now use that to prove a scalar Cramer--Rao bound for state estimation in the multiparameter setting.

\begin{theorem} \label{thm:state-rld-scalar-inequality}
    The following scalar Cramer--Rao bound holds for estimating multiple parameters $\bm{\theta}$ encoded in a family of quantum states $\{ \rho_{\bm{\theta}} \}_{\bm{\theta}}$:
    \begin{equation} \label{eq:state-rld-scalar-inequality}
    \Tr[ W \operatorname{Cov}(\bm{\theta}) ] \geq \frac{1}{   \widehat{I}_F (\bm{\theta},W;\{\rho_{\bm{\theta}}\}_{\bm{\theta}}) },
    \end{equation}
    where the weight matrix $W$ satisfies $\Tr[W]=1$.
\end{theorem}

\begin{proof}
We recall the following matrix inequality \cite{Helstrom1976, Helstrom1973, Sidhu2019}, which we also stated in Chapter~\ref{ch:prelims} as \eqref{eq:matrix-CRB-1}:
\begin{equation} \label{eq:app-matrix-crb}
    \text{Cov}(\bm{\theta}) \geq \widehat{I}_F(\bm{\theta};\{\rho_{\bm{\theta}}\}_{\bm{\theta}})^{-1}.
\end{equation}
Let us take $W$ as a non-zero, positive semi-definite matrix, and then define the normalized operator~$W'$ as
\begin{equation}
    W' \coloneqq \frac{W}{\Tr[W]}.
\end{equation}
The matrix inequality in \eqref{eq:app-matrix-crb} implies the following:
\begin{align}
    \Tr[W] \Tr[W' \text{Cov}(\bm{\theta}) ] &\geq \Tr[W] \Tr[W' \widehat{I}_F(\bm{\theta};\{\rho_{\bm{\theta}}\}_{\bm{\theta}})^{-1}] \\
                                      &= \Tr[W] \Tr[W'^{1/2} \widehat{I}_F(\bm{\theta};\{\rho_{\bm{\theta}}\}_{\bm{\theta}})^{-1} W'^{1/2}] \\
& = \Tr[W] \sum_k \langle k | W'^{1/2} \widehat{I}_F(\bm{\theta};\{\rho_{\bm{\theta}}\}_{\bm{\theta}})^{-1} W'^{1/2} | k \rangle
\\                                    &\geq 
\Tr[W]  \left[\sum_k \langle k | W'^{1/2} \widehat{I}_F(\bm{\theta};\{\rho_{\bm{\theta}}\}_{\bm{\theta}}) W'^{1/2} | k \rangle\right]^{-1}
\\& = \frac{\Tr[W]}{ \Tr[ W'^{1/2} \widehat{I}_F(\bm{\theta};\{\rho_{\bm{\theta}}\}_{\bm{\theta}}) W'^{1/2} ]  } \\
                                      &= \frac{\Tr[W]}{ \Tr[W' \widehat{I}_F(\bm{\theta};\{\rho_{\bm{\theta}}\}_{\bm{\theta}})] } \\
                                      &= \frac{\Tr[W]^2}{ \Tr[W \widehat{I}_F(\bm{\theta};\{\rho_{\bm{\theta}}\}_{\bm{\theta}})] }.
\end{align}
The first inequality is a consequence of \eqref{eq:app-matrix-crb}. The first equality follows from cyclicity of trace. The second inequality uses the operator Jensen inequality \cite{Hansen2003} for the operator convex function $f(x) = x^{-1}$. This inequality is saturated if and only if $\widehat{I}_F(\bm{\theta};\{\rho_{\bm{\theta}}\}_{\bm{\theta}})$ is diagonal in the eigenbasis of $W$, i.e., $\left[\widehat{I}_F(\bm{\theta};\{\rho_{\bm{\theta}}\}_{\bm{\theta}}), W  \right] = 0$. The next equality comes again from cyclicity of trace, and the last equality comes from the definition of $W'$.

The reasoning above leads us to
\begin{equation}
    \Tr[W \text{Cov} (\bm{\theta})] \geq \frac{(\Tr[W])^2}{\Tr[W \widehat{I}_F(\bm{\theta};\{\rho_{\bm{\theta}}\}_{\bm{\theta}})]}
\end{equation}
for every positive semi-definite matrix $W$, 
which implies that
\begin{equation} 
    \Tr[W' \text{Cov} (\bm{\theta})] \geq \frac{1}{\Tr[W' \widehat{I}_F(\bm{\theta};\{\rho_{\bm{\theta}}\}_{\bm{\theta}})]}
\end{equation}
for every positive semi-definite matrix $W'$ such that $\Tr[W'] = 1$.
\end{proof}

\section{Amortized RLD Fisher information value}

The next step towards establishing Cramer--Rao bounds for multiparameter channel estimation is to introduce and define the amortized RLD Fisher information value like we did in Chapter~\ref{ch:single}. We use the RLD Fisher information value of quantum channels, defined in~\eqref{eq:rld-value-def}, to define the amortized RLD Fisher information value of the channel family $\{ \mathcal{N}_{A \rightarrow B}^{\bm{\theta}} \}_{\bm{\theta}}$.

\begin{definition}
[Amortized RLD Fisher information value of quantum channels]The amortized RLD Fisher information value of a family $\{ \mathcal{N}_{A \rightarrow B}^{\bm{\theta}} \}_{\bm{\theta}}$ of quantum channels is defined as follows:
\begin{equation}
	\widehat{I}_{F}^{\mathcal{A}}(\bm{\bm{\theta}}, W; \{\mathcal{N}_{A\rightarrow B}^{\bm{\bm{\theta}}}\}_{\bm{\bm{\theta}}})\coloneqq  \\ \sup_{\{\rho_{RA}^{\bm{\bm{\theta}}}\}_{\bm{\bm{\theta}}}}  \widehat{I}_{F}(\bm{\bm{\theta}}, W; \{\mathcal{N}_{A\rightarrow B}^{\bm{\bm{\theta}}}(\rho_{RA}^{\bm{\bm{\theta}}})\}_{\bm{\bm{\theta}}})-\widehat{I}_{F}(\bm{\bm{\theta}}, W; \{\rho_{RA}^{\bm{\bm{\theta}}}\}_{\bm{\bm{\theta}}}).
\end{equation} \label{eq:amortized-rld-value-def}
\end{definition}

Just like in the case of single parameter estimation, the concept of amortization allows for a resource state $\rho_{RA}^{\bm{\theta}}$ at the channel input to help with the estimation task while also subtracting off its value to account for the RLD Fisher information value strictly present in the channel family. We also have the fact that amortization, or catalysis with an input state family, can never decrease the RLD Fisher information value of quantum channels.

\begin{proposition}
Let $\{\mathcal{N}_{A\rightarrow B}^{\bm{\theta}}\}_{\bm{\theta}}$ be a family of quantum channels. The RLD Fisher information value does not exceed the amortized one:
\begin{equation}
	\widehat{I}_{F}^{\mathcal{A}}(\bm{\bm{\theta}}, W; \{\mathcal{N}_{A\rightarrow B}^{\bm{\bm{\theta}}}\}_{\bm{\bm{\theta}}}) \geq \widehat{I}_{F}(\bm{\bm{\theta}}, W; \{\mathcal{N}_{A\rightarrow B}^{\bm{\bm{\theta}}}\}_{\bm{\bm{\theta}}}) \label{eq:multiparam-rld-amort-ineq-obvi-dir}
\end{equation}

\end{proposition}

\begin{proof}
This can be understood by considering the right-hand side to arise from restricting to input states with no parameter dependence. That is,
\begin{align}
\widehat{I}_F^{\mathcal{A}} (\bm{\theta}, W; \{ \mathcal{N}^{\bm{\theta}}_{A \rightarrow B} \}_{\bm{\theta}}) &\coloneqq \sup_{\{\rho_{RA}^{\bm{\theta}}\}_{\bm{\theta}}}\left[  \widehat{I}_{F}(\bm{\theta}, W;\{\mathcal{N}_{A\rightarrow B}^{\bm{\theta}}(\rho_{RA}^{\bm{\theta}})\}_{\bm{\theta}})-\widehat{I}_{F}(\bm{\theta}, W;\{\rho_{RA}^{\bm{\theta}}\}_{\bm{\theta}})\right] \\
&\geq \sup_{\{\rho_{RA}\}_{\bm{\theta}}}\left[  \widehat{I}_{F}(\bm{\theta}, W;\{\mathcal{N}_{A\rightarrow B}^{\bm{\theta}}(\rho_{RA})\}_{\bm{\theta}})-\widehat{I}_{F}(\bm{\theta}, W;\{\rho_{RA}\})\right] \\
&=  \sup_{\{\rho_{RA}\}_{\bm{\theta}}}\left[  \widehat{I}_{F}(\bm{\theta}, W;\{\mathcal{N}_{A\rightarrow B}^{\bm{\theta}}(\rho_{RA})\}_{\bm{\theta}})\right] \\
&= \widehat{I}_{F}(\bm{\theta}, W;\{\mathcal{N}_{A\rightarrow B}^{\bm{\theta}}\}_{\bm{\theta}}).
\end{align}
Since the inequality holds for all input states $\rho_{RA}$, we conclude \eqref{eq:multiparam-rld-amort-ineq-obvi-dir}.
\end{proof}

Qualitatively, it means that the marginal increase in Fisher information value can never be decreased by using a catalyst state family as input.

\subsection{Amortization collapse for RLD Fisher information value}

We now proceed to show how the direct inequality in \eqref{eq:multiparam-rld-amort-ineq-obvi-dir} can also be reversed, which would imply an amortization collapse for the RLD Fisher information value of quantum channels. Just like we did for the root SLD Fisher information and the RLD Fisher information of quantum channels in Chapter~\ref{ch:prelims}, here too we begin by stating a chain rule property obeyed by the RLD Fisher information value. The chain rule that we established for the RLD Fisher information in Chapter~\ref{ch:single} in \eqref{eq:single-RLD-chain-rule} is generalized to the case of multiparameter estimation. To do so, we first state and prove it in the form of an operator inequality:
\begin{proposition} \label{prop:chain-rule-operator-ineq} 
	Let $\{\mathcal{N}_{A\rightarrow B}^{\bm{\theta}}\}_{\bm{\theta}}$ be a
	differentiable family of quantum channels, and let $\{\rho_{RA}^{\bm{\theta}}%
	\}_{\bm{\theta}}$ be a differentiable family of quantum states. Then the
	following chain-rule operator inequality holds%
	\begin{equation} \label{app-eq:operator-chain-rule}
	\widehat{I}_F(\bm{\theta};\{\mathcal{N}_{A\rightarrow B}^{\bm{\theta}}(\rho
	_{RA}^{\bm{\theta}})\}_{\bm{\theta}})\leq\sum_{j,k=1}^{D}|j\rangle\!\langle
	k|\operatorname{Tr}[(\rho_{S}^{\bm{\theta}})^{T}\operatorname{Tr}%
	_{B}[(\partial_{\theta_{j}}\Gamma_{SB}^{\mathcal{N}^{\bm{\theta}}}%
	)(\Gamma_{SB}^{\mathcal{N}^{\bm{\theta}}})^{-1}(\partial_{\theta_{k}}%
	\Gamma_{SB}^{\mathcal{N}^{\bm{\theta}}})]
	+\widehat{I}_F(\bm{\theta};\{\rho_{RA}^{\bm{\theta}}\}_{\bm{\theta}}),
	\end{equation}
	where $\rho_S^{\bm{\theta}}$ is equal to the reduced state of $\rho_{RA}^{\bm{\theta}}$ on system $A$ and system $S$ is isomorphic to system $A$.
\end{proposition}

\begin{proof}
The proof of the above is similar to the proof of the single-parameter chain rule proved in Proposition~\ref{prop:chain-rule-RLD} in Chapter~\ref{ch:single} for the RLD Fisher information. First, we make use of the following identity:
\begin{equation}
\mathcal{N}_{A\rightarrow B}^{\bm{\theta}}(\rho_{RA}^{\bm{\theta}}) = \langle\Gamma|_{AS}\rho_{RA}^{\bm{\theta}}\otimes\Gamma_{SB}^{\mathcal{N}%
		^{\bm{\theta}}}|\Gamma\rangle_{AS},
\end{equation}
where
\begin{equation}
|\Gamma\rangle_{AS} \coloneqq \sum_{i} \ket{i}_A \ket{i}_S.
\end{equation}

Consider that%
	\begin{equation}
	\widehat{I}_F(\bm{\theta};\{\mathcal{N}_{A\rightarrow B}^{\bm{\theta}}(\rho
	_{RA}^{\bm{\theta}})\}_{\bm{\theta}})=\operatorname{Tr}_{2}\!\left[
	\sum_{j,k=1}^{D}|j\rangle\!\langle k|\otimes(\partial_{\theta_{j}}%
	\mathcal{N}_{A\rightarrow B}^{\bm{\theta}}(\rho_{RA}^{\bm{\theta}}%
	))(\mathcal{N}_{A\rightarrow B}^{\bm{\theta}}(\rho_{RA}^{\bm{\theta}}%
	))^{-1}(\partial_{\theta_{k}}\mathcal{N}_{A\rightarrow B}^{\bm{\theta}}%
	(\rho_{RA}^{\bm{\theta}}))\right]  .
	\end{equation}
	
We can, with a series of manipulations, show that
\begin{align}
	&  \sum_{j,k=1}^{D}|j\rangle\!\langle k|\otimes(\partial_{\theta_{j}}%
	\mathcal{N}_{A\rightarrow B}^{\bm{\theta}}(\rho_{RA}^{\bm{\theta}}%
	))(\mathcal{N}_{A\rightarrow B}^{\bm{\theta}}(\rho_{RA}^{\bm{\theta}}%
	))^{-1}(\partial_{\theta_{k}}\mathcal{N}_{A\rightarrow B}^{\bm{\theta}}%
	(\rho_{RA}^{\bm{\theta}}))\nonumber\\
	&  =\sum_{j,k=1}^{D}|j\rangle\!\langle k|\otimes(\partial_{\theta_{j}}%
	\langle\Gamma|_{AS}\rho_{RA}^{\bm{\theta}}\otimes\Gamma_{SB}^{\mathcal{N}%
		^{\bm{\theta}}}|\Gamma\rangle_{AS})(\langle\Gamma|_{AS}\rho_{RA}%
	^{\bm{\theta}}\otimes\Gamma_{SB}^{\mathcal{N}^{\bm{\theta}}}|\Gamma
	\rangle_{AS})^{-1}(\partial_{\theta_{k}}\langle\Gamma|_{AS}\rho_{RA}%
	^{\bm{\theta}}\otimes\Gamma_{SB}^{\mathcal{N}^{\bm{\theta}}}|\Gamma
	\rangle_{AS})\\
	&  =\langle\Gamma|_{AS}\left(  \sum_{j=1}^{D}|j\rangle\!\langle0|\otimes
	\partial_{\theta_{j}}(\rho_{RA}^{\bm{\theta}}\otimes\Gamma_{SB}^{\mathcal{N}%
		^{\bm{\theta}}})\right)  |\Gamma\rangle_{AS}(\langle\Gamma|_{AS}%
	(|0\rangle\!\langle0|\otimes\rho_{RA}^{\bm{\theta}}\otimes\Gamma_{SB}%
	^{\mathcal{N}^{\bm{\theta}}})|\Gamma\rangle_{AS})^{-1}\nonumber\\
	&  \qquad\times\langle\Gamma|_{AS}\left(  \sum_{k=1}^{D}|0\rangle\!\langle
	k|\otimes\partial_{\theta_{k}}(\rho_{RA}^{\bm{\theta}}\otimes\Gamma
	_{SB}^{\mathcal{N}^{\bm{\theta}}})\right)  |\Gamma\rangle_{AS}.
	\end{align}
	
	Then identifying
	\begin{align}
	X &  =\sum_{k=1}^{D}|0\rangle\!\langle k|\otimes\partial_{\theta_{k}}(\rho
	_{RA}^{\bm{\theta}}\otimes\Gamma_{SB}^{\mathcal{N}^{\bm{\theta}}}),\\
	Y &  =|0\rangle\!\langle0|\otimes\rho_{RA}^{\bm{\theta}}\otimes\Gamma
	_{SB}^{\mathcal{N}^{\bm{\theta}}}\\
	L &  =I_{D}\otimes\langle\Gamma|_{AS}\otimes I_{RB},
	\end{align}
	we can apply the well known transformer inequality $LX^{\dag}L^{\dag}(LYL^{\dag})^{-1}LXL^{\dag}\leq LX^{\dag}Y^{-1}XL^{\dag}$ (Lemma~\ref{lem:transformer-ineq-basic}) to find that the last line above is
	not larger than the following one in the operator-inequality sense:%
	\begin{align}
	&  \langle\Gamma|_{AS}\left(  \sum_{j=1}^{D}|j\rangle\!\langle0|\otimes
	\partial_{\theta_{j}}(\rho_{RA}^{\bm{\theta}}\otimes\Gamma_{SB}^{\mathcal{N}%
		^{\bm{\theta}}})\right)  (|0\rangle\!\langle0|\otimes\rho_{RA}^{\bm{\theta}}%
	\otimes\Gamma_{SB}^{\mathcal{N}^{\bm{\theta}}})^{-1}\left(  \sum_{k=1}%
	^{D}|0\rangle\!\langle k|\otimes\partial_{\theta_{k}}(\rho_{RA}^{\bm{\theta}}%
	\otimes\Gamma_{SB}^{\mathcal{N}^{\bm{\theta}}})\right)  |\Gamma\rangle
	_{AS}\nonumber\\
	&  =\langle\Gamma|_{AS}\sum_{j,k=1}^{D}|j\rangle\!\langle k|\otimes
	\partial_{\theta_{j}}(\rho_{RA}^{\bm{\theta}}\otimes\Gamma_{SB}^{\mathcal{N}%
		^{\bm{\theta}}})(\rho_{RA}^{\bm{\theta}}\otimes\Gamma_{SB}^{\mathcal{N}%
		^{\bm{\theta}}})^{-1}\partial_{\theta_{k}}(\rho_{RA}^{\bm{\theta}}%
	\otimes\Gamma_{SB}^{\mathcal{N}^{\bm{\theta}}})|\Gamma\rangle_{AS}\\
	&  =\sum_{j,k=1}^{D}|j\rangle\!\langle k|\otimes\langle\Gamma|_{AS}%
	\partial_{\theta_{j}}(\rho_{RA}^{\bm{\theta}}\otimes\Gamma_{SB}^{\mathcal{N}%
		^{\bm{\theta}}})(\rho_{RA}^{\bm{\theta}}\otimes\Gamma_{SB}^{\mathcal{N}%
		^{\bm{\theta}}})^{-1}\partial_{\theta_{k}}(\rho_{RA}^{\bm{\theta}}%
	\otimes\Gamma_{SB}^{\mathcal{N}^{\bm{\theta}}})|\Gamma\rangle_{AS}.
	\end{align}
	
	Consider that%
	\begin{equation}
	\partial_{\theta_{j}}(\rho_{RA}^{\bm{\theta}}\otimes\Gamma_{SB}^{\mathcal{N}%
		^{\bm{\theta}}})=(\partial_{\theta_{j}}\rho_{RA}^{\bm{\theta}})\otimes
	\Gamma_{SB}^{\mathcal{N}^{\bm{\theta}}}+\rho_{RA}^{\bm{\theta}}\otimes
	(\partial_{\theta_{j}}\Gamma_{SB}^{\mathcal{N}^{\bm{\theta}}}).
	\end{equation}
	Then we find that%
	\begin{align}
	&  \partial_{\theta_{j}}(\rho_{RA}^{\bm{\theta}}\otimes\Gamma_{SB}%
	^{\mathcal{N}^{\bm{\theta}}})(\rho_{RA}^{\bm{\theta}}\otimes\Gamma
	_{SB}^{\mathcal{N}^{\bm{\theta}}})^{-1}\nonumber\\
	&  =((\partial_{\theta_{j}}\rho_{RA}^{\bm{\theta}})\otimes\Gamma
	_{SB}^{\mathcal{N}^{\bm{\theta}}}+\rho_{RA}^{\bm{\theta}}\otimes
	(\partial_{\theta_{j}}\Gamma_{SB}^{\mathcal{N}^{\bm{\theta}}}))(\rho
	_{RA}^{\bm{\theta}}\otimes\Gamma_{SB}^{\mathcal{N}^{\bm{\theta }}})^{-1}\\
	&  =((\partial_{\theta_{j}}\rho_{RA}^{\bm{\theta}})\otimes\Gamma
	_{SB}^{\mathcal{N}^{\bm{\theta}}}+\rho_{RA}^{\bm{\theta}}\otimes
	(\partial_{\theta_{j}}\Gamma_{SB}^{\mathcal{N}^{\bm{\theta}}}))((\rho
	_{RA}^{\bm{\theta}})^{-1}\otimes(\Gamma_{SB}^{\mathcal{N}^{\bm{\theta}}}%
	)^{-1})\\
	&  =(\partial_{\theta_{j}}\rho_{RA}^{\bm{\theta}})(\rho_{RA}^{\bm{\theta}}%
	)^{-1}\otimes\Gamma_{SB}^{\mathcal{N}^{\bm{\theta}}}(\Gamma_{SB}%
	^{\mathcal{N}^{\bm{\theta}}})^{-1}+\rho_{RA}^{\bm{\theta }}(\rho
	_{RA}^{\bm{\theta}})^{-1}\otimes(\partial_{\theta_{j}}\Gamma_{SB}%
	^{\mathcal{N}^{\bm{\theta}}})(\Gamma_{SB}^{\mathcal{N}^{\bm{\theta}}})^{-1}\\
	&  =(\partial_{\theta_{j}}\rho_{RA}^{\bm{\theta}})(\rho_{RA}^{\bm{\theta}}%
	)^{-1}\otimes\Pi_{\Gamma^{\mathcal{N}^{\bm{\theta}}}}+\Pi_{\rho^{\bm{\theta}}%
	}\otimes(\partial_{\theta_{j}}\Gamma_{SB}^{\mathcal{N}^{\bm{\theta}}}%
	)(\Gamma_{SB}^{\mathcal{N}^{\bm{\theta}}})^{-1}.
	\end{align}
	
	Right multiplying this last line by $\partial_{\theta_{k}}(\rho_{RA}%
	^{\bm{\theta}}\otimes\Gamma_{SB}^{\mathcal{N}^{\bm{\theta}}})$ gives%
	\begin{multline}
	  ((\partial_{\theta_{j}}\rho_{RA}^{\bm{\theta}})(\rho_{RA}^{\bm{\theta}}%
	)^{-1}\otimes\Pi_{\Gamma^{\mathcal{N}^{\bm{\theta}}}}+\Pi_{\rho^{\bm{\theta}}%
	}\otimes(\partial_{\theta_{j}}\Gamma_{SB}^{\mathcal{N}^{\bm{\theta}}}%
	)(\Gamma_{SB}^{\mathcal{N}^{\bm{\theta}}})^{-1})\partial_{\theta_{k}}%
	(\rho_{RA}^{\bm{\theta}}\otimes\Gamma_{SB}^{\mathcal{N}^{\bm{\theta}}%
	})\\
	  =(\partial_{\theta_{j}}\rho_{RA}^{\bm{\theta}})(\rho_{RA}^{\bm{\theta}}%
	)^{-1}(\partial_{\theta_{k}}\rho_{RA}^{\bm{\theta}})\otimes\Gamma
	_{SB}^{\mathcal{N}^{\bm{\theta}}}+(\partial_{\theta_{j}}\rho_{RA}%
	^{\bm{\theta}})\Pi_{\rho^{\bm{\theta}}}\otimes\Pi_{\Gamma^{\mathcal{N}%
			^{\bm{\theta}}}}(\partial_{\theta_{k}}\Gamma_{SB}^{\mathcal{N}^{\bm{\theta}}%
	})\\
	+\Pi_{\rho^{\bm{\theta}}}(\partial_{\theta_{k}}\rho_{RA}%
	^{\bm{\theta }})\otimes(\partial_{\theta_{j}}\Gamma_{SB}^{\mathcal{N}%
		^{\bm{\theta}}})\Pi_{\Gamma^{\mathcal{N}^{\bm{\theta}}}}+\rho_{RA}%
	^{\bm{\theta}}\otimes(\partial_{\theta_{j}}\Gamma_{SB}^{\mathcal{N}%
		^{\bm{\theta}}})(\Gamma_{SB}^{\mathcal{N}^{\bm{\theta}}})^{-1}(\partial
	_{\theta_{k}}\Gamma_{SB}^{\mathcal{N}^{\bm{\theta}}}).	
	\end{multline}
	
	Adding in the terms $(\partial_{\theta_{j}}\rho_{RA}^{\bm{\theta}})\Pi
	_{\rho^{\bm{\theta}}}^{\perp}=0$, $(\partial_{\theta_{j}}\Gamma_{SB}%
	^{\mathcal{N}^{\bm{\theta}}})\Pi_{\Gamma^{\mathcal{N}^{\bm{\theta
	}}}}^{\perp}=0$, $\Pi_{\rho^{\bm{\theta}}}^{\perp}(\partial_{\theta_{k}}%
	\rho_{RA}^{\bm{\theta}})=0$, and $\Pi_{\Gamma^{\mathcal{N}^{\bm{\theta}}}%
	}^{\perp}(\partial_{\theta_{k}}\Gamma_{SB}^{\mathcal{N}^{\bm{\theta}}})=0$,
	which follow from the support conditions for finite RLD Fisher information, the last line above becomes as follows:%
	\begin{multline}
	=(\partial_{\theta_{j}}\rho_{RA}^{\bm{\theta}})(\rho_{RA}^{\bm{\theta }}%
	)^{-1}(\partial_{\theta_{k}}\rho_{RA}^{\bm{\theta}})\otimes\Gamma
	_{SB}^{\mathcal{N}^{\bm{\theta}}}+(\partial_{\theta_{j}}\rho_{RA}%
	^{\bm{\theta}})\otimes(\partial_{\theta_{k}}\Gamma_{SB}^{\mathcal{N}%
		^{\bm{\theta}}})\\
	+(\partial_{\theta_{k}}\rho_{RA}^{\bm{\theta}})\otimes(\partial_{\theta_{j}%
	}\Gamma_{SB}^{\mathcal{N}^{\bm{\theta}}})+\rho_{RA}^{\bm{\theta}}%
	\otimes(\partial_{\theta_{j}}\Gamma_{SB}^{\mathcal{N}^{\bm{\theta}}}%
	)(\Gamma_{SB}^{\mathcal{N}^{\bm{\theta}}})^{-1}(\partial_{\theta_{k}}%
	\Gamma_{SB}^{\mathcal{N}^{\bm{\theta}}}).
	\end{multline}
	
	So then the relevant matrix is simplified as follows:%
	\begin{multline}
	\sum_{j,k=1}^{D}|j\rangle\!\langle k|\otimes\langle\Gamma|_{AS}[(\partial
	_{\theta_{j}}\rho_{RA}^{\bm{\theta}})(\rho_{RA}^{\bm{\theta}})^{-1}%
	(\partial_{\theta_{k}}\rho_{RA}^{\bm{\theta}})\otimes\Gamma_{SB}%
	^{\mathcal{N}^{\bm{\theta}}}+(\partial_{\theta_{j}}\rho_{RA}^{\bm{\theta}}%
	)\otimes(\partial_{\theta_{k}}\Gamma_{SB}^{\mathcal{N}^{\bm{\theta}}})\\
	+(\partial_{\theta_{k}}\rho_{RA}^{\bm{\theta}})\otimes(\partial_{\theta_{j}%
	}\Gamma_{SB}^{\mathcal{N}^{\bm{\theta}}})+\rho_{RA}^{\bm{\theta}}%
	\otimes(\partial_{\theta_{j}}\Gamma_{SB}^{\mathcal{N}^{\bm{\theta}}}%
	)(\Gamma_{SB}^{\mathcal{N}^{\bm{\theta}}})^{-1}(\partial_{\theta_{k}}%
	\Gamma_{SB}^{\mathcal{N}^{\bm{\theta}}})]|\Gamma\rangle_{AS}.
	\end{multline}
	Thus, we have established the following operator inequality:
	\begin{multline}
	\sum_{j,k=1}^{D}|j\rangle\!\langle k|\otimes(\partial_{\theta_{j}}%
	\mathcal{N}_{A\rightarrow B}^{\bm{\theta}}(\rho_{RA}^{\bm{\theta}}%
	))(\mathcal{N}_{A\rightarrow B}^{\bm{\theta}}(\rho_{RA}^{\bm{\theta}}%
	))^{-1}(\partial_{\theta_{k}}\mathcal{N}_{A\rightarrow B}^{\bm{\theta}}%
	(\rho_{RA}^{\bm{\theta}}))\\
	\leq\sum_{j,k=1}^{D}|j\rangle\!\langle k|\otimes\langle\Gamma|_{AS}%
	[(\partial_{\theta_{j}}\rho_{RA}^{\bm{\theta}})(\rho_{RA}^{\bm{\theta}}%
	)^{-1}(\partial_{\theta_{k}}\rho_{RA}^{\bm{\theta}})\otimes\Gamma
	_{SB}^{\mathcal{N}^{\bm{\theta}}}+(\partial_{\theta_{j}}\rho_{RA}%
	^{\bm{\theta}})\otimes(\partial_{\theta_{k}}\Gamma_{SB}^{\mathcal{N}%
		^{\bm{\theta}}})\\
	+(\partial_{\theta_{k}}\rho_{RA}^{\bm{\theta}})\otimes(\partial_{\theta_{j}%
	}\Gamma_{SB}^{\mathcal{N}^{\bm{\theta}}})+\rho_{RA}^{\bm{\theta}}%
	\otimes(\partial_{\theta_{j}}\Gamma_{SB}^{\mathcal{N}^{\bm{\theta}}}%
	)(\Gamma_{SB}^{\mathcal{N}^{\bm{\theta}}})^{-1}(\partial_{\theta_{k}}%
	\Gamma_{SB}^{\mathcal{N}^{\bm{\theta}}})]|\Gamma\rangle_{AS}.
	\end{multline}
	We can then take a partial trace over the $RB\ $systems and arrive at the
	following operator inequality:%
	\begin{multline}
	\widehat{I}_F(\bm{\theta};\{\mathcal{N}_{A\rightarrow B}^{\bm{\theta}}(\rho
	_{RA}^{\bm{\theta}})\}_{\bm{\theta}})\\
	\leq\sum_{j,k=1}^{D}|j\rangle\!\langle k|~\operatorname{Tr}_{RB}[\langle
	\Gamma|_{AS}[(\partial_{\theta_{j}}\rho_{RA}^{\bm{\theta}})(\rho
	_{RA}^{\bm{\theta}})^{-1}(\partial_{\theta_{k}}\rho_{RA}^{\bm{\theta}}%
	)\otimes\Gamma_{SB}^{\mathcal{N}^{\bm{\theta}}}+(\partial_{\theta_{j}}%
	\rho_{RA}^{\bm{\theta}})\otimes(\partial_{\theta_{k}}\Gamma_{SB}%
	^{\mathcal{N}^{\bm{\theta}}})\\
	+(\partial_{\theta_{k}}\rho_{RA}^{\bm{\theta}})\otimes(\partial_{\theta_{j}%
	}\Gamma_{SB}^{\mathcal{N}^{\bm{\theta}}})+\rho_{RA}^{\bm{\theta}}%
	\otimes(\partial_{\theta_{j}}\Gamma_{SB}^{\mathcal{N}^{\bm{\theta}}}%
	)(\Gamma_{SB}^{\mathcal{N}^{\bm{\theta}}})^{-1}(\partial_{\theta_{k}}%
	\Gamma_{SB}^{\mathcal{N}^{\bm{\theta}}})]|\Gamma\rangle_{AS}].
	\end{multline}
	Now we evaluate each term on the right:%
	\begin{align}
&  \!\!\!\!\!\langle\Gamma|_{AS}\operatorname{Tr}_{RB}[(\partial_{\theta_{j}%
}\rho_{RA}^{\bm{\theta}})(\rho_{RA}^{\bm{\theta}})^{-1}(\partial_{\theta_{k}%
}\rho_{RA}^{\bm{\theta}})\otimes\Gamma_{SB}^{\mathcal{N}^{\bm{\theta}}%
}]|\Gamma\rangle_{AS}\nonumber\\
&  =\langle\Gamma|_{AS}\operatorname{Tr}_{R}[(\partial_{\theta_{j}}\rho
_{RA}^{\bm{\theta}})(\rho_{RA}^{\bm{\theta}})^{-1}(\partial_{\theta_{k}}%
\rho_{RA}^{\bm{\theta}})]\otimes\operatorname{Tr}_{B}[\Gamma_{SB}%
^{\mathcal{N}^{\bm{\theta}}}]|\Gamma\rangle_{AS}\\
&  =\langle\Gamma|_{AS}\operatorname{Tr}_{R}[(\partial_{\theta_{j}}\rho
_{RA}^{\bm{\theta}})(\rho_{RA}^{\bm{\theta}})^{-1}(\partial_{\theta_{k}}%
\rho_{RA}^{\bm{\theta}})]\otimes I_{S}|\Gamma\rangle_{AS}\\
&  =\operatorname{Tr}[(\partial_{\theta_{j}}\rho_{RA}^{\bm{\theta}})(\rho
_{RA}^{\bm{\theta}})^{-1}(\partial_{\theta_{k}}\rho_{RA}^{\bm{\theta}})],\\
&  \!\!\!\!\!\langle\Gamma|_{AS}\operatorname{Tr}_{RB}[(\partial_{\theta_{k}%
}\rho_{RA}^{\bm{\theta}})\otimes(\partial_{\theta_{j}}\Gamma_{SB}%
^{\mathcal{N}^{\bm{\theta}}})]|\Gamma\rangle_{AS}\nonumber\\
&  =\langle\Gamma|_{AS}\operatorname{Tr}_{R}[\partial_{\theta_{k}}\rho
_{RA}^{\bm{\theta}}]\otimes\operatorname{Tr}_{B}[\partial_{\theta_{j}}%
\Gamma_{SB}^{\mathcal{N}^{\bm{\theta}}}]|\Gamma\rangle_{AS}\\
&  =\langle\Gamma|_{AS}\operatorname{Tr}_{R}[\partial_{\theta_{k}}\rho
_{RA}^{\bm{\theta}}]\otimes\partial_{\theta_{j}}\operatorname{Tr}_{B}%
[\Gamma_{SB}^{\mathcal{N}^{\bm{\theta}}}]|\Gamma\rangle_{AS}\\
&  =\langle\Gamma|_{AS}\operatorname{Tr}_{R}[\partial_{\theta_{k}}\rho
_{RA}^{\bm{\theta}}]\otimes(\partial_{\theta_{j}}I_{S})|\Gamma\rangle_{AS}\\
&  =0, \\
&  \!\!\!\!\!\langle\Gamma|_{AS}\operatorname{Tr}_{RB}[(\partial_{\theta_{j}%
}\rho_{RA}^{\bm{\theta}})\otimes(\partial_{\theta_{k}}\Gamma_{SB}%
^{\mathcal{N}^{\bm{\theta}}})]|\Gamma\rangle_{AS}=0, \text{and} \\
&  \!\!\!\!\!\langle\Gamma|_{AS}\operatorname{Tr}_{RB}[\rho_{RA}%
^{\bm{\theta}}\otimes(\partial_{\theta_{j}}\Gamma_{SB}^{\mathcal{N}%
^{\bm{\theta}}})(\Gamma_{SB}^{\mathcal{N}^{\bm{\theta}}})^{-1}(\partial
_{\theta_{k}}\Gamma_{SB}^{\mathcal{N}^{\bm{\theta}}})]|\Gamma\rangle
_{AS}\nonumber\\
&  =\langle\Gamma|_{AS}\rho_{A}^{\bm{\theta}}\otimes\operatorname{Tr}%
_{B}[(\partial_{\theta_{j}}\Gamma_{SB}^{\mathcal{N}^{\bm{\theta}}}%
)(\Gamma_{SB}^{\mathcal{N}^{\bm{\theta}}})^{-1}(\partial_{\theta_{k}}%
\Gamma_{SB}^{\mathcal{N}^{\bm{\theta}}})]|\Gamma\rangle_{AS}\\
&  =\operatorname{Tr}[(\rho_{S}^{\bm{\theta}})^{T}\operatorname{Tr}%
_{B}[(\partial_{\theta_{j}}\Gamma_{SB}^{\mathcal{N}^{\bm{\theta}}}%
)(\Gamma_{SB}^{\mathcal{N}^{\bm{\theta}}})^{-1}(\partial_{\theta_{k}}%
\Gamma_{SB}^{\mathcal{N}^{\bm{\theta}}})].
\end{align}
	
	Substituting back above, we find that
	\begin{align}
	& \widehat{I}_F(\bm{\theta};\{\mathcal{N}_{A\rightarrow B}^{\bm{\theta}}%
	(\rho_{RA}^{\bm{\theta}})\}_{\bm{\theta}})\nonumber\\
	& \leq\sum_{j,k=1}^{D}|j\rangle\!\langle k|\ \operatorname{Tr}[(\partial
	_{\theta_{j}}\rho_{RA}^{\bm{\theta}})(\rho_{RA}^{\bm{\theta}})^{-1}%
	(\partial_{\theta_{k}}\rho_{RA}^{\bm{\theta}})]\nonumber\\
	& \qquad+\sum_{j,k=1}^{D}|j\rangle\!\langle k|\ \operatorname{Tr}[(\rho
	_{S}^{\bm{\theta}})^{T}\operatorname{Tr}_{B}[(\partial_{\theta_{j}}\Gamma
	_{SB}^{\mathcal{N}^{\bm{\theta}}})(\Gamma_{SB}^{\mathcal{N}^{\bm{\theta}}%
	})^{-1}(\partial_{\theta_{k}}\Gamma_{SB}^{\mathcal{N}^{\bm{\theta}}})]\\
	& =\widehat{I}_F(\bm{\theta};\{\rho_{RA}^{\bm{\theta}}\}_{\bm{\theta}}%
	)+\sum_{j,k=1}^{D}|j\rangle\!\langle k|\ \operatorname{Tr}[(\rho_{S}%
	^{\bm{\theta}})^{T}\operatorname{Tr}_{B}[(\partial_{\theta_{j}}\Gamma
	_{SB}^{\mathcal{N}^{\bm{\theta}}})(\Gamma_{SB}^{\mathcal{N}^{\bm{\theta}}%
	})^{-1}(\partial_{\theta_{k}}\Gamma_{SB}^{\mathcal{N}^{\bm{\theta}}})].
	\end{align}
	This concludes the proof.
\end{proof}

\medskip

Now we show how the operator-inequality chain rule can be used to derive a scalar chain rule for the RLD Fisher information value. This is what we will use to prove its amortization collapse and then to establish Cramer--Rao bounds for multiparameter channel estimation in the sequential setting. 

\begin{proposition} \label{prop:rld-chain-rule}
With $W$ an appropriately chosen weight matrix (i.e. $W \geq0$ and $\Tr[W]=1$), let $\{\mathcal{N}_{A\rightarrow B}^{\bm{\theta}}\}_{\bm{\theta}}$ be a differentiable family of quantum channels, and let $\{\rho_{RA}^{\bm{\theta}}\}_{\bm{\theta}}$ be a differentiable family of quantum states. Then the following scalar chain-rule inequality holds:
\begin{equation} \label{eq:multi-RLD-chain-rule}
	\widehat{I}_F(\bm{\theta},W;\{\mathcal{N}_{A\rightarrow B}^{\bm{\theta}}(\rho_{RA}^{\bm{\theta}})\}_{\bm{\theta}})\leq  \widehat{I}_F(\bm{\theta},W;\{\mathcal{N}_{A\rightarrow B}^{\bm{\theta}}\}_{\bm{\theta}}) +\widehat{I}_F(\bm{\theta},W;\{\rho_{RA}^{\bm{\theta}}\}_{\bm{\theta}}).
\end{equation}
\end{proposition}

\begin{proof}
We start by restating the operator-inequality chain rule in \eqref{app-eq:operator-chain-rule}:
\begin{equation}
\widehat{I}_F(\bm{\theta};\{\mathcal{N}_{A\rightarrow B}^{\bm{\theta}}(\rho
_{RA}^{\bm{\theta}})\}_{\bm{\theta}})\leq\sum_{j,k=1}^{D}|j\rangle\!\langle
k|\operatorname{Tr}[(\rho_{S}^{\bm{\theta}})^{T}\operatorname{Tr}%
_{B}[(\partial_{\theta_{j}}\Gamma_{SB}^{\mathcal{N}^{\bm{\theta}}}%
)(\Gamma_{SB}^{\mathcal{N}^{\bm{\theta}}})^{-1}(\partial_{\theta_{k}}%
\Gamma_{SB}^{\mathcal{N}^{\bm{\theta}}})]
+\widehat{I}_F(\bm{\theta};\{\rho_{RA}^{\bm{\theta}}\}_{\bm{\theta}}).
\end{equation}
	
We sandwich the above with $W^{1/2}$ on both sides and take the trace (a positive map overall) to preserve the inequality and obtain
\begin{multline}
\Tr[ W \widehat{I}_F(\bm{\theta};\{\mathcal{N}_{A\rightarrow B}^{\bm{\theta}}(\rho
_{RA}^{\bm{\theta}})\}_{\bm{\theta}}) ] \leq
\Tr \!\left[ W\sum_{j,k=1}^{D}|j\rangle\!\langle
k|\operatorname{Tr}[(\rho_{S}^{\bm{\theta}})^{T}\operatorname{Tr}%
_{B}[(\partial_{\theta_{j}}\Gamma_{SB}^{\mathcal{N}^{\bm{\theta}}}%
)(\Gamma_{SB}^{\mathcal{N}^{\bm{\theta}}})^{-1}(\partial_{\theta_{k}}%
\Gamma_{SB}^{\mathcal{N}^{\bm{\theta}}})]]\right] \\
+ \Tr[ W\widehat{I}_F(\bm{\theta};\{\rho_{RA}^{\bm{\theta}}\}_{\bm{\theta}})].
\end{multline}	
Using the definition of the RLD Fisher information value in~\eqref{eq:rld-value-def}, the above simplifies to
\begin{multline} \label{app-eq:chain-rule-intermediate}
\widehat{I}_F(\bm{\theta},W;\{\mathcal{N}_{A\rightarrow B}^{\bm{\theta}}%
(\rho_{RA}^{\bm{\theta}})\}_{\bm{\theta}}) \leq 
\Tr\! \left[ W\sum_{j,k=1}^{D}|j\rangle\!\langle
k|\operatorname{Tr}[(\rho_{S}^{\bm{\theta}})^{T}\operatorname{Tr}%
_{B}[(\partial_{\theta_{j}}\Gamma_{SB}^{\mathcal{N}^{\bm{\theta}}}%
)(\Gamma_{SB}^{\mathcal{N}^{\bm{\theta}}})^{-1}(\partial_{\theta_{k}}%
\Gamma_{SB}^{\mathcal{N}^{\bm{\theta}}})] ]\right] \\
+ \widehat{I}_F(\bm{\theta},W;\{\rho_{RA}^{\bm{\theta}}\}_{\bm{\theta}}).
\end{multline}
Now we consider that
\begin{align}
& \Tr \!\left[ W\sum_{j,k=1}^{D}|j\rangle\!\langle
k|\operatorname{Tr}[(\rho_{S}^{\bm{\theta}})^{T}\operatorname{Tr}
_{B}[(\partial_{\theta_{j}}\Gamma_{SB}^{\mathcal{N}^{\bm{\theta}}}
)(\Gamma_{SB}^{\mathcal{N}^{\bm{\theta}}})^{-1}(\partial_{\theta_{k}}
\Gamma_{SB}^{\mathcal{N}^{\bm{\theta}}})]\right] \notag \\
& \qquad = \sum_{j,k=1}^D \langle k | W | j \rangle \Tr\!\left[ (\rho_{S}^{\bm{\theta}})^{T} \operatorname{Tr}%
_{B}[(\partial_{\theta_{j}}\Gamma_{SB}^{\mathcal{N}^{\bm{\theta}}}%
)(\Gamma_{SB}^{\mathcal{N}^{\bm{\theta}}})^{-1}(\partial_{\theta_{k}}%
\Gamma_{SB}^{\mathcal{N}^{\bm{\theta}}})]\right] \\
& \qquad = \Tr\!\left[ (\rho_{S}^{\bm{\theta}})^{T} \sum_{j,k=1}^D \langle k | W | j \rangle \operatorname{Tr}%
_{B}[(\partial_{\theta_{j}}\Gamma_{SB}^{\mathcal{N}^{\bm{\theta}}}%
)(\Gamma_{SB}^{\mathcal{N}^{\bm{\theta}}})^{-1}(\partial_{\theta_{k}}%
\Gamma_{SB}^{\mathcal{N}^{\bm{\theta}}})]  \right] \\
& \qquad \leq \sup_{\rho^{\bm{\theta}}_S }  \Tr \!\left[ (\rho_{S}^{\bm{\theta}})^{T} \sum_{j,k=1}^D \langle k | W | j \rangle \operatorname{Tr}%
_{B}[(\partial_{\theta_{j}}\Gamma_{SB}^{\mathcal{N}^{\bm{\theta}}}%
)(\Gamma_{SB}^{\mathcal{N}^{\bm{\theta}}})^{-1}(\partial_{\theta_{k}}%
\Gamma_{SB}^{\mathcal{N}^{\bm{\theta}}})]  \right] \\
& \qquad =  \left\Vert \sum_{j,k=1}^D \langle k | W | j \rangle \operatorname{Tr}%
_{B}[(\partial_{\theta_{j}}\Gamma_{SB}^{\mathcal{N}^{\bm{\theta}}}%
)(\Gamma_{SB}^{\mathcal{N}^{\bm{\theta}}})^{-1}(\partial_{\theta_{k}}%
\Gamma_{SB}^{\mathcal{N}^{\bm{\theta}}})]   \right\Vert_\infty \\
& \qquad = \widehat{I}_F
(\bm{\theta},W;\{\mathcal{N}_{A\rightarrow B}^{\bm{\theta}}\}_{\bm{\theta}}. \
\end{align}

Using the above and \eqref{app-eq:chain-rule-intermediate}, we conclude that
\begin{equation}
\widehat{I}_F(\bm{\theta},W;\{\mathcal{N}_{A\rightarrow B}^{\bm{\theta}}%
(\rho_{RA}^{\bm{\theta}})\}_{\bm{\theta}}) \leq \widehat{I}_F
(\bm{\theta},W;\{\mathcal{N}_{A\rightarrow B}^{\bm{\theta}}\}_{\bm{\theta}}
+ \widehat{I}_F(\bm{\theta},W;\{\rho_{RA}^{\bm{\theta}}\}_{\bm{\theta}}).
\end{equation}
This concludes the proof.
\end{proof}

We now show how the chain rule proved above results in an amortization collapse for the RLD Fisher information value of quantum channels.

\begin{corollary}
	Let $\{\mathcal{N}_{A\rightarrow B}^{\bm{\theta}}\}_{\bm{\theta}}$ be a differentiable family of quantum channels. Then the following amortization collapse occurs for the RLD Fisher information value of quantum channels:
	\begin{equation}
		\widehat{I}_F^{\mathcal{A}}(\bm{\theta},W;\{\mathcal{N}_{A\rightarrow B}^{\bm{\theta}}\}_{\bm{\theta}}) = \widehat{I}_F(\bm{\theta},W;\{\mathcal{N}_{A\rightarrow B}^{\bm{\theta}}\}_{\bm{\theta}}). \label{eq:rld-fisher-value-amort-collapse}
	\end{equation} 
\end{corollary}

\begin{proof}
	The chain rule in \eqref{eq:multi-RLD-chain-rule} can be rewritten as
\begin{equation}
	\widehat{I}_F(\bm{\theta},W;\{\mathcal{N}_{A\rightarrow B}^{\bm{\theta}}(\rho_{RA}^{\bm{\theta}})\}_{\bm{\theta}}) - \widehat{I}_F(\bm{\theta},W;\{\rho_{RA}^{\bm{\theta}}\}_{\bm{\theta}}) \leq  \widehat{I}_F(\bm{\theta},W;\{\mathcal{N}_{A\rightarrow B}^{\bm{\theta}}\}_{\bm{\theta}}).
\end{equation}
Taking the supremum over input-state families on the left-hand side gives
\begin{equation}
	\widehat{I}_F^{\mathcal{A}}(\bm{\theta},W;\{\mathcal{N}_{A\rightarrow B}^{\bm{\theta}}\}_{\bm{\theta}}) \leq \widehat{I}_F(\bm{\theta},W;\{\mathcal{N}_{A\rightarrow B}^{\bm{\theta}}\}_{\bm{\theta}}).
\end{equation}
When the above is combined with the direct inequality obeyed by the amortized RLD Fisher information value \eqref{eq:multiparam-rld-amort-ineq-obvi-dir}, we obtain the desired amortization collapse in \eqref{eq:rld-fisher-value-amort-collapse}.
\end{proof}

\section{Limits on multiparameter channel estimation using RLD Fisher information value of quantum channels}

Now that we have established the amortization collapse for the RLD Fisher information value, we are in a position to use it to establish a Cramer--Rao bound for sequential channel estimation. However, just as we did in Chapter~\ref{ch:single} for the single parameter case, we need to first prove a meta-converse theorem that connects the amortized RLD Fisher information value to the RLD Fisher information value achieved by a sequential channel estimation protocol.

\begin{figure}
\centering
\includegraphics[width=\linewidth]{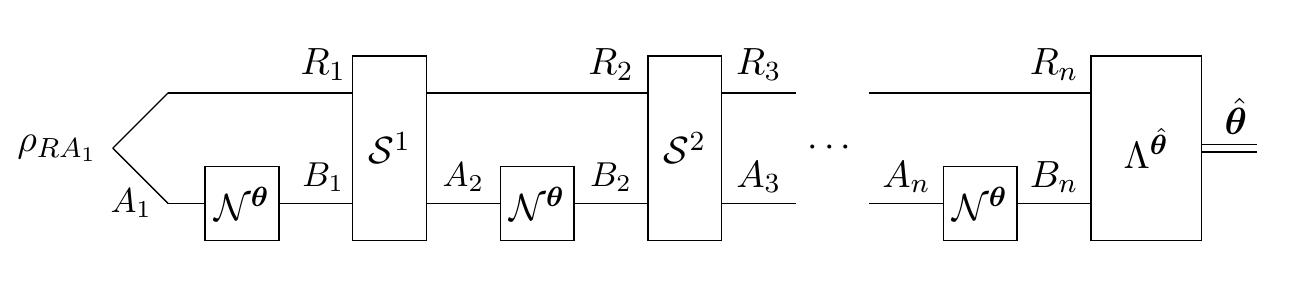}
\caption{Processing a channel sequentially for multiparameter estimation}{Processing $n$ uses of channel $\mathcal{N}^{\bm{\theta}}$ in a sequential or adaptive manner is the most general approach to channel parameter estimation. The $n$ uses of the channel are interleaved with $n-1$ quantum channels $\mathcal{S}^{1}$ through $\mathcal{S}^{n-1}$, which can also share memory systems with each other. The final measurement's outcome is then used to obtain an estimate of the unknown parameter vector~$\bm{\theta}$.  }
\label{fig:multi-sequential-strategies}
\end{figure}

\begin{theorem} \label{thm:multi-meta-converse}
Consider a general sequential estimation protocol as described in detail for the single parameter case in Chapter~\ref{ch:prelims} and reproduced for the multiparameter case here in Figure~\ref{fig:multi-sequential-strategies}. The following inequality holds:
\begin{equation}
\widehat{I}_{F}(\bm{\theta}, W;\{\omega_{R_{n}B_{n}}^{\bm{\theta}}\}_{\bm{\theta}})\leq
n\cdot \widehat{I}_{F}^{\mathcal{A}}(\bm{\theta}, W;\{\mathcal{N}_{A\rightarrow B}^{\bm{\theta}
}\}_{\bm{\theta}}), \label{eq:multiparam-meta-converse}
\end{equation}
where $\omega_{R_{n}B_{n}}^{\bm{\theta}}$ is the final state of an $n$-round sequential estimation protocol, and $\mathcal{S}^1$ through $\mathcal{S}^{n-1}$ are interleaving quantum channels in the protocol, both as in Figure \ref{fig:multi-sequential-strategies}.
\end{theorem}

\begin{proof}
Consider that
\begin{align}
&  \widehat{I}_{F}(\bm{\theta}, W;\{\omega_{R_{n}B_{n}}^{\bm{\theta}}\}_{\bm{\theta}})\nonumber\\
&  =\widehat{I}_{F}(\bm{\theta}, W;\{\omega_{R_{n}B_{n}}^{\bm{\theta}}\}_{\bm{\theta}
})-\widehat{I}_{F}(\bm{\theta}, W;\{\rho_{R_{1}A_{1}}\}_{\bm{\theta}})\\
&  =\widehat{I}_{F}(\bm{\theta}, W;\{\omega_{R_{n}B_{n}}^{\bm{\theta}}\}_{\bm{\theta}
})-\widehat{I}_{F}(\bm{\theta}, W;\{\rho_{R_{1}A_{1}}\}_{\bm{\theta}}) \\
&\qquad + \sum_{i=2}^{n}\left(
\widehat{I}_{F}(\bm{\theta}, W;\{\rho_{R_{i}A_{i}}^{\bm{\theta}}\}_{\bm{\theta}})-\widehat{I}
_{F}(\bm{\theta}, W;\{\rho_{R_{i}A_{i}}^{\bm{\theta}}\}_{\bm{\theta}})\right) \\
&  =\widehat{I}_{F}(\bm{\theta}, W;\{\omega_{R_{n}B_{n}}^{\bm{\theta}}\}_{\bm{\theta}
})-\widehat{I}_{F}(\bm{\theta}, W;\{\rho_{R_{1}A_{1}}\}_{\bm{\theta}})\nonumber\\
&  \qquad+\sum_{i=2}^{n}\left(  \widehat{I}_{F}(\bm{\theta}, W;\{\mathcal{S}
_{R_{i-1}B_{i-1}\rightarrow R_{i}A_{i}}^{i-1}(\rho_{R_{i-1}B_{i-1}}^{\bm{\theta}
})\}_{\bm{\theta}})-\widehat{I}_{F}(\bm{\theta}, W;\{\rho_{R_{i}A_{i}}^{\bm{\theta}}\}_{\bm{\theta}
})\right) \\
&  \leq\widehat{I}_{F}(\bm{\theta}, W;\{\omega_{R_{n}B_{n}}^{\bm{\theta}}\}_{\bm{\theta}
})-\widehat{I}_{F}(\bm{\theta}, W;\{\rho_{R_{1}A_{1}}\}_{\bm{\theta}})\nonumber\\
&  \qquad+\sum_{i=2}^{n}\left(  \widehat{I}_{F}(\bm{\theta}, W;\{\rho_{R_{i-1}B_{i-1}
}^{\bm{\theta}}\}_{\bm{\theta}})-\widehat{I}_{F}(\bm{\theta}, W;\{\rho_{R_{i}A_{i}}^{\bm{\theta}
}\}_{\bm{\theta}})\right) \\
&  =\sum_{i=1}^{n}\left(  \widehat{I}_{F}(\bm{\theta}, W;\{\rho_{R_{i}B_{i}}^{\bm{\theta}
}\}_{\bm{\theta}})-\widehat{I}_{F}(\bm{\theta}, W;\{\rho_{R_{i}A_{i}}^{\bm{\theta}}\}_{\bm{\theta}
}\right) \\
&  =\sum_{i=1}^{n}\left(  \widehat{I}_{F}(\bm{\theta}, W;\{\mathcal{N}_{A_{i}
	\rightarrow B_{i}}^{\bm{\theta}}(\rho_{R_{i}A_{i}}^{\bm{\theta}})\}_{\bm{\theta}}
)-\widehat{I}_{F}(\bm{\theta}, W;\{\rho_{R_{i}A_{i}}^{\bm{\theta}}\}_{\bm{\theta}}\right) \\
&  \leq n\cdot\sup_{\{\rho_{RA}^{\bm{\theta}}\}_{\bm{\theta}}}\left[  \widehat{I}
_{F}(\bm{\theta}, W;\{\mathcal{N}_{A\rightarrow B}^{\bm{\theta}}(\rho_{RA}^{\bm{\theta}
})\}_{\bm{\theta}})-\widehat{I}_{F}(\bm{\theta}, W;\{\rho_{RA}^{\bm{\theta}})\}_{\bm{\theta}})\right]
\\
&  =n\cdot \widehat{I}_{F}^{\mathcal{A}}(\bm{\theta}, W;\{\mathcal{N}_{A\rightarrow
	B}^{\bm{\theta}}\}_{\theta}).
\end{align}
The first equality follows because the initial state $\rho_{R_1 A_1}$ has no dependence on any of the parameters in $\bm{\theta}$. The first inequality arises due to the data-processing inequality for the RLD Fisher information. The other steps are straightforward manipulations.
\end{proof}

We have now assembled all the necessary components needed to establish the main result of this Chapter, a multiparameter Cramer--Rao bound for sequential channel estimation. 

\begin{theorem} \label{thm:single-letter-multi-param-crb}
	For a differentiable channel family $\{\mathcal{N}_{A\rightarrow B}^{\bm{\theta}}\}_{\bm{\theta}}$ and a positive semidefinite weight matrix $W$ with $\Tr[W] = 1$, the following multiparameter Cramer--Rao bound holds:
	\begin{equation}
		\Tr[ W \operatorname{Cov}(\bm{\theta}) ] \geq \frac{1}{n  \widehat{I}_F (\bm{\theta},W;\{ \mathcal{N}_{A\rightarrow B}^{\bm{\theta}} \}_{\bm{\theta}}) }. \label{eq:single-letter-multi-param-crb}
	\end{equation}
\end{theorem}

\begin{proof}
	The multiparameter Cramer--Rao bound follows as a direct consequence of the chain rule \eqref{eq:multi-RLD-chain-rule}, ensuing amortization collapse \eqref{eq:rld-fisher-value-amort-collapse}, as well as the meta-converse in \eqref{eq:multiparam-meta-converse}.
\end{proof}

The bound in~\eqref{eq:single-letter-multi-param-crb} is single-letter. That is, even though the bound holds for an $n$-use sequential strategy, it requires an evaluation of a quantity involving only a single channel use, i.e $\widehat{I}_F (\bm{\theta},W;\{ \mathcal{N}_{A\rightarrow B}^{\bm{\theta}} \}_{\bm{\theta}})$. Further, quantum channels encompass all possible evolutions allowed by quantum mechanics. Our Cramer--Rao bound holds for estimating channels in the sequential setting and therefore applies to all possible estimation tasks. Another desirable aspect of our bound is that the RLD Fisher information value can be computed using a semi-definite program, details of which we provide later in Section~\ref{sec:multiparam-sdp}.

\section{Example: Estimating the parameters of the generalized amplitude damping channel}

In Chapter~\ref{ch:single}, we evaluated the RLD Fisher information of a generalized amplitude damping channel (GADC) for various single parameter estimation tasks. We can generalize that example and use the techniques developed in this chapter to evaluate Cramer--Rao bounds for the task of simultaneously estimating the noise and loss parameters of a GADC. First, we repeat the Choi operator of a GADC $\mathcal{A}_{\gamma,N}$ with loss parameter $\gamma$ and noise parameter $N$ that we gave in Chapter~\ref{ch:single}.
\begin{equation}
	\Gamma_{RB}^{\mathcal{A}_{\gamma,N}}  \coloneqq 
	\begin{bmatrix}
		1-\gamma N & 0 & 0 & \sqrt{1-\gamma}\\
		0 & \gamma N & 0 & 0\\
		0 & 0 & \gamma\left(  1-N\right)  & 0\\
		\sqrt{1-\gamma} & 0 & 0 & 1-\gamma\left(  1-N\right)
	\end{bmatrix}.
\end{equation}

For the purposes of this example, let us choose
\begin{equation}
W = \frac{1}{4}  \begin{pmatrix}1&1\\1&3\end{pmatrix}
,
\label{eq:W-matrix-choice-end}
\end{equation}
which satisfies the requirements for a valid weight matrix (positive semi-definite with unit trace). To compute the RLD bound in \eqref{eq:single-letter-multi-param-crb}, we calculate
\begin{multline} \label{eq:gadc-rld-info-value}
	\widehat{I}_F(\{\gamma, N\}, W; \{\mathcal{A}_{\gamma, N} \}_{\gamma, N})  = \frac{1}{4} \Big\Vert \Tr_B[ (\partial_{\gamma} \Gamma^{\mathcal{A}}) (\Gamma^{\mathcal{A}})^{-1} (\partial_\gamma \Gamma^{\mathcal{A}}) + (\partial_{\gamma} \Gamma^{\mathcal{A}}) (\Gamma^{\mathcal{A}})^{-1} (\partial_N \Gamma^{\mathcal{A}}) \\ + (\partial_{N} \Gamma^{\mathcal{A}}) (\Gamma^{\mathcal{A}})^{-1} (\partial_\gamma \Gamma^{\mathcal{A}}) + 3 (\partial_{N} \Gamma^{\mathcal{A}}) (\Gamma^{\mathcal{A}})^{-1} (\partial_N \Gamma^{\mathcal{A}})]  \Big\Vert_{\infty}.
\end{multline}
where $\mathcal{A}$ is used as shorthand for $\mathcal{A}_{\gamma, N}$ and the system labels for $RB$ are omitted. As a consequence of the Cramer--Rao bound \eqref{eq:single-letter-multi-param-crb}, the inverse of \eqref{eq:gadc-rld-info-value} is a lower bound on $\Tr[ W \text{Cov}(\{ \gamma, N \}) ] $.

We have
\begin{align}
	\partial_{\gamma}\Gamma_{RB}^{\mathcal{A}_{\gamma,N}} &=
\begin{bmatrix}
-N & 0 & 0 & -\frac{1}{2\sqrt{1-\gamma}}\\
0 & N & 0 & 0\\
0 & 0 & 1-N & 0\\
-\frac{1}{2\sqrt{1-\gamma}} & 0 & 0 & -\left(  1-N\right)
\end{bmatrix}, \text{ and} \\
	\partial_{N}\Gamma_{RB}^{\mathcal{A}_{\gamma,N}} &=-\gamma\left(  I_{2}
\otimes\sigma_{Z}\right)  .
\end{align}

We then have
\begin{align}
	\Tr_B[(\partial_{\gamma} \Gamma^{\mathcal{A}_{\gamma, N}}) (\Gamma^{\mathcal{A}_{\gamma, N}})^{-1} (\partial_\gamma \Gamma^{\mathcal{A}_{\gamma, N}})  ] &=
	\begin{bmatrix}
		\frac{\frac{1}{N-\gamma N}+\frac{1}{1-N}-4}{4 \gamma^2} & 0 \\
		0 & \frac{\frac{1}{(\gamma-1) (N-1)}+\frac{1}{N}-4}{4 \gamma^2}
	\end{bmatrix},
	\\
	\Tr_B[(\partial_{\gamma} \Gamma^{\mathcal{A}_{\gamma, N}}) (\Gamma^{\mathcal{A}_{\gamma, N}})^{-1} (\partial_N \Gamma^{\mathcal{A}_{\gamma, N}})  ] &= 
	\begin{bmatrix}
	-\frac{1-2 N}{2 \gamma N(1-N)}  & 0 \\
	0 & -\frac{1-2 N}{2 \gamma N(1-N)}
	\end{bmatrix},
	\\
	\Tr_B[(\partial_{N} \Gamma^{\mathcal{A}_{\gamma, N}}) (\Gamma^{\mathcal{A}_{\gamma, N}})^{-1} (\partial_\gamma \Gamma^{\mathcal{A}_{\gamma, N}})  ] &=
	\begin{bmatrix}
	-\frac{1-2 N}{2 \gamma N(1-N)}  & 0 \\
	0 & -\frac{1-2 N}{2 \gamma N(1-N)}
	\end{bmatrix},
	\\
	\Tr_B[(\partial_{N} \Gamma^{\mathcal{A}_{\gamma, N}}) (\Gamma^{\mathcal{A}_{\gamma, N}})^{-1} (\partial_N \Gamma^{\mathcal{A}_{\gamma, N}})  ] &= 
	\begin{bmatrix}
		\frac{1}{N(1-N)} & 0 \\
		0 & \frac{1}{N(1-N)}
	\end{bmatrix},
\end{align}
with which the expression in \eqref{eq:gadc-rld-info-value} is directly calculated. Furthermore, we verify our direct calculation by explicitly calculating~\eqref{eq:gadc-rld-info-value} using the semi-definite program for the RLD Fisher information value which we later provide explicitly in Proposition \ref{prop:rld-channels-sdp}. We find agreement between the two approaches up to eight digits of precision.

We compare the RLD Fisher information bound to the generalized Helstrom Cramer--Rao bound \cite{Albarelli2020}. In this case, it reduces to the SLD Fisher information bound as we are in the regime of parametric, rather than semiparametric, estimation. This lower bound can be achieved asymptotically up to a constant prefactor \cite{Guta2006, Hayashi2008, Yamagata2013, Yang2019a, Tsang2019a}.

To calculate the SLD Fisher information, we choose input probe state $|\psi(p)\rangle \coloneqq \sqrt{p} \ket{00} + \sqrt{1-p} \ket{11}$, with the same $W$ as in the RLD calculation. It suffices to optimize the single-copy SLD Fisher information over such states due to the $\sigma_Z$ covariance of the channel $\mathcal{A}_{\gamma, N}$. This SLD bound for estimating just the loss parameter $\gamma$ of a GADC was calculated in ~\cite{Fujiwara2004}.

With our particular choice of input state, we have
\begin{equation}
	\rho_{\gamma, N} \coloneqq \mathcal{A}_{\gamma,N}(\psi_p) = 
	\begin{bmatrix}
		p (1 - \gamma N) & 0 & 0 & \sqrt{p(1-p)(1-\gamma)} \\
		0 & p \gamma N & 0 & 0 \\
		0 & 0 & (1-p) \gamma (1-N) & 0 \\
		\sqrt{p(1-p)(1-\gamma)} & 0 & 0 & (1-p)( 1 - (1-N)\gamma )
	\end{bmatrix}.
\end{equation}

The SLD Fisher information takes the form of a $2 \times 2$ matrix:
\begin{equation} \label{eq:app-sld-fisher-matrix}
	I_F({\gamma, N}; \{ \rho_{\gamma, N} \}_{\gamma, N})_{jk} = \Tr\!\left[ \rho_{\gamma, N} L_j L_k \right] ,
\end{equation}
where $j$ and $k$ each take values $1$ or $2$ that correspond to either $\gamma$ or $N$. The SLD operators $L_j$ are defined implicitly via
$\partial_j \rho_{\gamma, N} = \frac{1}{2} \left( \rho_{\gamma, N} L_j + L_j \rho_{\gamma, N} \right)$.

Suppose that $\rho_{\gamma, N}$ has spectral decomposition
\begin{equation}
\rho_{\gamma, N} = \sum_j \lambda_{\gamma, N}^j | \psi_{\gamma, N}^j \rangle \langle \psi_{\gamma, N}^j |.
\end{equation}
We then have
\begin{align}
L_\gamma &= 2\sum_{j,k:\lambda^{j}_{\gamma, N}+\lambda^{k}_{\gamma, N}>0} \frac{\langle\psi^{j}_{\gamma, N}|(\partial_{\gamma}\rho_{\gamma, N})|\psi^{k}_{\gamma, N}\rangle}{\lambda^{j}_{\gamma, N}+\lambda^{k}_{\gamma, N}}|\psi^{j}_{\gamma, N}\rangle\!\langle\psi^{k}_{\gamma, N}| \text{~and} \\
L_N &= 2\sum_{j,k:\lambda^{j}_{\gamma, N}+\lambda^{k}_{\gamma, N}>0} \frac{\langle\psi^{j}_{\gamma, N}|(\partial_{N}\rho_{\gamma, N})|\psi^{k}_{\gamma, N}\rangle}{\lambda^{j}_{\gamma, N}+\lambda^{k}_{\gamma, N}}|\psi^{j}_{\gamma, N}\rangle\!\langle\psi^{k}_{\gamma, N}|,
\end{align}
which enables us to calculate the elements of the SLD Fisher information matrix given in~\eqref{eq:app-sld-fisher-matrix}.

\begin{figure}
 	\begin{subfigure}{.48\textwidth}
 		\centering
 		\includegraphics[width=0.9\linewidth]{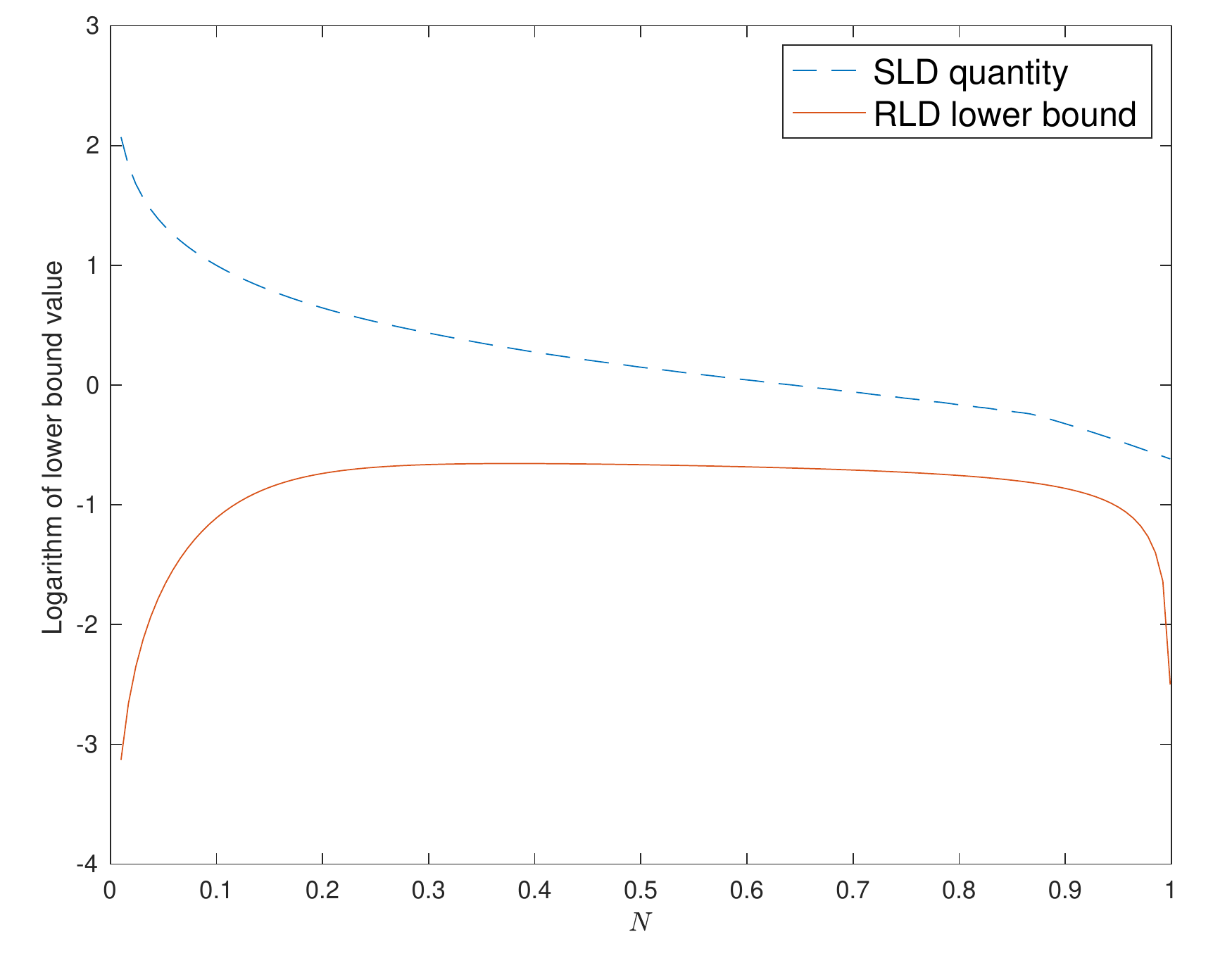}
 		\caption{}
 		\label{fig:fixed-N-0-2}
 	\end{subfigure}
	\begin{subfigure}{.48\textwidth}
		\centering
		\includegraphics[width=0.9\linewidth]{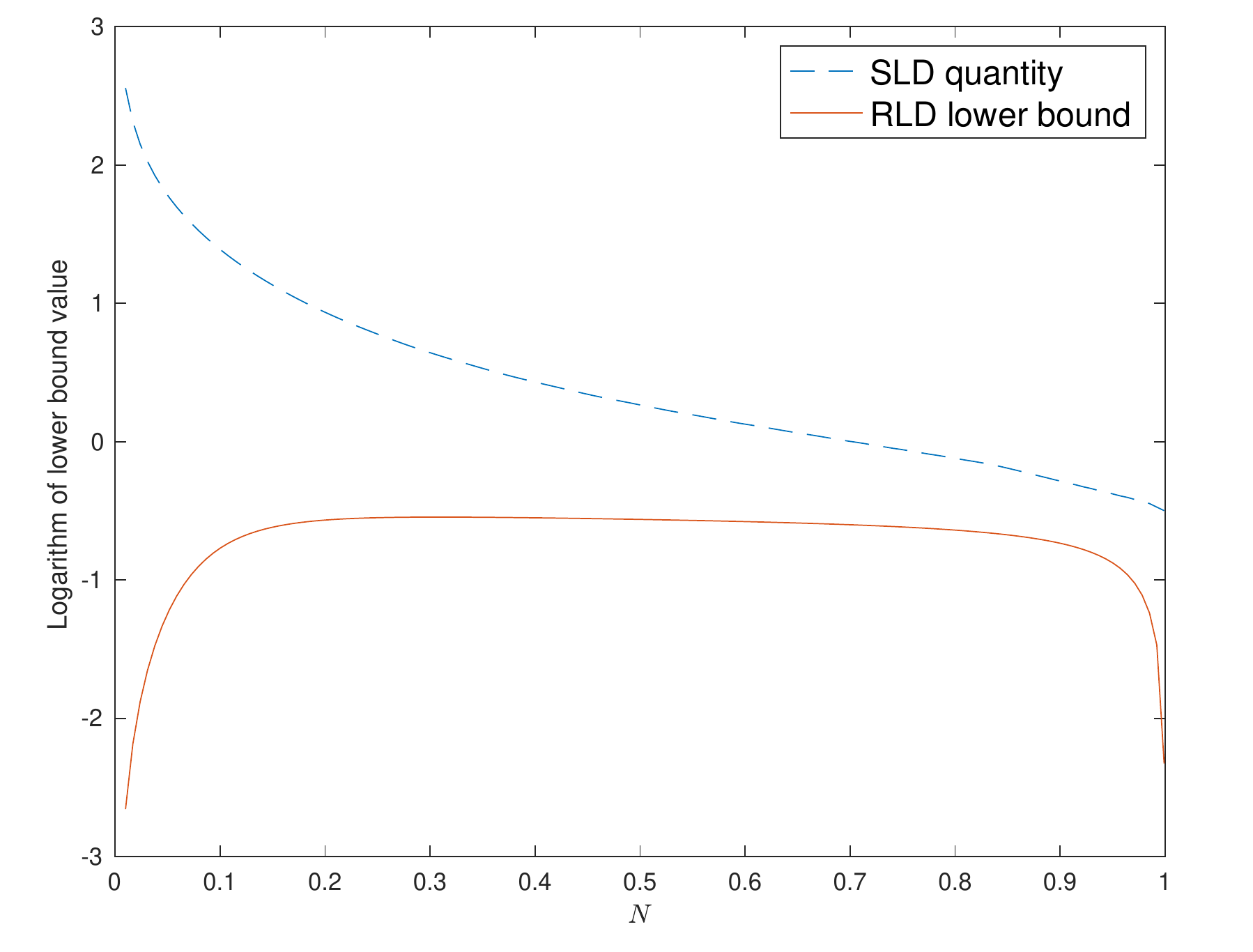}
		\caption{}
		\label{fig:fixed-N-0-3}
	\end{subfigure}
	\caption{Comparing RLD and SLD bounds for multiparameter estimation of a generalized amplitude damping channel, with fixed noise}{Logarithms of the SLD quantity \eqref{eq:sld-quantity-gadc} and the inverse of the RLD Fisher information value \eqref{eq:gadc-rld-info-value} versus loss $\gamma$ with fixed noise $N$. In (a), $N=0.2$, and in (b), $N=0.3$.}
	\label{fig:gadc-estimation-two-parameters-fixed-N}
	\end{figure}
	\begin{figure}
	\begin{subfigure}{.48\textwidth}
		\centering
		\includegraphics[width=0.9\linewidth]{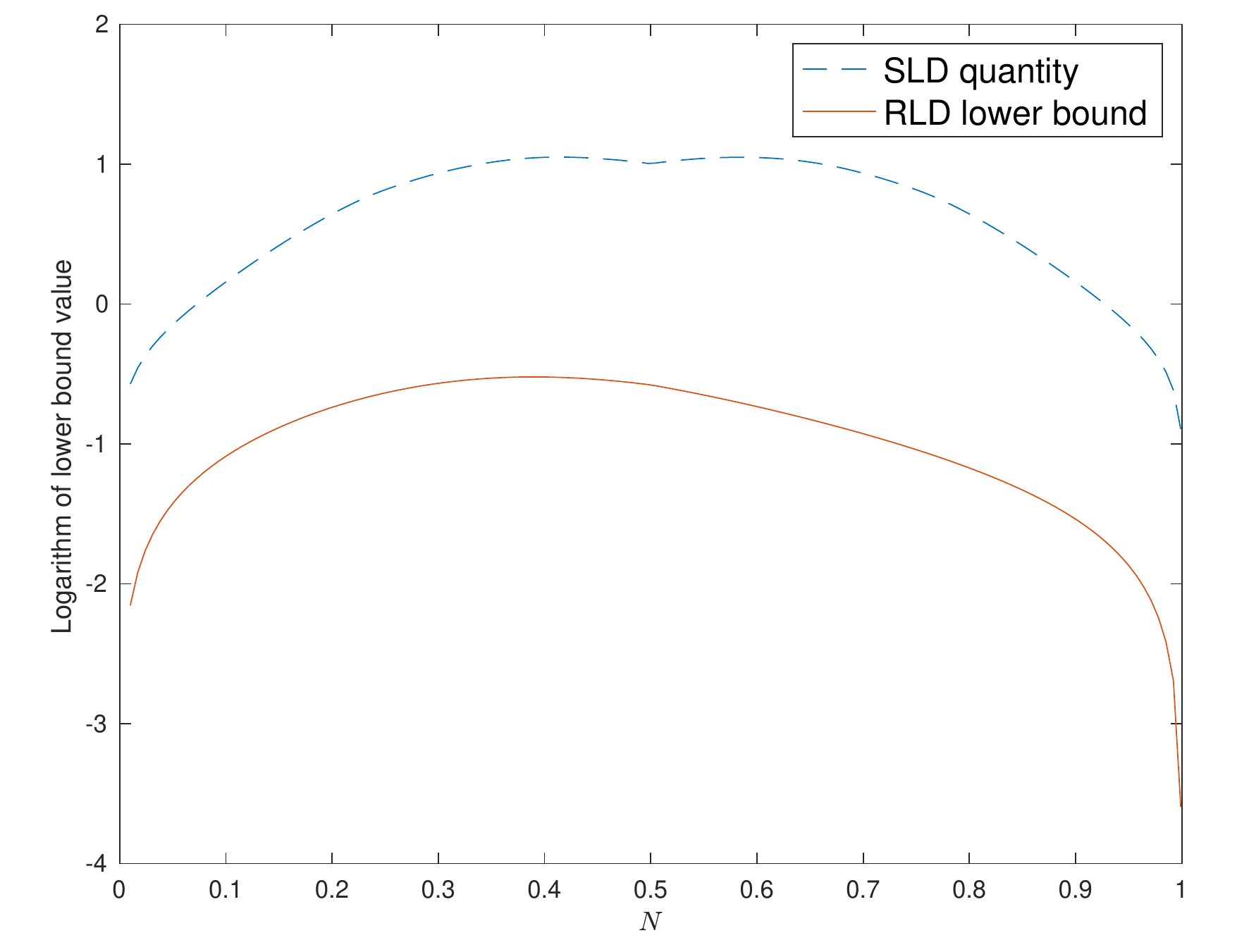}
		\caption{}
		\label{fig:fixed-g-0-2}
	\end{subfigure}
	\begin{subfigure}{.48\textwidth}
		\centering
		\includegraphics[width=0.9\linewidth]{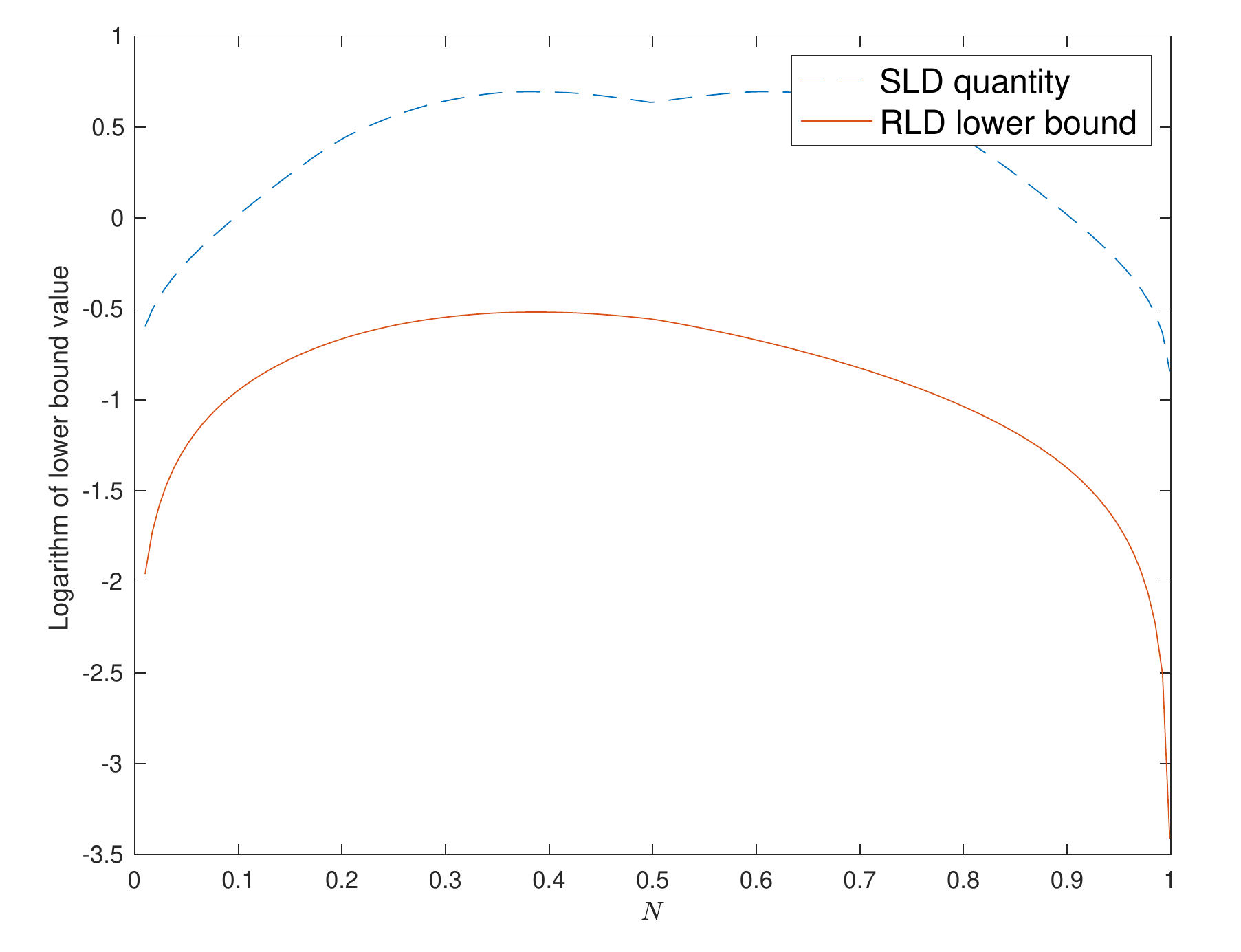}
		\caption{}
		\label{fig:fixed-g-0-3}
	\end{subfigure}
	\caption{Comparing RLD and SLD bounds for multiparameter estimation of a generalized amplitude damping channel, with fixed loss}{Logarithms of the SLD quantity \eqref{eq:sld-quantity-gadc} and the inverse of the RLD Fisher information value \eqref{eq:gadc-rld-info-value} versus noise $N$ with fixed loss $\gamma$. In (a), $\gamma=0.2$, and in (b), $\gamma=0.3$. In each of the four figures above, both lines indicate lower bounds on the quantity $\Tr[ W \text{Cov} (\{ \gamma, N \} )]$ where $W$ is chosen in \eqref{eq:W-matrix-choice-end} (as discussed, the SLD Fisher information value is a lower bound up to a constant prefactor). For the SLD quantity, we optimize over input states of the form $\sqrt{p} \ket{00}+ \sqrt{1-p} \ket{11}$.}
	\label{fig:gadc-estimation-two-parameters-fixed-g}
\end{figure}

We optimize over the parameter $p$ to obtain the SLD Fisher information bound. To be clear, we minimize
\begin{equation} \label{eq:sld-quantity-gadc}
\Tr[W [I_F(\{\gamma,N\}; \{\mathcal{A}_{\gamma, N}(\psi(p)))\}_{\{\gamma,N\}}]^{-1}]
\end{equation}
with respect to $p$. We compare the SLD quantity \eqref{eq:sld-quantity-gadc} to our RLD lower bound \eqref{eq:gadc-rld-info-value} in Figures~\ref{fig:gadc-estimation-two-parameters-fixed-N} and \ref{fig:gadc-estimation-two-parameters-fixed-g}. In Figure~\ref{fig:gadc-estimation-two-parameters-fixed-N}, we keep $N$ fixed and vary $\gamma$ from 0 to 1. In Figure~\ref{fig:gadc-estimation-two-parameters-fixed-g}, instead we keep $\gamma$ fixed and vary $N$ from 0 to 1. We find that the RLD lower bound is within one to two orders of magnitude of the SLD Fisher information.

\section{Optimization formulae} \label{sec:multiparam-sdp}

Finally, we provide semi-definite programs to compute the RLD Fisher information value of both quantum states and quantum channels. 

\subsection{Semi-definite program for RLD\ Fisher information value of quantum states}

\begin{proposition}
	Let $\{\rho_{A}^{\bm{\theta}}\}_{\bm{\theta}}$ be a differentiable family of quantum states, and let $W$ be a $D\times D$ weight matrix.
	Suppose that the finiteness condition for the RLD Fisher information value \eqref{eq:app:state-RLD-finiteness-cond} holds.
	Then the RLD Fisher information value of quantum states can be calculated via the following semi-definite program:
	\begin{multline} \label{eq:primal-rld-value-states}
	\widehat{I}_{F}(\bm{\theta}, W; \{\rho_{A}^{\bm{\theta}}\})=\inf\Big\{  \operatorname{Tr}
	[(W_F \otimes I_A)M_{FA}]:M_{FA}\geq0, \\
	\begin{bmatrix}
	M_{FA} & \sum_{j =1}^D \ket{j}_F \otimes (\partial_{\theta_j} \rho^{{\bm{\theta}}}_{A} )\\
	\sum_{j =1}^D \bra{j}_F \otimes (\partial_{\theta_j} \rho^{{\bm{\theta}}}_{A} ) & \rho_A^{\bm{\theta}}
	\end{bmatrix}
	\geq0\Big\}  .
	\end{multline}
	
	The dual program is given by
	\begin{equation} \label{eq:dual-rld-value-states}
	\sup_{P_{FA}, Q_{FA\to A}, R_A} 2 \left(\sum_{j=1}^D \operatorname{Re}[\operatorname{Tr}
	[ Q_{FA\to A} ( | j \rangle_F \otimes  (\partial_{\theta_j} \rho_A^{\bm{\theta}})) ]] \right)-\operatorname{Tr}[R_A \rho_A^{\bm{\theta}}],
	\end{equation}
	subject to
	\begin{equation}
	P_{FA} \leq (W_F \otimes I_A),\quad
	\begin{bmatrix}
	P_{FA} & (Q_{FA\to A})^{\dag}\\
	Q_{FA\to A} & R_A
	\end{bmatrix}
	\geq0,
	\end{equation}
	where $P_{FA}$ and $R_A$ are Hermitian, and $Q_{FA\to A}$ is a linear operator.
\end{proposition}

\begin{proof}
	We begin with the formula
	\begin{align}
	\widehat{I}(\bm{\theta},W;\{\rho^{\bm{\theta}}_A\}_{\bm{\theta}}) &
	:=\operatorname{Tr}\!\left[  (W_F \otimes I_{A})\left(  \sum_{j,k=1}^{D}%
	|j\rangle\!\langle k|_F \otimes(\partial_{\theta_{j}}\rho_A^{\bm{\theta}}%
	)(\rho_A^{\bm{\theta}})^{-1}(\partial_{\theta_{k}}\rho_A^{\bm{\theta}})\right)
	\right]  \\
	&  =\sum_{j,k=1}^{D}\langle k|W|j\rangle\operatorname{Tr}\!\left[
	(\partial_{\theta_{j}}\rho_A^{\bm{\theta}})(\rho_A^{\bm{\theta}})^{-1}%
	(\partial_{\theta_{k}}\rho_A^{\bm{\theta}})\right]  .
	\end{align}
	The above can be written as
	\begin{equation}
	\Tr \left[ (W_F \otimes I_{A}) (X^{\dag} Y^{-1} X) \right]
	\end{equation}
	where $X = \sum_{j =1}^D \bra{j}_F \otimes (\partial_{\theta_j} \rho^{{\bm{\theta}}}_{A} )$ and $Y = \rho_A^{\bm{\theta}}$.
	
	We combine the above with the Schur complement lemma (Lemma \ref{lem:min-XYinvX} stated in Chapter~\ref{ch:single}), to obtain the desired primal form in \eqref{eq:primal-rld-value-states}.
	
	To obtain the dual program, we apply Lemma~\ref{lem:freq-used-SDP-primal-dual} from Chapter~\ref{ch:single}. 
\end{proof}

\subsection{Semi-definite program for RLD\ Fisher information value of quantum channels}

\begin{proposition} \label{prop:rld-channels-sdp}
	Let $\{\mathcal{N}_{A\rightarrow B}^{\bm{\theta}}\}_{\bm{\theta}}$ be a differentiable family of quantum channels, and let $W$ be a $D\times D$ weight matrix. Suppose that the finiteness condition for the RLD Fisher information value \eqref{eq:app:channel-RLD-finiteness-cond} holds. Then the RLD\ Fisher information value of quantum channels can be calculated via the following semi-definite program:
	\begin{equation} \label{eq:primal-rld-value-channels}
	\widehat{I}_{F}(\bm{\theta},W;\{\mathcal{N}_{A\rightarrow B}^{\bm{\theta}}%
	\}_{\bm{\theta}})=\inf\lambda\in\mathbb{R}^{+}, 
	\end{equation}
	subject to
	\begin{equation} \label{eq:primal-rld-value-channels-2}
	\lambda I_{R}\geq\operatorname{Tr}_{FB}[(W_F \otimes I_{RB} ) M_{FRB}], \qquad
	\begin{bmatrix}
	M_{FRB} & \sum_j | j \rangle_F \otimes  (\partial_{\theta_j} \Gamma^{\mathcal{N}^{\bm{\theta}}}_{RB} )    \\
	\sum_j  \langle j |_F \otimes (\partial_{\theta_j} \Gamma^{\mathcal{N}^{\bm{\theta}}}_{RB} )  & \Gamma^{\mathcal{N}^{\bm{\theta}}}_{RB}
	\end{bmatrix}
	\geq0. 
	\end{equation}
	
	The dual program is given by
	\begin{equation}
	\sup_{\rho_{R}\geq0,P_{FRB},Z_{FRB\to RB},Q_{RB}} 2 \left( \sum_{j=1}^D \operatorname{Re}[\operatorname{Tr}
	[Z_{FRB\to RB} ( | j \rangle_F \otimes  (\partial_{\theta_j} \Gamma^{\mathcal{N}^{\bm{\theta}}}_{RB} )) ]] \right) -\operatorname{Tr}[Q_{RB}\Gamma_{RB}^{\mathcal{N}^{\bm{\theta}}}],
	\end{equation}
	subject to
	\begin{equation}
	\operatorname{Tr}[\rho_{R}]\leq1,\quad
	\begin{bmatrix}
	P_{FRB} & (Z_{FRB\to RB})^{\dag}\\
	Z_{FRB\to RB} & Q_{RB}
	\end{bmatrix}
	\geq0,\quad P_{FRB}\leq\rho_{R} \otimes W_F \otimes I_{B}.
	\end{equation}
		
\end{proposition}

\begin{proof}
	The form of the primal program relies on the combination of a few facts. First, we use the following characterization of the infinity norm of a positive semi-definite operator $A$:
	\begin{equation}
	\left\Vert A\right\Vert _{\infty}=\inf\left\{  \lambda\geq0:A\leq\lambda I\right\}  .
	\end{equation}	
	
	Next we observe that 
	\begin{equation}
	\sum_{j,k=1}^{D}\langle k|W|j\rangle \operatorname{Tr}_{B}[(\partial_{\theta_{j}}\Gamma_{RB}^{\mathcal{N}^{\bm{\theta}}})(\Gamma_{RB}^{\mathcal{N}^{\bm{\theta}}})^{-1}(\partial_{\theta_{k}}\Gamma_{RB}^{\mathcal{N}^{\bm{\theta}}})]
	\end{equation}
	can be written as
	\begin{equation}
	\Tr_{FB} \!\left[ (W_F \otimes I_{RB}) (X^{\dag} Y^{-1} X) \right]
	\end{equation}
	where $X = \sum_{j =1}^D \bra{j}_F \otimes (\partial_{\theta_j} \Gamma^{\mathcal{N}^{\bm{\theta}}}_{RB} )$ and $Y = \Gamma^{\mathcal{N}^{\bm{\theta}}}_{RB}$.
	
	We next use the Schur complement lemma (Lemma \ref{lem:min-XYinvX} stated in Chapter~\ref{ch:single}) and combine the above with the explicit form of the RLD Fisher information value of quantum channels \eqref{eq:RLD-value-channels-inf-norm} to obtain the desired primal form in \eqref{eq:primal-rld-value-channels}.
	
	To arrive at the dual program, we use the standard forms of primal and dual
	semi-definite programs for Hermitian operators $A$ and $B$ and a
	Hermiticity-preserving map $\Phi$ \cite{Watrous2018}:
	\begin{equation}
	\sup_{X\geq0}\left\{  \operatorname{Tr}[AX]:\Phi(X)\leq B\right\}  ,
	\qquad\inf_{Y\geq0}\left\{  \operatorname{Tr}[BY]:\Phi^{\dag}(Y)\geq
	A\right\}  . \label{eq:standard-SDP-form-RLD-ch-helper-channels}
	\end{equation}
	From \eqref{eq:primal-rld-value-channels}--\eqref{eq:primal-rld-value-channels-2}, we
	identify
	\begin{align}
		B  &  =
		\begin{bmatrix}
			1 & 0\\
			0 & 0
		\end{bmatrix},
		\quad Y=
		\begin{bmatrix}
			\lambda & 0\\
			0 & M_{FRB}
		\end{bmatrix}
		,\quad\Phi^{\dag}(Y)=
		\begin{bmatrix}
			\lambda I_{R}-\operatorname{Tr}_{FB}[(W_F \otimes I_{RB})M_{FRB}] & 0 & 0\\
			0 & M_{FRB} & 0\\
			0 & 0 & 0
		\end{bmatrix}
		,\\
		A  &  =
		\begin{bmatrix}
			0 & 0 & 0\\
			0 & 0 & -\sum_{j=1}^D | j \rangle_F \otimes  (\partial_{\theta_j} \Gamma^{\mathcal{N}^{\bm{\theta}}}_{RB} )\\
			0 & -\sum_{j=1}^{D} \bra{j}_F \otimes  (\partial_{\theta_j} \Gamma^{\mathcal{N}^{\bm{\theta}}}_{RB} ) & -\Gamma^{\mathcal{N}^{\bm{\theta}}}_{RB}
		\end{bmatrix}
		.	
	\end{align}

	Upon setting
	\begin{equation}
	X=
	\begin{bmatrix}
	\rho_{R} & 0 & 0\\
	0 & P_{FRB} & (Z_{FRB\to RB})^{\dag}\\
	0 & Z_{FRB\to RB} & Q_{RB}
	\end{bmatrix}
	,
	\end{equation}
	we find that
	\begin{align}
		& \operatorname{Tr}[X\Phi^{\dag}(Y)]  \notag \\
		&  =\operatorname{Tr}\!\left[
		\begin{bmatrix}
			\rho_{R} & 0 & 0\\
			0 & P_{FRB} & (Z_{FRB\to RB})^{\dag}\\
			0 & Z_{FRB\to RB} & Q_{RB}
		\end{bmatrix}
		\begin{bmatrix}
			\lambda I_{R}-\operatorname{Tr}_{FB}[(W_F \otimes I_{RB})M_{FRB}] & 0 & 0\\
			0 & M_{FRB} & 0\\
			0 & 0 & 0
		\end{bmatrix}
		\right] \\
		&  =\operatorname{Tr}[\rho_{R}(\lambda I_{R}-\operatorname{Tr}_{FB}
		[(W_F \otimes I_{RB})M_{FRB}])]+\operatorname{Tr}[P_{FRB}M_{FRB}]\\
		&  =\lambda\operatorname{Tr}[\rho_{R}]+\operatorname{Tr}[(P_{FRB}-\rho
		_{R}\otimes W_F \otimes I_{B})M_{FRB}]\\
		&  =\operatorname{Tr}\!\left[
		\begin{bmatrix}
			\lambda & 0\\
			0 & M_{FRB}
		\end{bmatrix}
		\begin{bmatrix}
			\operatorname{Tr}[\rho_{R}] & 0\\
			0 & P_{FRB}-\rho_{R} \otimes W_F \otimes I_{B}
		\end{bmatrix}
		\right]  \\
		&= \Tr[Y \Phi(X)]
		,
	\end{align}
	to find that the dual is given by
	\begin{equation}
	\sup_{\substack{\rho_{R},P_{FRB},\\Z_{FRB\to RB},Q_{RB}}}\operatorname{Tr}\left[AX
	\right]  ,
	\end{equation}
	subject to
	\begin{equation}
	\begin{bmatrix}
	\rho_{R} & 0 & 0\\
	0 & P_{FRB} & (Z_{FRB\to RB})^{\dag}\\
	0 & Z_{FRB\to RB} & Q_{RB}
	\end{bmatrix}
	\geq0,\qquad
	\begin{bmatrix}
	\operatorname{Tr}[\rho_{R}] & 0\\
	0 & P_{FRB}-\rho_{R} \otimes W_F \otimes I_{B}
	\end{bmatrix}
	\leq
	\begin{bmatrix}
	1 & 0\\
	0 & 0
	\end{bmatrix}
	.
	\end{equation}
	We can swap $Z_{FRB\to RB} \rightarrow - Z_{FRB\to RB}$ with no change to the optimal value. This leads to the following simplified form of the dual program:
	\begin{equation}
	\sup_{\rho_{R}\geq0,P_{FRB},Z_{FRB},Q_{RB}}2 \sum_{j=1}^D \operatorname{Re}[\operatorname{Tr}
	[Z_{FRB\to RB} ( | j \rangle_F \otimes  \partial_{\theta_j} \Gamma^{\mathcal{N}^{\bm{\theta}}}_{RB}  ) ]]-\operatorname{Tr}[Q_{RB}\Gamma_{RB}^{\mathcal{N}^{\bm{\theta}}}],
	\end{equation}
	subject to
	\begin{equation}
	\operatorname{Tr}[\rho_{R}]\leq1,\quad
	\begin{bmatrix}
	P_{FRB} & -(Z_{FRB\to RB})^{\dag}\\
	-Z_{FRB\to RB} & Q_{RB}
	\end{bmatrix}
	\geq0,\quad P_{FRB}\leq\rho_{R} \otimes W_F \otimes I_{B}.
	\end{equation}
	Then we note that
	\begin{equation}
	\begin{bmatrix}
	P_{FRB} & -(Z_{FRB\to RB})^{\dag}\\
	-Z_{FRB\to RB} & Q_{RB}
	\end{bmatrix}
	\geq0 \quad\Longleftrightarrow\quad
	\begin{bmatrix}
	P_{FRB} & (Z_{FRB\to RB})^{\dag}\\
	Z_{FRB\to RB} & Q_{RB}
	\end{bmatrix}
	\geq0
	\end{equation}
	This concludes the proof.
\end{proof}

\pagebreak

\chapter{Conclusion and Open Questions} \label{ch:conclusion}
\allowdisplaybreaks

\vspace{0.5em}

In this dissertation, we studied the problem of estimating one or more parameters encoded in an unknown quantum state or channel. In particular, we proved a number of novel Cramer--Rao bounds (lower bounds on the estimation error of the unknown parameter(s)) for quantum channel estimation. Our bounds are universal, in that they apply to all quantum channels and thus encompass all possible quantum dynamics. Further, our bounds hold in the most general case of sequential strategies, where we are given $n$ copies of the unknown channel and process them so that they occur one after the other.

In particular, for both single and multiparameter estimation of quantum channels, we established the concept of amortized Fisher information and showed how to connect the amortized Fisher information of a family of channels to the corresponding Fisher information achieved by a sequential estimation protocol using meta-converse theorems in Chapters~\ref{ch:single} and \ref{ch:multi}.

We then showed that amortization collapses occur for the following quantities in Chapter~\ref{ch:single} for the single parameter case:
\begin{itemize}
	\item the SLD Fisher information of classical-quantum channels,
	\item the root SLD Fisher information of general quantum channels, and
	\item the RLD Fisher information of general quanum channels,
\end{itemize}
and in Chapter~\ref{ch:multi} for the multiparameter case, we showed an amortization collapse for
\begin{itemize}
	\item the RLD Fisher information value of general quantum channels.
\end{itemize}

Combining the meta-converse theorems and specific amortization collapses leads to single-letter Cramer--Rao bounds for channel estimation in the sequential setting. That is, in Chapter~\ref{ch:single}, we establish the following Cramer--Rao bounds:
\begin{itemize}
	\item a Cramer–Rao bound for estimation of classical-quantum channels using the SLD Fisher information,
	\item a Cramer–Rao bound for estimation of general quantum channels using the SLD Fisher information, and
	\item a Cramer–Rao bound for estimation of general quantum channels using the RLD Fisher information,
\end{itemize}
and in Chapter~\ref{ch:multi},
\begin{itemize}
	\item a scalar Cramer--Rao bound for estimation of quantum states using the RLD Fisher information value, and
	\item a scalar Cramer--Rao bound for estimation of quantum channels using the RLD Fisher information value.
\end{itemize}

Our bounds for channel estimation all have the desirable characteristic that they are single-letter; i.e., they are applicable for $n$-round sequential strategies yet only require the Fisher information in question to be evaluated for a single channel use. Further, we provide various optimization problem formulations to compute our bounds in Chapters~\ref{ch:single} and \ref{ch:multi}. Thus we believe that these bounds will be useful as a theoretical tool offering a new perspective on the problem, and as a framwork to analyze sequential estimation strategies. Further, the universality and numerical accessibility of our bounds make them amenable to direct experimental application.

Our RLD-based bounds for both single and multiparameter estimation of channels have important implications for the attainability of Heisenberg scaling with respect to the number of channel uses in a sequential protocol. That is, our RLD bound in Chapter~\ref{ch:single} implies that if the RLD Fisher information of a channel family is finite, then Heisenberg scaling is unattainable. Further, our multiparameter RLD bound in Chapter~\ref{ch:multi} implies that if the RLD Fisher information value of a channel family is finite, then Heisenberg scaling is unattainable. We note that another recent sufficient condition for the non-attainability of Heisenberg scaling is the HNLS condition of Ref.~\cite{Zhou2018}. Further, our results complement other recent work in this area~\cite{Fujiwara2008, Matsumoto2010, Hayashi2011,DemkowiczDobrzanski2017}.

Finally, we evaluate our RLD-based Cramer--Rao bounds for the physically motivated example of estimating the parameters of a generalized amplitude damping channel. In Chapter~\ref{ch:single}, we evaluate the RLD Cramer--Rao bound for individually estimating the parameters of the channel, and in Chapter~\ref{ch:multi}, we evalute the RLD Fisher information value of the unknown channel, considering the task of estimating its two parameters simultaneously.

A topic for future research is to incorporate energy constraints into the Fisher information of quantum states and channels. We defined the generalized Fisher information in Chapter~\ref{ch:prelims} without any energy constraints. However, it is of practical interest to study the power of estimation strategies when the input probe states are subject to an energy constraint. That is, we can define the energy-constrained generalized Fisher information of a quantum channel family as follows:
\begin{equation}
    \mathbf{I}_{F, E} ( \theta; \{ \mathcal{N}^{\theta}_{A \rightarrow B} \}_{\theta} ) = \sup_{\rho_{RA}: \mathrm{Tr}[ H_A \rho_A ] \leq E} \mathbf{I}_{F} ( \theta; \{ \mathcal{N}^{\theta}_{A \rightarrow B} (\rho_{RA}) \}_{\theta} )
\end{equation} 
where $H_A$ is a Hamiltonian acting on the input system of the channel $\mathcal{N}_{A \rightarrow B}^{\theta}$. This definition generalizes the energy-constrained channel divergence introduced in \cite{Sharma2018}. Further, the energy-constrained amortized Fisher information can also be defined using the quantity defined above. The study of sequential estimation strategies in this thesis can then be generalized to incorporate the aforementioned energy constraint.

We showed that with respect to the RLD Fisher information of quantum channels, sequential estimation strategies offer no advantage over parallel ones. However, the RLD-based Cramer--Rao bound is a loose one in general, and therefore it is still an interesting open question to determine whether sequential estimation offers a benefit over parallel estimation. For the particular case of classical-quantum channels, we showed in Chapter~\ref{ch:single} that the optimum estimation strategy is a parallel one. However, we are unable as yet to make a similar statement about all quantum channels in general and leave it for future work to do so. This topic has been studied recently in Ref.~\cite{Zhou2020}.

Another open question is to determine the operational interpretation of the RLD Fisher information of channels. One possibility is to generalize the work on ``reverse estimation'' of Ref.~\cite{Matsumoto2005} from quantum states to channels. From Ref.~\cite{Matsumoto2005}, we know that the RLD Fisher information of states is the optimal classical Fisher information needed to simulate the state family in a local way. It is as yet open to show the analogous result for quantum channels. This is related to the task of coherence distillation of quantum channels~\cite{Marvian2020}.

\pagebreak

\addtocontents{toc}{\vspace{12pt}}
\addcontentsline{toc}{chapter}{\hspace{-1.6em} REFERENCES}

\singlespacing
\bibliographystyle{alpha}
\bibliography{combined_refs_vishal_thesis}

\pagebreak

\chapter*{\Large Vita}
\setlength{\parindent}{1.75em}
\vspace{0.2em}
\addtocontents{toc}{\vspace{12pt}}
\addcontentsline{toc}{chapter}{\hspace{-1.5em} VITA}

Vishal Katariya was born in Chennai, India in 1995. He has lived in, and did his schooling in, Hyderabad, India, Zurich, Switzerland, and Chennai, India. In 2013, he joined the Indian Institute of Technology Madras to pursue a B.Tech. degree in Engineering Physics. Upon completion, he immediately joined a PhD program in 2017 at Louisiana State University, Baton Rouge. He worked from the outset in the field of quantum information with Mark M. Wilde, as part of the Quantum Science and Technologies Group. He plans to graduate in May 2022.

\end{document}